\documentclass[authoryear, 11pt]{article}

\RequirePackage{amsthm,amsmath,amsfonts,amssymb}
\RequirePackage{natbib}
\usepackage[a4paper, total={6in, 8in}]{geometry}
\RequirePackage[colorlinks,citecolor=blue,urlcolor=blue]{hyperref}
\RequirePackage{graphicx}
\usepackage{mathtools}
\usepackage{subcaption}
\usepackage{titlesec}
\usepackage{makecell}
\usepackage{booktabs}  % For better-looking tables
\usepackage{array}     % For advanced column formatting
\usepackage{float}
\usepackage{lmodern}

\usepackage{algorithm}
\usepackage{algpseudocode}

\usepackage{float}
\usepackage{xcolor}

\theoremstyle{plain}
\newtheorem{Proposition}{Proposition}
\newtheorem{Lemma}{Lemma}

\newtheorem{Theorem}{Theorem}

\theoremstyle{remark}
\newtheorem{Assumption}{Assumption}
\newtheorem{Definition}{Definition}
\newtheorem{Setting}{Setting}

\usepackage{newtxmath} %for indicator function
\usepackage{enumerate}
\usepackage{chapterbib}

\newcommand{\R}{\mathbb{R}}
\newcommand{\N}{\mathbb{N}}

\newcommand{\E}{\mathbb{E}}
   
\newcommand{\Cov}{\operatorname{Cov}}

\newcommand{\EINS}{\vmathbb{1}}           

\newcommand{\Pj}{\mathbb{P}}

\newcommand\argmin{\operatornamewithlimits{argmin}}

\newcommand{\Clambda}{C_1}
\newcommand{\CtrendSeasonalBound}{C_2}
\newcommand{\CfinalEstimate}{C_3}
\newcommand{\CvarianceEstimateLower}{C_4}
\newcommand{\CvarianceEstimate}{C_5}
\newcommand{\CremainderSeasonal}{C_6}
\newcommand{\CfinalTrendEstimate}{C_7}
\newcommand{\CerrorTrend}{C_8}
\newcommand{\Cdeterminant}{C_9}
\newcommand{\ClengthL}{C_{10}}
\newcommand{\Ctrace}{C_{11}}
\newcommand{\CcofactorXu}{C_{12}}
\newcommand{\CnoiseLin}{C_{13}}
\newcommand{\Cr}{C_{14}}
\newcommand{\CcostEst}{C_{15}}
\newcommand{\CcostExtraCp}{C_{16}}
\newcommand{\Cmiss}{C_{17}}
\newcommand{\CextraCP}{C_{18}}
\newcommand{\CremainderTrend}{C_{19}}
\newcommand{\CfinalSeasonalEstimate}{C_{20}}
\newcommand{\CsizeOfVzero}{C_{21}}
\newcommand{\CtildeQ}{C_{22}}
\newcommand{\CprobBoundTildeq}{C_{23}}
\newcommand{\CsuccessProbability}{C_{24}}
\newcommand{\CfinalAnomalyEstimate}{C_{25}}

\newcommand{\mn}{{m_n}}
\newcommand{\rQ}{{r_Q}}
\newcommand{\Ktau}{{K_{\tau}}}

\usepackage{hyperref}
\hypersetup{
    colorlinks=true,
    linkcolor=blue,
    citecolor = blue,
    filecolor=magenta,      
    urlcolor=cyan,
    pdftitle={Overleaf Example},
    pdfpagemode=FullScreen,
    }

\begin{document}

\title{Detection of collective and point anomalies in the presence of trend and seasonality}

\author{Yiyin Zhang, Florian Pein, Idris A.\ Eckley}

\maketitle

\begin{abstract}
Detecting anomalies in time series data is a challenging task with broad relevance in many applications. Existing methods work effectively only under idealized conditions, typically focusing on point anomalies or assuming a constant baseline. Our approach overcomes these limitations by detecting both collective and point anomalies, while allowing for polynomial trends and seasonal patterns. We establish statistical theory demonstrating that our method accurately decomposes the time series into anomaly, trend, seasonality, and remainder components. We further show that the approach provides a consistent estimate of the number of anomalies and their locations. Simulation studies confirm its strong detection performance with finite samples, and an application to energy price data illustrates its practical utility.
\end{abstract}

\section{Introduction}\label{section:introduction}
        
%The ability to accurately and efficiently estimate the location  of anomalies in time series data has become increasingly important in various domains, including energy \citep{serrano2020time, zhang2021time}, network monitoring \citep{wu2021graph} and economics \citep{zamanzadeh2024deep}.

The challenge of modelling in the presence of anomalies has been an area of active research for many years. 
%or identifying theHistorically, there has been extensive research on anomalies within both the statistical and machine learning literature. 
For example, a rich literature exists describing inference methods that can cope with the effect of outlying observations. Notable contributions include the foundational work of \citet{huber1981robust} and \citet{rousseeuw1984robust} who introduced M- and S-estimators to provide robust estimates of the mean and variance, effectively mitigating the impact of outliers. In addition, several  time series approaches have been developed to provide methods which are robust to anomalous observations. See, for example, work developing robust modelling for ARMA \citep{muler2009robust}, ARCH \citep{muler2002robust}, and GARCH \citep{muler2008robust} settings. A recent overview is also provided by \citet{maronna2019robust}. This article addresses a different challenge, namely the identification of point and collective anomalies in the presence of trend and seasonality. 

%In parallel, the machine learning community has contributed a diverse set of methods developed for anomaly detection in time series data. These range from tree-based approaches such as Isolation Forests \citep{liu2008isolation} to deep learning-based techniques. For instance, \citet{sakurada2014anomaly} apply auto-encoders to detect anomalies by reconstructing input data and measuring the reconstruction error. We refer the interested reader to reviews by \citet{chandola2009anomaly, chalapathy2019deep, zhao2023anomaly} and references therein.

Traditionally, anomaly detection focuses on identifying individual data points or entire sequences that deviate from expected patterns. More recent research, however, has shifted toward integrating changepoint detection with anomaly detection to capture collective anomalies -- groups of consecutive data points that appear anomalous only when considered together. A prominent example is the work of \citet{fisch2020real}, who introduced the CAPA (Collective And Point Anomalies) algorithm, capable of detecting both collective and point anomalies effectively. It operates under the assumption \eqref{eq:capa_assumption} that the data sequence follows a constant baseline with Gaussian distributed noise, treating any deviations as anomalies:
\begin{equation}
\label{eq:capa_assumption}
    Y = A + B + \varepsilon,
\end{equation}
where $ Y = (y_1, \ldots, y_n) $ is the observed time series, $ A = (a_1, \ldots, a_n) $ represents point and collective anomalies, $ B $ is the constant baseline, and $ \varepsilon = (\varepsilon_1,\ldots,\varepsilon_n) $ is the noise vector. By employing a penalized cost method, CAPA efficiently identifies both the number and locations of point and collective anomalies, balancing model complexity with accuracy. However, real-world time series data often exhibit more complicated structures, including trends, seasonal patterns, and noise. Recognizing these patterns is vital for accurate analysis across various domains such as finance \cite{liu2008isolation, kumar2006financial}, energy \cite{hong2016energy, amjady2010short}, and healthcare \cite{clifford2006advanced, liu2015early}. An example of such data is displayed in Figure~\ref{fig:data_oracle}, consisting of a seasonal component and a quadratic trend:
\begin{equation}
    \label{eq:model_assumption}
    Y = A + S + T + \varepsilon.
\end{equation}
The observed data $Y$, anomaly component $A$, and noise $\varepsilon$ follow the structure in \eqref{eq:capa_assumption}. Instead of a constant trend, we use a polynomial trend $T=(t_1,\ldots,t_n)$ and a periodic seasonal component $S=(s_1,\ldots,s_n)$; the full model appears in Section~\ref{sec:model}. Intuitively, $A$ captures atypical deviations, $T$ the long-term deterministic evolution, $S$ recurring fluctuations, and $\varepsilon$ the remaining random variation.

\begin{figure}[htbp]
    \centering
    \includegraphics[width=0.6\textwidth]{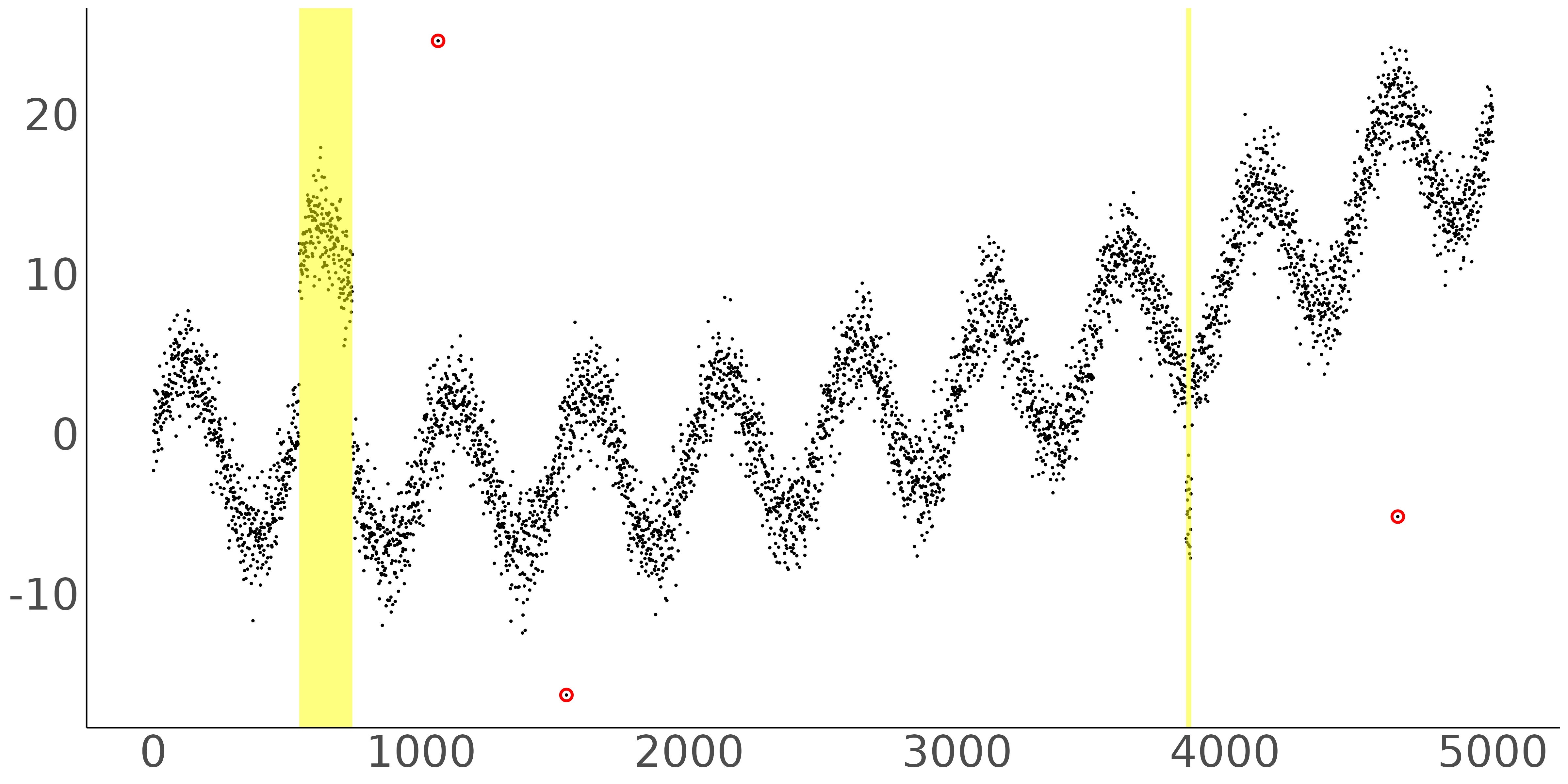}
    \caption{A simulated example illustrating a time series that includes a quadratic trend and a sinusoidal seasonal component, as well as both collective and point anomalies. The yellow-shaded regions highlight the collective anomalies, while red-circled points represent the point anomalies.}
    \label{fig:data_oracle}
\end{figure}

Distinguishing between genuine anomalies and typical variations within such data pre\-sents a significant challenge, particularly when trends and seasonal effects are pronounced. In such situations, CAPA and other variants can face limitations. For example, the presence of strong trend and seasonal components can obscure or mask anomalies, leading to false positives, as seen in Figure \ref{fig:capa}.

\begin{figure}[htbp]
    \centering
    \includegraphics[width=0.6\textwidth]{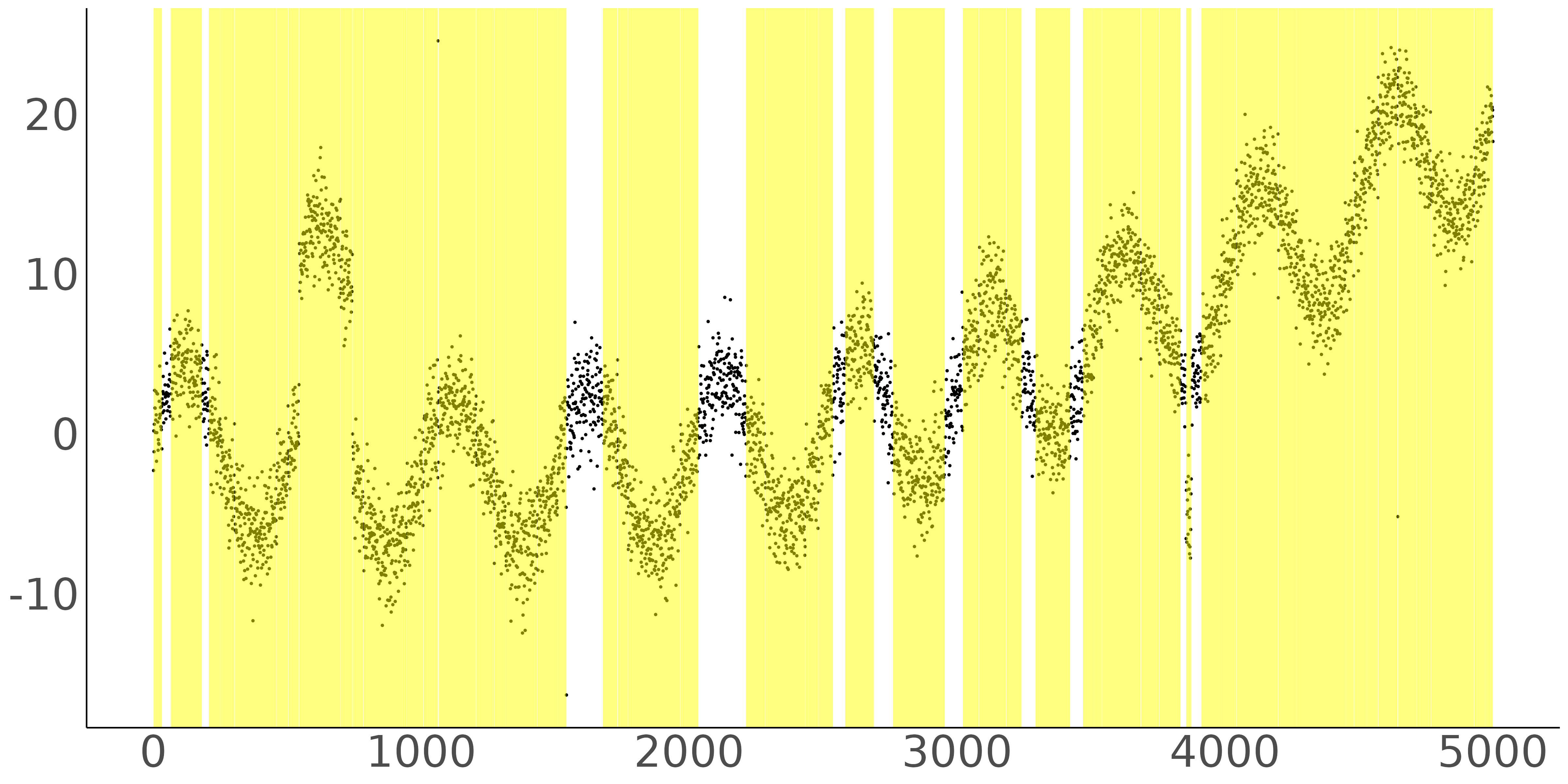}
    \caption{\small %As an illustration of its anomaly detection capabilities, 
    CAPA applied to the example shown in Figure~\ref{fig:data_oracle}, with the yellow areas highlighting the collective anomalies estimated by CAPA. %The significant number of false positives indicates that CAPA's performance is inadequate when dealing with data that exhibits strong trend and seasonal components.
    }
    \label{fig:capa}
\end{figure}

The challenge of decomposing a time series \emph{without} anomalies into different components has been well studied in the literature, for an overview see \citep[Chapter 6]{hyndman2018forecasting} and \citep{dama2021time, DasBarman2025}. %Several decomposition methods have been introduced to separate time series data into trend, seasonal, and residual components. 
Established methods include Seasonal and Trend decomposition using Loess (STL) \citep{cleveland1990stl}, its version robust to outliers RobustSTL \citep{wen2019robuststl}, and Prophet \citep{taylor2018forecasting}. Such methods have been designed for a model without (collective) anomalies, i.e.~model \eqref{eq:model_assumption} without the anomaly component $A$, or it being limited to point anomalies. Conversely, these approaches are not limited to polynomial trends, but allow a smoothly changing trend. \citet{wen2019robuststl} demonstrate that such a decomposition enhances the precision of anomaly detection by estimating trend and seasonal components and then focusing on the remainder to detect any anomalies. This usually works well for point anomalies, but struggles when time series contain collective anomalies. This is illustrated in Figure \ref{fig:stl_decompose}, where we apply STL to the time series shown in Figure \ref{fig:data_oracle} which contains a collective anomaly, trend, and seasonality. Note how the collective anomaly structure is incorporated into the estimated trend and seasonal components, thus hindering the detection of anomalies. RobustSTL is specifically designed to be robust to abrupt changes while providing a decomposition into trend, seasonality, and remainder. However, practically it suffers from being a computationally intensive procedure. Due to its comparatively slow runtime, we did not apply it to the example shown in Figure \ref{fig:data_oracle} with $5\,000$ observations, but only to settings with fewer data points. In these experiments, we observed that the method tends to mix parts of the trend component into the seasonal component and occasionally introduces spurious point anomalies. Nevertheless, it produces trend and remainder estimates that are (nearly) free of seasonality.
%We remark that RobustSTL was too slow to be applied to this data example with $5\,000$ observations, but examples with fewer observations demonstrated that its decomposition has similar issues.

\begin{figure}[htbp]
    \centering
\includegraphics[width=0.9\textwidth]{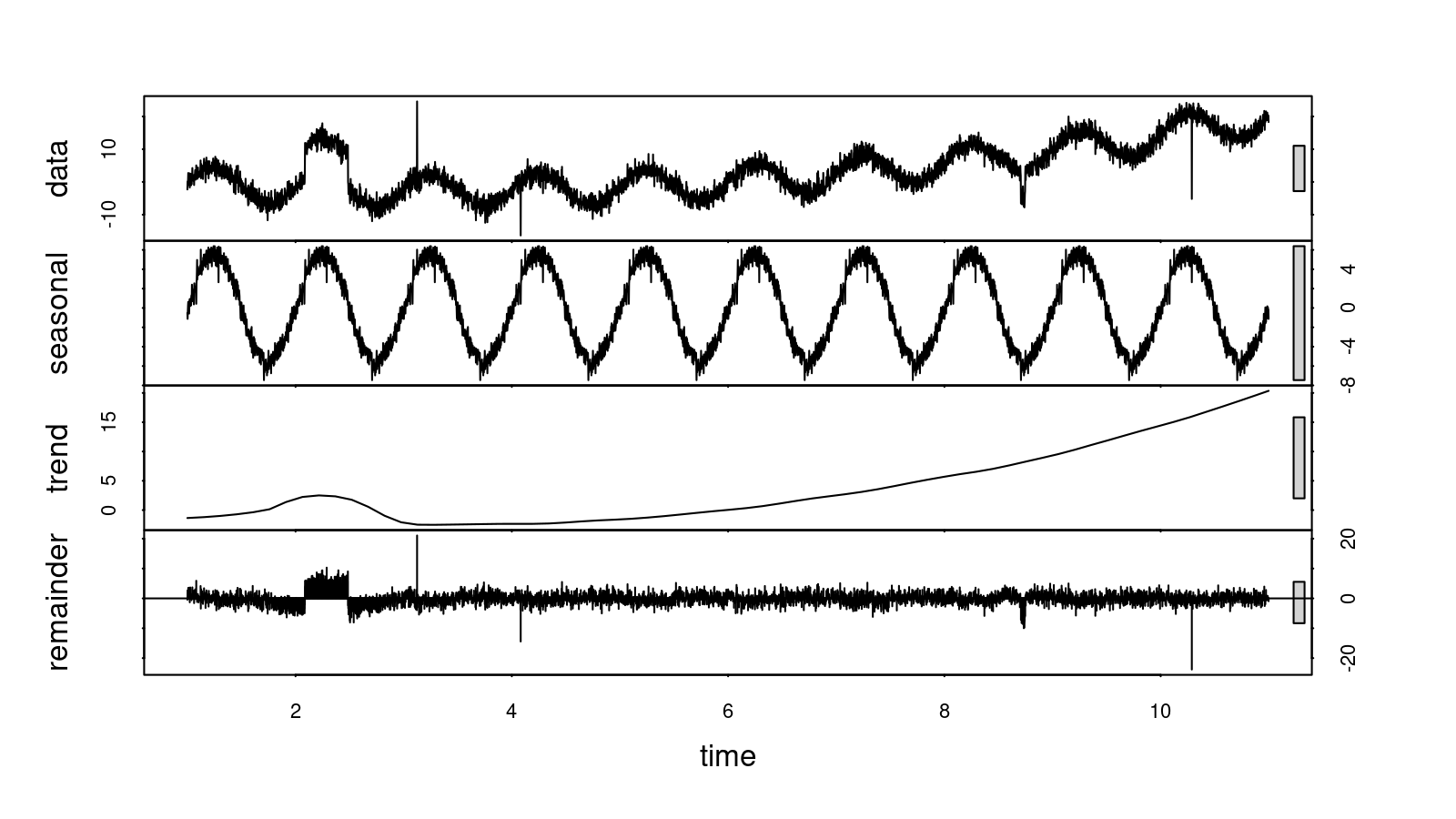}
    \caption{\small Decomposition of the simulated example illustrated in Figure \ref{fig:data_oracle} using STL. It consists of four subplots: the original data; the decomposed seasonal component; the decomposed trend; and the remainder. The collective anomalies have affected the decomposition, showcasing noticeable variations in both the seasonal and trend components.}
    \label{fig:stl_decompose}
\end{figure}
Consequently, current methods are not well-suited for handling such complex anomaly structures, highlighting the need for more robust approaches. To overcome the limitations of existing methods, we propose \textbf{S}easonal \textbf{T}rend \textbf{A}nomaly \textbf{D}etection, (\textbf{STAD}), a novel  approach extending existing anomaly detection methods by allowing a polynomial trend and seasonality as part of the typical pattern. At the same time, STAD decomposes the data into four components: anomaly, trend, seasonality, and remainder. %By convention the baseline is included in the trend component. As discussed before none of the components can be estimated reliably by existing robust estimators when the other three components are present.

Our proposed method consists of a few steps. Firstly, it uses differenced data, where the lag-difference is a multiple of the seasonal length. This removes the seasonal component and allows to estimate the trend in the next step. Standard robust methods can effectively estimate the polynomial trend at presence of point anomalies but struggle with collective anomalies, since they are non-symmetric and contain potentially a large number of anomalous points. We aim to overcome this challenge by finding a subsample of the data that is free of collective anomalies and large enough to estimate the trend reliably. To this end, we propose a novel systematic sampling strategy that gives us a collection of subsamples. We subtract the estimated trend and obtain a robust estimate of the seasonality. Similarly, once the seasonality is estimated, the trend can be estimated without that differencing is required. To improve accuracy we make us of it and re-estimate the trend and seasonal component once before estimating the anomalies. Section~\ref{section:method} gives full details of the proposed methodology.

In Section~\ref{section:theory} we establish statistical theory to confirm the STAD's detection properties. We show that it successfully decomposes the time series by providing asymptotic L2 bounds for the error of its estimation of the anomaly, trend and seasonal component. Furthermore, we show consistent estimation of the number of anomalies and that estimation of start and end points have a small asymptotic error. These results substantially build on the results established by \citet{fisch2020real} for CAPA.

The finite sample performance of STAD is studied via simulation studies in Section~\ref{section:simulation}. %STAD estimates anomalies almost as good as its oracle version that knows the trend and seasonal component unless anomalies are extremely long (more than a quarter of the total time series) or the seasonal component has very few repetitions (less than five). 
In Section~\ref{section:applications} we use STAD to detect anomalies in the price of energy data. A brief illustration of the effectiveness of our method is given in Figures~\ref{fig:STAD_decompose}~and~\ref{fig:STAD_capa}, in which we have applied STAD to the data shown Figure~\ref{fig:data_oracle}.

\begin{figure}[hbtp]
    \centering   \includegraphics[width=0.6\textwidth]{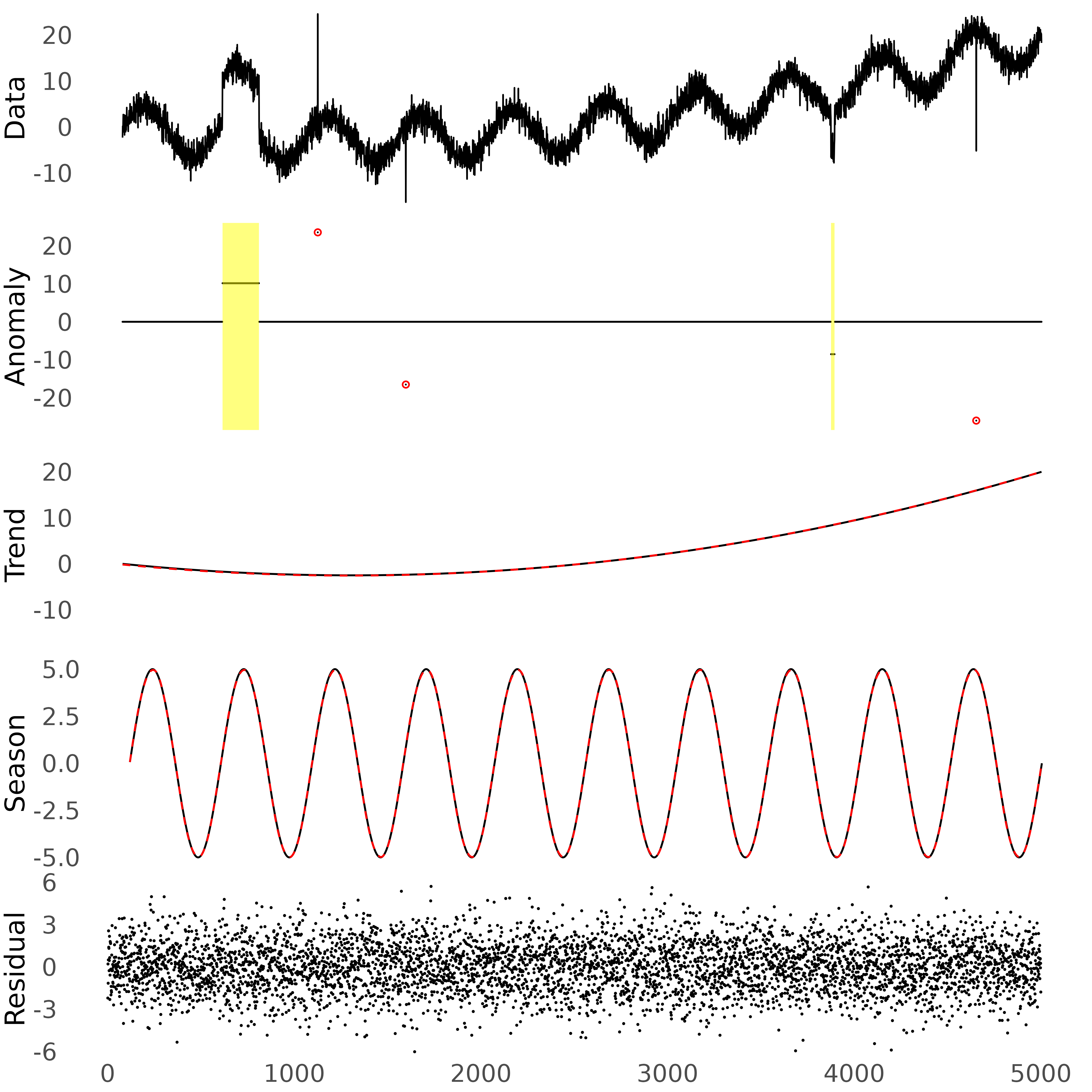}
    \caption{\small The decomposition and anomaly detection results using STAD. It contains five subplots: the original data from Figure \ref{fig:data_oracle}; the estimated anomaly component, with yellow areas for collective anomalies and red circles for point anomalies; the trend and seasonal component (the solid black lines are the truth and the dashed red lines are the estimates); and the estimated remainder. Different to CAPA and STL, confer Figures~\ref{fig:capa}~and~\ref{fig:stl_decompose}, it does not include false positives and decomposes the data successfully into anomaly, trend, seasonal, and residual components. }
    \label{fig:STAD_decompose}
\end{figure}

\begin{figure}[hbtp]
    \centering
    \includegraphics[width=0.6\textwidth]{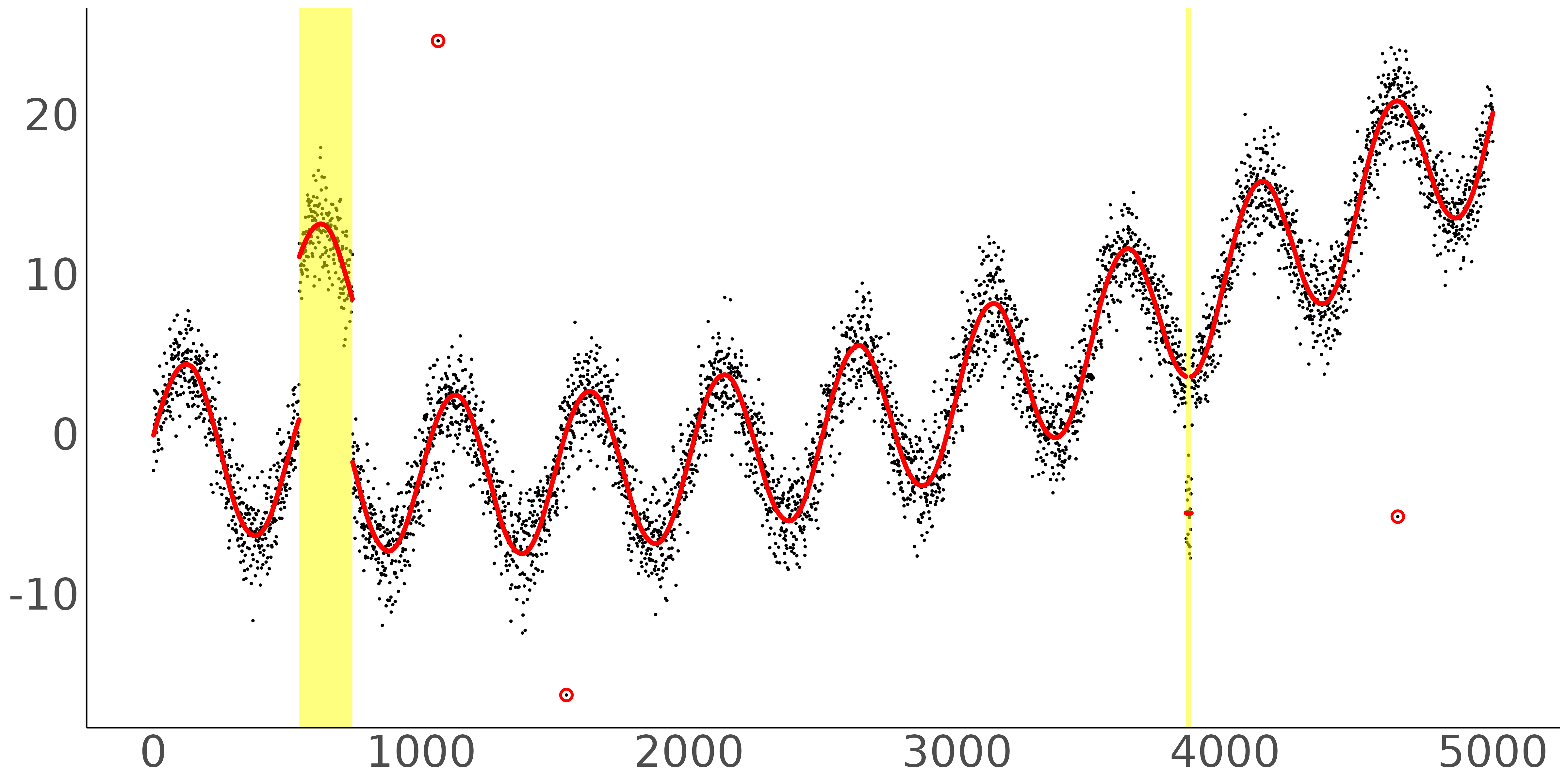}
    \caption{\small The anomaly detection results using STAD on the example in Figure \ref{fig:data_oracle}. The yellow areas indicate the estimated collective anomalies, the red circles mark the detected point anomalies, and the red solid line represents the estimated signal, including anomaly, trend, and seasonal components.}
    \label{fig:STAD_capa}
\end{figure}

\section{Methodology}
\label{section:method}

In this section we introduce the core of the STAD methodology. The approach essentially consists of four key steps:

\paragraph{Step 1: Initial Estimation} 
Obtain an initial trend estimation $\hat{T}_{0}$ by applying the proposed robust trend estimator in \eqref{eq:trend_estimator} to a differenced sequence of $Y$, which has removed the seasonality. Subsequently, obtain an initial estimate of the seasonality $\hat{S}_{0}$ by applying the robust seasonality estimator in \eqref{eq:season_m_estimator} to the remainder $Y - \hat{T}_{0}$.

\paragraph{Step 2: Trend Estimation}\label{step:trend}
Obtain a refined trend estimation $\hat{T}$ by applying the robust trend estimator in \eqref{eq:trend_estimator} to the remainder $Y - \hat{S}_{0}$. This process simultaneously estimates anomalous intervals, denoted by $\hat{\mathcal{A}}_0$.

\paragraph{Step 3: Seasonality Estimation}
Obtain a refined seasonality estimation $\hat{S}$ by applying the robust seasonality estimator in \eqref{eq:season_m_estimator} to the remainder $Y - \hat{T}$, but excluding the estimated anomalous intervals $\hat{\mathcal{A}}_0$.

\paragraph{Step 4: Anomaly Detection}\label{step:anomaly}
Finally, we estimate the anomalous intervals $\mathcal{A}$ and the anomaly component $A$ by our modified version of the CAPA algorithm, detailed in Section~\ref{section:anomaly_detection}, applied to the remainder $Y - \hat{S} - \hat{T}$.

We begin in Section~\ref{sec:model} by introducing our modelling framework. If the noise level $\sigma_0$ is unknown, then we estimate it by the robust estimator $\hat{\sigma}_0$ as defined in \eqref{eq:sigma0_iqr}. To ease presentation, we start with simplified settings: In Section~\ref{section:anomaly_detection} we explain how we estimate the anomaly component once trend and seasonality have been removed. In Section~\ref{section:trend_estimation} we propose a robust method to estimate the trend in the presence of collective anomalies, but under the assumption that the seasonality has been removed, either by estimation or as detailed in Section~\ref{section:initial_trend} by differencing the data. Similarly, Section~\ref{section:season} details the estimation of the seasonal component in the presence of collective anomalies, but when the trend has been removed. All estimators are summarized in Algorithms~\ref{alg:stad}, \ref{alg:estimate_trend}, \ref{alg:estimate_season}, and \ref{alg:estimate_anomaly} in Appendix~\ref{sec:pseudocodes}.

\subsection{Model}
\label{sec:model}
Recall from Section \ref{section:introduction} that the observation vector $Y = (y_1,\ldots,y_n)$ is given by the composition
\begin{equation}\label{eq:assumptionY}
Y = A + T + S + \varepsilon,
\end{equation}
with anomalies $A = (a_1, \ldots, a_n)$, trend $T = (t_1,\ldots,t_n)$, seasonality $S = (s_1,\ldots,s_n)$, and noise $\varepsilon = (\varepsilon_1, \ldots, \varepsilon_n)$. The components of \eqref{eq:assumptionY} are defined as follows: 

\textit{Anomalies}, $A$, consist of both collective and point anomalies
\begin{equation*}
     a_i = a_i^{\operatorname{coll}} + a_i^{\operatorname{point}}.
\end{equation*}
Here, collective anomalies are defined as
\begin{equation*}
    a_i^{\operatorname{coll}} = 
    \begin{cases}
    \mu_k, & \text{if } b_k < i \leq e_k, \\
    0, & \text{otherwise},
    \end{cases}
\end{equation*}
with $K \in \mathbb{N}$ being the number, $0 \leq b_1 < e_1 \leq \cdots \leq b_K < e_K \leq n$ being the start and end points, and $\mu_1,\ldots,\mu_K$ being the magnitudes of the collective anomalies. Furthermore, let $\mathcal{A} = \bigcup_{k = 1}^{K} (b_k, e_k]$ be the union of all anomalous intervals. Let $O = \{o_1, \ldots, o_m\}$ be the index set of the point anomalies, i.e.~$ a_i^{\operatorname{point}} \neq 0 $ if and only if $ i \in O $.

% Trend Component
The \textit{trend} is given by a polynomial of degree $Q$, i.e.
\begin{equation}
    \label{eq:trend_matrix}
    T = X \beta,
\end{equation}
with regression matrix $ X = (x_{i,j})_{i = 1, \ldots, n;\, j = 1, \ldots, Q+1} $ having entries
\begin{equation}
    \label{x_iq}
    x_{i,j} = \left( \frac{i}{n} \right)^{j - 1},
\end{equation}
and coefficients $ \beta = (\beta_0, \ldots, \beta_Q) \in \mathbb{R}^{Q + 1} $.

% Seaoson Component
We assume that the \textit{seasonality} has a known period $P$, and is given by 
\begin{equation}
    \label{eq:season_index}
    s_i = \vartheta_{(i \bmod P)},
\end{equation}
with $\vartheta_0,\ldots,\vartheta_{P - 1} \in \mathbb{R}$. For identifiability, we further assume 
\begin{equation}
    \label{eq:season_sum}
    \sum_{i = 1}^{n} s_i = 0.
\end{equation}
I.e.,~the global constant is assumed to be specified by the trend, without a baseline in the anomaly component or a contribution by the seasonality.

Finally, we assume the \textit{noise} $\varepsilon_1,\ldots,\varepsilon_n$ to be i.i.d. centered Gaussian distributed with variance $\sigma_0^2$, i.e. 
\begin{equation}
    \label{assump:varepsilon}
    \varepsilon_i \sim \mathcal{N}(0, \sigma_0^2).
\end{equation}
The global standard deviation $\sigma_0$ can be pre-estimated using robust difference based estimators \citep{frick2014multiscale, dette1998estimating}. More precisely, we use the IQR-based estimator from \citet{pein2017heterogeneous}:
\begin{equation}
    \label{eq:sigma0_iqr}
    \hat{\sigma}_0 = \frac{1}{\sqrt{2}} \cdot \frac{\text{IQR}(\Delta y)}{\text{IQR}_{\Phi}},
\end{equation}
where $\Delta y = (y_{i+1} - y_i)_{i=1}^{n-1}$ are the first order differences, $\text{IQR}(\Delta y)$ is the interquartile range of $\Delta y$, and $\text{IQR}_{\Phi}$ is the interquartile range of a standard Gaussian distribution. The factor $1/\sqrt{2}$ corrects for the variance inflation introduced by differencing.

\subsection{Anomaly Detection}
\label{section:anomaly_detection}
We now present our methodology for estimating the collective and point anomalies. To this end, we assume that trend and seasonality have been removed, up to small estimation errors $R, W \in \mathbb{R}^n$, i.e.,~our method is based on the model
\begin{equation} 
\label{eq:anomaly_model}
Y = A + R + W + \varepsilon.
\end{equation}
To infer the number and locations of anomalous segments, we use a penalized cost approach analogous to CAPA \citep{fisch2022linear}. We employ a change-in-mean version and assume that the trend estimation has removed any baseline, i.e.~we apply CAPA with baseline $\mu_0 = 0$. More precisely, for any vector of observations $ Z \coloneq (z_1, z_2, \ldots, z_n)$ we define the costs for a constant segment from $i$ to $j$ as
\begin{equation}
    \label{eq:segment_cost}
    C_{i,j}(Z) \coloneq \sum_{l=i}^{j} z_l^2,\quad 1 \leq i \leq j \leq n.
\end{equation}
Let $ \tilde{\mathcal{A}} \coloneq \bigcup_{k = 1}^{\tilde{K}} (\tilde{b}_k, \tilde{e}_k]$ be a candidate set of anomalies. Here, $ \tilde{b}_k $ and $ \tilde{e}_k $ denote the start and end of the $ k $-th anomalous segment, respectively. Let further $ \tilde{\mu} = (\tilde{\mu}_1, \ldots, \tilde{\mu}_{\tilde{K}}) $ be the candidate vector for the anomaly magnitudes. We then define their penalised costs as
\begin{equation}
    \label{eq:anomaly_cost}
    C(Z, \tilde{\mathcal{A}}, \tilde{\mu}) \coloneq \sum_{k=0}^{\tilde{K}} C_{\tilde{e}_k + 1, \tilde{b}_{k+1}}(Z) + \sum_{k=1}^{\tilde{K}} C_{\tilde{b}_k + 1, \tilde{e}_k}(Z - \tilde{\mu}_k) + \tilde{K} \lambda_{\operatorname{coll}} \hat{\sigma}_0^2 \log(n),
\end{equation}
 By convention, we set $ \tilde{e}_0 = 0 $ and $ \tilde{b}_{\tilde{K} + 1} = n$ and for any constant $a \in \mathbb{R}$ we write $Z - a$ for the vector $(z_1 - a, \ldots, z_n - a)$. The first two terms in \eqref{eq:anomaly_cost} quantify the deviation of observations from the baseline $0$ and the anomaly magnitude respectively. The final term, $\tilde{K} \lambda_{\operatorname{coll}} \hat{\sigma}_0^2 \log(n)$ is a SIC-type penalty. The cost function can readily be extended to allow for point anomalies. Specifically, let $\tilde{\mathcal{O}} \subset \{1, \ldots, n\}$ be a candidate set of point anomalies. Then, \eqref{eq:anomaly_cost} is replaced by 
\begin{equation}\label{eq:
anomaly_cost_point}
\begin{split}
    C_O(Z, \tilde{\mathcal{A}}, \tilde{\mu}) \coloneq & \sum_{\scriptscriptstyle i \notin \tilde{\mathcal{A}} \cup \tilde{\mathcal{O}}} z_i^2  + \sum_{\scriptscriptstyle k=1}^{\scriptscriptstyle \tilde{K}} \Big[ C_{\tilde{b}_k + 1, \tilde{e}_k}(Z - \tilde{\mu}_k) + \lambda_{\operatorname{coll}} \hat{\sigma}_0^2 \log(n)\Big] 
     + \sum_{\scriptscriptstyle i \in \tilde{\mathcal{O}}} \big[ \lambda_{\operatorname{point}} \hat{\sigma}_0^2 \log(n)\big],
\end{split}
\end{equation}
The reformulation of the first term in \eqref{eq:anomaly_cost} follows directly from \eqref{eq:segment_cost}. The second term is unchanged and for the final term note that the unpenalised cost term is equal to $0$ for a point anomaly, since the estimated function value is equal to the observation value. Similar to the collective anomaly penalty in \eqref{eq:anomaly_cost}, we add a SIC-type penalty $\lambda_{\operatorname{point}} \hat{\sigma}_0^2 \log(n)$. By default, we recommend $\lambda_{\operatorname{coll}} = 4$ for collective anomalies and $\lambda_{\operatorname{point}} = 3$ for point anomalies. The reasoning behind this choice is discussed in Section~\ref{Simu:penalty}.

The final estimate of the anomalies is obtained by minimizing the penalized cost
\begin{equation}
    \label{eq:capa_estimator}
    (\hat{\mathcal{A}}, \hat{\mu}) \coloneq \operatorname*{argmin}_{(\tilde{\mathcal{A}}, \tilde{\mu})} C(Z, \tilde{\mathcal{A}}, \tilde{\mu}),
\end{equation}
and the final estimate of the anomaly component is $\hat{A} = (\hat{a}_1,\ldots,\hat{a}_n)$ with
\begin{equation}\label{eq:estimateAnomalyComponent}
  \hat{a}_i = \begin{cases}
\hat{\mu}_k & \text{if } \exists\ k\,:\, \hat{b}_k < i \leq \hat{e}_k,\\
y_i & \text{if } i \in \hat{\mathcal{O}},\\
0 & \text{otherwise.}
\end{cases}
\end{equation}
The penalized cost is additive over segments and can therefore be minimized via dynamic programming, leveraging the shared zero mean across non-anomalous segments \citep{fisch2022linear}. To calculate \eqref{eq:capa_estimator}, we employ the Generalized Functional Pruning Optimal Partitioning (GFPOP) software package \citep{hocking2022generalized}.

\subsection{Trend Estimation in the presence of anomalies}
\label{section:trend_estimation}
We now turn to consider the challenge of estimating the polynomial trend in the presence of anomalies, and in particular collective anomalies. We begin by assuming that the seasonality has been removed up to a small estimation error $W \in \mathbb{R}^n$, that is, we model our data by
\begin{equation}
\label{eq:trend_anomaly_model}
Y = T + A + W + \varepsilon.
\end{equation} 
Our procedure is based on finding an index set \(U \subset \{1,\ldots,n\}\), such that $U \cap \mathcal{A} = \emptyset$, and \(U\) is large enough and well structured to permit a reliable estimation of the trend. Specifically, we propose in Section~\ref{section:sampling} a sampling procedure that obtains a collection of index sets $\mathcal{U}$. For each $U \in \mathcal{U}$, the trend is then estimated by a robust estimator using only the observations in $U$, see Section~\ref{section:M-estimator}. To select our final estimate we apply our anomaly detection method to the remainder of the trend estimate and choose the one that minimises the costs \eqref{eq:anomaly_cost}, full details are given in Section~\ref{section:trend_model_select}.

\subsubsection{M-estimator}
\label{section:M-estimator}
% given a U, we define an estimator for \beta$
Given an index set $U \subset \{1,\ldots, n\}$, which may include both anomalous and non-anomalous observations, we estimate \(\beta\) by an M-estimator using solely the data in $U$. We aim to obtain a good trend estimate if $U \cap \mathcal{A} = \emptyset$. We use the following wording in the following. 
\begin{Definition}\label{def:anomalous}
We call an index set $U \subset \{1,\ldots,n\}$ \textit{anomalous} if and only if at least one of its points is part of a collective anomaly, i.e.~$U \cap \mathcal{A} \neq \emptyset$.
\end{Definition}
If $U$ were free of any anomalies, then least squares regression would clearly be appropriate. However, since point anomalies may be present, a robust trend estimator is preferable. Because of its strong robustness properties, we adopt the M-estimator \citep{huber1964robust} with Tukey's biweight loss. For any observation vector $Z = (z_1, \ldots, z_n)$ and any index set $U$, we define an M-estimator for the polynomial trend $T$ in \eqref{eq:trend_anomaly_model} by
\begin{equation}
    \label{eq:m-estimator}
    \hat{\beta}(U) \coloneq \operatorname*{argmin}_{\beta \in \mathbb{R}^{Q + 1}} \sum_{i \in U
}\rho\left[\frac{z_{i} - \sum_{q=0}^{Q}\beta_{q}x_{i, q}}{\hat{\sigma_0}}\right],
\end{equation}
% illustration of parameters
where \(\rho(\cdot)\) is the robust loss function, $x_{i,q}$ are the polynomial regressors \eqref{x_iq} and $\hat{\sigma}_0$ is the estimated scale parameter \eqref{eq:sigma0_iqr}. For the loss function $\rho$, we choose Tukey's biweight loss function \citep{beaton1974fitting},
\begin{equation}
    \label{eq:tukey_loss}
    \rho_c(x) = 
    \begin{cases}
    \frac{c^2}{6} \left[ 1 - \left(1 - \left(\frac{x}{c}\right)^2\right)^3 \right], & \text{if } |x| \leq c, \\
    \frac{c^2}{6}, & \text{if } |x| > c.
    \end{cases}
\end{equation}
Unlike for instance Huber's loss function which only reduces the contribution of outliers to the loss, the biweight function does not increase for outliers beyond the threshold $c$; set to 4.685 by default. M-estimator \citep{huber1964robust} with Tukey's biweight loss is robust against substantial outlier contamination while maintaining statistical efficiency \citep{alma2011comparison}.

The non-convex optimization challenge coming from using Tukey's biweight loss is alleviated by our sampling approach as we rerun the algorithm multiple times for different observations. We employ the widely used iteratively reweighted least squares (IRLS) algorithm \citep{beaton1974fitting} to compute \eqref{eq:m-estimator}, details are provided in Appendix~\ref{IRLS}.

% % non-convex
% However, robust loss functions like Tukey's biweight are non-convex, introducing optimization challenges due to the presence of local minima. 

% % Other method for start methods
% Standard initialization approaches such as ordinary least squares estimation \citep{holland1977robust} and median regression \citep{andrews1974robust}, fail to handle collective or highly influential point anomalies effectively.
% % Our Sampling method for starting points
% To alleviate the non-convex optimization challenge when using Tukey's biweight loss, we propose a sampling-based initialization strategy for Iteratively Reweighted Least Squares. 

\subsubsection{Sampling Procedure}
\label{section:sampling}

We aim at sampling a collection of index sets $\mathcal{U}$ that contains at least one index set $U$ that is non-anomalous, i.e.~$U \cap \mathcal{A} = \emptyset$, and secondly that is large enough and well structured to allow a good estimate of $\beta$. 
%$U_j$ is large enough to estimate $\beta$ well enough.
%Each $U_j$ comprises $Q+1$ uniformly sampled segments corresponding to a polynomial trend of degree $Q$, with nonzero spacing between segments.

% Summary of our method
We use the following structured partitioning approach to sample $U$. First, we divide the index set $\{1,\ldots,n\}$ into \(2(Q+1)\) disjoint segments $u_1,\ldots,u_{2(Q+1)}$ of roughly equal length (be made more precise later).
% \begin{equation}
%     \label{eq:segments_sample_degree)}
%     u_l = \left\{ \Big\lfloor (k - 1) \frac{n}{2(Q+1)} \Big\rceil + 1, \ldots, \Big\lfloor k \frac{n}{2(Q+1)} \Big\rceil \right\}, \quad k = 1, \ldots, 2(Q+1).
% \end{equation}  
We will sample a sub-segment from either each odd or from each even indexed segment. More precisely, let $l \sim \operatorname{Bern}(\frac{1}{2})$ be a Bernoulli-distributed random variable. We then sample sub-segments from $u_{l}, u_{l + 2}, \ldots, u_{l + 2Q}$ (see below). This constructions ensures that the sampled sub-segments are well separated from each other, which is a key requirement to estimate the $Q + 1$ parameters in $\beta$ well.

% This ensures that we have selected $(Q+1)$ sub-segments that are spaced at least $\Big\lfloor \frac{n}{2(Q+1)} \Big\rceil$ apart. This is necessary to estimate the $Q + 1$ parameters in $\beta$ well. For the proof of the probability of sampling guarantee, see Lemma~\ref{lemma:probGoodSample}.

We now detail how we sample from $u_k$. We further partition each segment $u_k$ into $B$ sub-segments of roughly equal length. In total, we obtain
\begin{equation}
\label{eq:V_B_Q}
V = B \times 2(Q + 1)
\end{equation}  
sub-segments
\begin{equation}
    \label{eq:segments_sample}
    \tilde{u}_k \coloneq \left\{ \Big\lfloor (k - 1) \frac{n}{V} \Big\rceil + 1, \ldots, \Big\lfloor k \frac{n}{V} \Big\rceil \right\}, \quad k = 1, \ldots, V.
\end{equation}  
We then sample one sub-segment from each selected segment. More precisely, we obtain the index set
\begin{equation}
   \label{eq:final_sample}
   U \coloneq \bigcup_{k = 0}^{Q} \tilde{u}_{(2k + l) B + m_k},
\end{equation}  
where $l \sim \operatorname{Bern}\big(\frac{1}{2}\big)$ is a Bernoulli-distributed random variable and $m_k \sim \text{Uniform}\big(\{1,\ldots,B\}\big)$. All random variables are independent of each other.

In summary, suppose appropriately chosen tuning parameters (see Section~\ref{section:trend_model_select}), the proposed sampling procedure obtains a collection of index sets $\mathcal{U}$ that ensures the following: With high probability at least one of the $U \in \mathcal{U}$ is non-anomalous, see Lemma~\ref{lemma:probGoodSample}. At the same time each $U \in \mathcal{U}$ is constructed such that if it is non-anomalous, then the trend is estimated well (with high probability), see Lemma~\ref{lemma:trendEstimation}. A key step was Lemma~\ref{lemma:determinant} in which we have shown that the determinant of $X_U^\top X_U$ is lower bounded. This lemma relies heavily on the fact that the $\tilde{u}_{(2k + l) B + m_k}$'s are well separated from each other.

\subsubsection{Trend estimate selection}
\label{section:trend_model_select}

To obtain a good estimation of the trend $T$, we repeat the sampling procedure multiple times. We also have to select the tuning parameter \(B\) and polynomial degree \(Q\). 

To this end, let \(\mathbf{B} = \{B_1,\ldots,B_b\}\) and \(\mathbf{Q} = \{Q_1,\ldots,Q_q\}\) be the candidate sets for \(B\) and \(Q\), respectively. For each combination \(\tilde{B} \in \mathbf{B}\) and \(\tilde{Q} \in \mathbf{Q}\) we repeat the sampling procedure in Section~\ref{section:sampling} $J$ times. We obtain the collection of index sets
\begin{equation*}
  \mathcal{U} \coloneq \big\{U_{\tilde{B}, \tilde{Q}, j} \subset \{1,\ldots,n\} : \tilde{B} \in \mathbf{B},\, \tilde{Q} \in \mathbf{Q},\, j = 1,\ldots,J  \big\},
\end{equation*}
and trend estimates
\begin{equation}\label{eq:setOfTrendEstimates}
\hat{\mathcal{T}} = \big\{\hat{T}_{U} : U \in \mathcal{U} \big\},
\end{equation}
with $\hat{T}_{U} \coloneq X \hat{\beta}(U)$ and $\hat{\beta}(\cdot)$ as defined in \eqref{eq:m-estimator}.

We then select the final trend estimate \(\hat{T}\) by selecting the one that minimizes the penalized cost of the anomaly estimator in \eqref{eq:anomaly_cost} applied to $Y - \hat{T}_{U}$ plus a penalty for the degree of the trend, i.e.
\begin{equation}
    \label{eq:trend_estimator}
    \hat{T} = \argmin_{T \in \hat{\mathcal{T}}} \min_{(\tilde{\mathcal{A}}, \tilde{\mu})}  C(Y - T, \tilde{\mathcal{A}}, \tilde{\mu}) + \tilde{Q}\, \hat{\sigma}_0 \log(n).
\end{equation}

The tuning parameter $J$ is chosen to be sufficiently large such that with high probability at least one of the $J$ sampled segments is non-anomalous, confer Lemmas~\ref{lemma:intersectionAnomaly}~and~\ref{lemma:probGoodSample}. By default we suggest $J = 20$ based on the simulations in Section~\ref{subsection:sensitivity_J}. The optimal choice of the tuning parameter $B$ is contingent on anomaly characteristics, e.g. number and length of collective anomalies. If $B$ is too small and collective anomalies are long, then all samples in $\mathcal{U}$ may intersect with collective anomalies. Conversely, if $B$ is too large, then the sample $U_{j, B, Q}$ is too small to allow a good estimation of the trend. These points are further illustrated in Appendix~\ref{simu:B}. Finally, the correct polynomial degree $Q$ should be contained in $\mathbf{Q}$, though a too large degree may still allow good estimates. Based on empirical studies, in the absence of further information we recommend $\mathbf{B} = \{1, 3, 5\}$ and $\mathbf{Q} = \{0, 1, 2, 3\}$ as they provide a good trade-off between statistical efficiency and computational demands.

\subsection{Initial trend estimation}
\label{section:initial_trend}
To estimate the trend without interference from the seasonality, we apply differencing to $Y$ with lag $D = C \cdot P$, where $C \in \mathbb{N}$. More precisely, we use the differenced series $Y_D$, with entries
\begin{equation*}
y_{D,i} = y_{i + D} - y_i, \quad \text{for} \quad i = 1, \ldots, n - D,
\end{equation*}
which has no seasonal component, confer Lemma~\ref{lemma:differencedSequence}. The differenced sequence has trend
\begin{equation*}
T_D = X_D \boldsymbol{\beta}_D,
\end{equation*}
where the regression matrix $X_D$ has entries
\begin{equation*}
(X_D)_{i,q} \coloneq X_{i+D, q + 1} - X_{i, q + 1} = \left(\frac{i+D}{n}\right)^q - \left(\frac{i}{n}\right)^q, 
\end{equation*}
for $i = 1,\ldots,n - D,\ q = 1,\ldots,Q$, and parameter vector $\boldsymbol{\beta}_D = (\beta_1,\ldots,\beta_Q)$.
% difference on noise and anomaly
Lemma~\ref{lemma:differencedSequence} confirms this trend and details the effects of differencing on the noise and anomaly components, respectively. We then estimate $\boldsymbol{\beta}_D$ and hence all entries of $\boldsymbol{\beta}$ except $\beta_0$ by applying the trend estimation method from Section~\ref{section:trend_estimation} to $Y_D$. We denote the resulting parameter estimate by 
\begin{equation*}
    \hat{\beta}_{D} = (\hat{\beta}_{D, 1}, \ldots, \hat{\beta}_{D, Q})
\end{equation*}
% initial season estimation
and we obtain the initial trend $\hat{T}_0$ as
\begin{equation}\label{eq:initial_trend}
    \hat{T}_0 \coloneq X \begin{pmatrix} 0 \\ \hat{\beta}_D \end{pmatrix}.
\end{equation}

The selection of an appropriate multiplier $C$ is crucial, as an inappropriate choice can lead to non-detectable anomalies. Theoretical justification is provided in Lemma~\ref{lemma:differencedSequence}. By default, $C$ is chosen to be an integer such that $D = C \cdot P \approx 0.1 n$ and if $P \ge n/10$, then $C=1$.

\subsection{Seasonal Component Estimation in the presence of anomalies}
\label{section:season}
% Model assumption for season estimation
% brief introduction of estimator
We now introduce a robust estimator for the seasonal component that allows for the presence  of anomalies. To this end, we assume that the trend has been removed, up to a small estimation error $R \in \mathbb{R}^n$. Under this assumption, the data follows the model:  
\begin{equation}
\label{eq:season_model}
Y = S + A + R +\varepsilon.
\end{equation}

We also allow the estimator to take into account an estimation of the anomalies $\hat{\mathcal{A}}$, obtained by \eqref{eq:trend_estimator}. Let $\hat{\mathcal{A}} = \emptyset$ for the initial seasonality estimation $\hat{S}_0$.

For $p = 0,\ldots,P -1$, to estimate $\vartheta_p$ we use the index set
\begin{equation}\label{eq:setSeasonalEstimation}
V_p = \big\{ i \in \{1,\ldots,n\} \setminus \hat{\mathcal{A}} : (i \bmod P) = p\big\},
\end{equation}
and Tukey's M-estimator
\begin{equation}
\label{eq:season_m_estimator}
   \hat{\vartheta}_p = \hat{\beta}(V_p),
\end{equation}
with $\hat{\beta}$ being the M-estimator from \eqref{eq:m-estimator} and $\beta \in \mathbb{R}$. For robustness, we initialize the estimator with the median of $Y_{V_p}$ rather than using Ordinary Least Squares (OLS). The final seasonal estimate $\hat{S} = (\hat{s}_1,\ldots,\hat{s}_n)$ is then defined as
\begin{equation}\label{eq:estimateSeasonality}
\tilde{s}_i \coloneq \hat{\vartheta}_{i \bmod P} \text{ and } \hat{s}_i \coloneq \tilde{s}_i - n^{-1} \sum_{i = 1}^{n} \tilde{s}_{i}.
\end{equation}
This ensures that $\hat{S}$ satisfies \eqref{eq:season_sum}.

\subsubsection{Smoothing}\label{sec:smoothing}
If the underlying seasonality is assumed to be smooth, we further improve the estimate $\hat{S}$ by applying spline smoothing. This enhances accuracy and interpretability. We choose spline smoothing due to its statistical robustness, computational efficiency, and flexibility in capturing complex seasonal patterns \citep{de1978practical}, but other smoothing approaches would be suitable as well. More precisely, we apply cubic smoothing splines to $(\hat{\vartheta}_0, \ldots, \hat{\vartheta}_{P - 1})$ with the smoothing parameter selected by Generalized Cross-Validation. To this end, we use the R function \textit{smooth.spline}.

\section{Theory}\label{section:theory}

With our methodology now established, we discuss theoretical properties of the STAD approach in detecting anomalies and decomposing the time series into anomalies, trend, seasonality, plus a remainder. Specifically in Theorem~\ref{theorem:main} we establish asymptotic L2-bounds on the error of each component's estimate individually. Theorem~\ref{theorem:main} also shows that STAD estimates the number of anomalies consistently and their start and end points are close to the true ones. More precisely, we show that $\hat{\mathcal{A}} \in \mathcal{B}$, where
\begin{equation*}
\mathcal{B} \coloneq \{\tilde{\mathcal{A}} \text{ s.t. } \tilde{K} = K, \vert \tilde{b}_k - b_k \vert < \kappa_{k}, \vert \tilde{e}_k - e_k \vert  < \kappa_{k} ,\ \forall\ k = 1,\ldots,K\}.
\end{equation*}
Here, the errors $\kappa_k$ are defined as
\begin{equation}\label{eq:definitionKappak}
\kappa_k \coloneq \left\{\bigg\lceil \Clambda\frac{\sigma_0^2}{\mu_k^2}\log(n)^{1 + \delta} \bigg\rceil,\ 4^{Q + 1}\right\},\quad k = 1,\ldots,K,
\end{equation}
with $\Clambda > 0$ being a large enough constant such that the last step in the proof of Lemma~\ref{lemma:costMissingCp} holds. The maximum with $4^{Q + 1}$ ensures that we have locally enough observations. For Theorem~\ref{theorem:main} to hold we require the following assumptions.
\begin{Assumption}\label{assumption:segmentLength}
There exists a $\tilde{\delta} > 0$ such that
\begin{equation}\label{eq:firstAssumptionSegmentLength}
e_k - b_k \geq \frac{\sigma_0^2}{\mu_k}\log(n)^{1 + \delta + \tilde{\delta}},\ \forall\ k = 1,\ldots,K.
\end{equation}
Furthermore, we require that
\begin{equation}\label{eq:secondAssumptionSegmentLength}
\min_{k = 0,\ldots,K} \{b_{k + 1} - e_k\} \geq 3\max_{k = 1,\ldots,K} \{e_k - b_k\} + 1,
\end{equation}
and $D = C P$ is such that 
\begin{equation}\label{eq:conditionD}
2 \max_{k = 1,\ldots,K} \{e_k - b_k\} \leq D = C P \leq \min_{k = 0,\ldots,K} \{b_{k + 1} - e_k\} - \max_{k = 1,\ldots,K} \{e_k - b_k\} - 1.
\end{equation}
\end{Assumption}
 Equation~\eqref{eq:firstAssumptionSegmentLength} in Assumption~\ref{assumption:segmentLength} is a standard assumption on the length and size of an collective anomaly to be detectable. It mirrors Assumption~2 in \citep{fisch2020real} for the CAPA algorithm in the absence of trend and seasonality. Note that the result $\hat{\mathcal{A}} \in \mathcal{B}$ is identical, up to the definition of $\kappa_k$ as we consider a change-in-mean instead of a change-in-mean-and-variance problem. The second requirement \eqref{eq:secondAssumptionSegmentLength} ensures that anomalies are separated and hence individually detectable after differencing and is not required in the absence of either trend or seasonality. 

\begin{Assumption}\label{assumption:anomalyLengthCombined}
The total length of collective anomalies is bounded. More precisely, there exists a $\gamma < 1$ such that
\begin{equation*}
\sum_{k = 1}^{K} e_k - b_k < \gamma \frac{1}{\max(4 (Q + 1), 8)}.
\end{equation*}
\end{Assumption}
Assumption~\ref{assumption:anomalyLengthCombined} ensures that we have locally enough non-anomalous observations to estimate trend and seasonality well. Remarkably, the total length of anomalies can be of order $n$. This shows that we have overcome the substantial limitation of CAPA given by Assumption~2(b) in \citep{fisch2020real} that anomalies can be at most of order $\sqrt{n}$ to allow a good estimate of the baseline. The $4 (Q + 1)$-term is required for a good trend estimation, while the $8$ is required for a good seasonality estimation.

\begin{Assumption}\label{assumption:gridBandQ}
\begin{enumerate}[(i)]
\item \label{assumption:gridB} There exists a large enough $B \in \mathbf{B}$ such that \eqref{eq:conditionOnB},~\eqref{eq:secondConditionB},~and~\eqref{eq:thirdConditionB} are satisfied.
\item \label{assumption:gridQ} We have that $Q \in \mathbf{Q}$.
\end{enumerate}
\end{Assumption}
Assumption~\ref{assumption:gridBandQ} is on the tuning parameters $Q$ and $B$ to ensure that we can estimate the trend well. It requires the true order of the polynomial $Q$ to be in the set $\mathbf{Q}$. This could be weakened to only a $\tilde{Q} \geq Q$ being required. Secondly, the tuning parameter $B$ has to be large enough. We remark that the precise conditions work well for the asymptotic theory, but are not of practical relevance. We refer to Section~\ref{section:trend_model_select} and Appendix~\ref{simu:B} for a discussion of its choice.

\begin{Assumption}\label{assumption:trendSeasonalBound}
There exits a constant $\CtrendSeasonalBound > 0$ such that
\begin{equation}\label{eq:maxTrend}
\max_{i = 2,\ldots,n} \vert t_i - t_{i - 1} \vert \leq \CtrendSeasonalBound \sigma_0,
\end{equation}
and
\begin{equation}\label{eq:maxSeason}
\max_{i = 2,\ldots,n} \vert s_i - s_{i - 1} \vert \leq \CtrendSeasonalBound \sigma_0.
\end{equation}
\end{Assumption}
Assumption~\ref{assumption:trendSeasonalBound} is only used to bound the variance estimate $\hat{\sigma}_0^2$, see Proposition~\ref{proposition:varianceEstimate}. The assumption could be relaxed by allowing a small proportion of points to fail conditions \eqref{eq:maxTrend} and \eqref{eq:maxSeason}.

Building on these assumptions, Theorem~\ref{theorem:main} establishes that STAD successfully solves the structural challenge of decomposing a time series into anomalies, trend, and seasonality and to estimate collective anomalies when the typical distribution is not stationary. Our theory also accounts for the initial variance estimate. However since, to the best of our knowledge, finite sample bounds for robust estimators are not available, we assume that we have no point anomalies and use the least squares estimator instead of Tukey's estimator. I.e.,~the loss function $\rho$ is given by the $L_2$-norm. For the seasonal component, \eqref{eq:season_m_estimator} is replaced by the median of the observations indexed by $V_p$. Note that this median was the initial estimate prior to the Tukey estimate in \eqref{eq:season_m_estimator}, so in other words, we omit the refinement by the Tukey estimate. Using robust estimators instead will have slightly worse statistical efficiency, but better robustness to potential anomalies in $U$, either small parts of collective anomalies or point anomalies.

\begin{Theorem}\label{theorem:main}
We assume the model from Section~\ref{sec:model}. Let $\hat{A}, \hat{T}, \hat{S}$, and $\hat{\mathcal{A}}$ be the STAD estimators as defined in Section~\ref{section:method}. Suppose Assumptions~\ref{assumption:segmentLength},~\ref{assumption:anomalyLengthCombined},~\ref{assumption:gridBandQ},~and~\ref{assumption:trendSeasonalBound} hold. Let $J \to \infty$, as $n \to \infty$. Then,
\begin{equation}\label{eq:goodFinalAnomalyDetection}
\Pj\big(\hat{\mathcal{A}} \in \mathcal{B}\big) \to 1, \text{ as } n \to \infty.
\end{equation}
Furthermore, there exists a constant $\CfinalEstimate > 0$ such that
\begin{equation}\label{eq:goodFinalEstimate}
\Pj\Big( \max\big\{  \| \hat{A} - A \|_2^2,\, \| \hat{T} - T \|_2^2,\, \| \hat{S} - S \|_2^2\big\} \leq \CfinalEstimate \sigma_0^2 \log(n) \Big) \to 1,
\end{equation}
as $n \to \infty$.
\end{Theorem}
\begin{proof}
See Appendix~\ref{sec:proofs}.
\end{proof}

\section{Simulation study}\label{section:simulation}

We now assess STAD's finite sample performance by simulation studies. We compare STAD's performance against oracle versions of itself that know the trend (Oracle-Trend), the seasonality (Oracle-Seasonality), and both components (Full Oracle). For all methods we use the default choices for tuning parameters as detailed in Section~\ref{section:method}. A comparison with other methods is omitted, as most existing anomaly methods focus on point anomalies only and to the best of our knowledge no existing methods achieves acceptable results for detecting collective anomalies at the presence of significant trends and seasonality such that a systematic comparison is not meaningful, c.f.~the illustrations in Section~\ref{section:introduction}.

Our simulatios are based on synthetic time series of length $n = 5000$, generated from the model in Section~\ref{sec:model}. We chose the trend component to be $t_i = 2(i/n)^2 - 2(i/n)$, a sine-wave seasonality $s_i = 2 \sin(2\pi i / P)$, where $P$ is the period length, and the noise component is $\epsilon_i \sim \mathcal{N}(0, 1)$. Simulations (not displayed) showed that results are stable in the choice of the concrete form of the trend and seasonality, so for brevity we do not vary their form in the following. Instead we consider various settings for the anomalies and also vary in one study the period length $P$, which is otherwise chosen as $P=250$.

In earlier experiments (not displayed) we found that the anomaly length, the number of seasonal repetitions, and of course the signal to noise ratio have the most dominant effect on the performance of all contrasted methods. We demonstrate the first two points in the most simple setting with a single anomaly: In Sections~\ref{simulation:lengthSize}~and~\ref{simulation:seasonalLength} we vary the length of the single anomaly and the period length, respectively. Finally, in Section~\ref{simulation:multiple} we consider a more evolved setting with multiple anomalies which considers different signal to noise ratios but also different anomaly lengths. In between we discuss tuning parameter choices. Most tuning parameters have been theoretically discussed in Section~\ref{section:method}. So, in Sections~\ref{subsection:sensitivity_J}~and~\ref{Simu:penalty} we focus on the number of repetitions $J$ and the penalty, respectively.

All simulations are repeated $1\;000$ times. We then evaluate the methods based on their average accuracy in detecting the collective anomalies. We define an estimation as positive if the estimated number of anomalies is equal to the true number of anomalies and if start and end points differ by at most $\lfloor \log(n) \rfloor$. We omit other metrics for brevity and since we found for instance that the average mean square error for the anomaly, trend, and seasonal component estimation shows qualitatively the same result as the accuracy. Also changing the $\lfloor \log(n) \rfloor$ tolerance level for the locations of anomalies allows for the same comparisons.

\subsection{Impact of anomaly length}\label{simulation:lengthSize}
We consider a single collective anomaly and vary its length. The starting position is drawn uniformly from all valid positions. To create a challenging setting, we also adjust the magnitude such that Full Oracle has slightly more than 90\% accuracy.

\begin{table}[htbp]
\centering
\begin{minipage}{0.8\textwidth}
    \captionsetup{width=\linewidth}
    \caption{\footnotesize Average anomaly detection accuracy of STAD compared to three oracle version of it. We vary the length and magnitude of a single anomaly. STAD performs nearly as good as the oracles unless the anomaly is longer than a quarter of the total time series.}
    \label{tab:accuracy_comparison} 
    \setlength{\tabcolsep}{3pt} % Default is 6pt
    \centering
    \begin{tabular}
    {@{}lccccc@{}}
    \toprule
    Length (\%) & Magnitude  & STAD & Oracle Trend & Oracle Season & Oracle \\
    \midrule  
    5 (0.1\%)   & 3.24 & 0.903 & 0.907 & 0.909 & 0.918 \\
    28 (0.56\%)  & 1.42 & 0.905 & 0.907 & 0.919 & 0.924 \\
    51 (1.02\%)  & 1.25 & 0.916 & 0.921 & 0.923 & 0.926 \\
    108 (2.16\%) & 1.24 & 0.921 & 0.920 & 0.926 & 0.923 \\
    278 (5.56\%) & 1.21 & 0.917 & 0.915 & 0.925 & 0.924 \\
    768 (15.36\%)& 1.20 & 0.929 & 0.934 & 0.941 & 0.938 \\
    1178 (23.56\%)& 1.19 & 0.909 & 0.917 & 0.920 & 0.926 \\ 
    1278 (25.56\%)& 1.19 & 0.903 & 0.916 & 0.905 & 0.920 \\ 
    1378 (27.56\%)& 1.19 & 0.887 & 0.920 & 0.886 & 0.926 \\ 
    1499 (29.98\%)& 1.19 & 0.770 & 0.914 & 0.761 & 0.912 \\ 
    1678 (33.56\%)& 1.18 & 0.565 & 0.920 & 0.566 & 0.926 \\ 
    \bottomrule
    \end{tabular}
\end{minipage}
\end{table}
% robustness of STAD 

Table~\ref{tab:accuracy_comparison} confirms that STAD performs very well and detects the anomaly almost as good as the oracles. Only when the anomaly is longer than a quarter of the total length of the time series its performance degenerates significantly. In these setting it fails to estimate the trend well, as evidenced by the performance of Oracle Season which performs more or the same as STAD. This is because both methods struggle to find a non-anomalous sample $U$ as even for large $B$ too many sub-segments are anomalous.

\subsection{Effect of the seasonal length} \label{simulation:seasonalLength}

In this simulation study we vary he seasonal period $P$. To this end, we reconsider a representative case from Table~\ref{tab:accuracy_comparison} with an anomaly length of 768, but results do not depend much on this choice as long as the anomaly is not too long.

\begin{table}[htbp]
\centering
\begin{minipage}{0.8\textwidth}
    \captionsetup{width=\linewidth}
    \caption{\footnotesize 
    We consider a single anomaly of length 768 and vary the seasonal period $P$. STAD performs almost as good as the oracles unless less than five seasonal cycles $P/n$ are present.
    \label{tab:accuracy_season}
    }
    \centering
    \begin{tabular}{@{}lcccc@{}}
    \toprule
    Seasonal Period ($P$) & STAD & Oracle Trend & Oracle Season & Oracle \\ 
    \midrule
    50   & 0.936 & 0.937 & 0.941 & 0.938 \\
    100  & 0.933 & 0.930 & 0.941 & 0.938 \\
    250  & 0.933 & 0.934 & 0.941 & 0.938 \\
    500  & 0.931 & 0.932 & 0.941 & 0.938 \\
    800  & 0.924 & 0.911 & 0.941 & 0.938 \\
    1000 & 0.890 & 0.928 & 0.941 & 0.938 \\
    1500 & 0.023 & 0.008 & 0.941 & 0.938 \\
    \bottomrule
    \end{tabular}
\end{minipage}
\end{table}
% robustness
Table~\ref{tab:accuracy_season} confirms STAD's excellent performance, matching the oracles, across a wide range of seasonal repeats. However, when there are less than five repeats of the season, c.f.~$P = 1500$, both STAD and Oracle-Trend show significant degradation in performance, while oracles that know the seasonality still perform well. In such settings a very small number of points is available to estimate each value in the seasonal component and hence the anomaly effects the estimation too much.

\subsection{Choice of J} 
\label{subsection:sensitivity_J}
We now investigate choice of the tuning parameter $J$, the number of sampling iterations used to estimate the trend. We use an anomaly length of 1378 as this was a challenging scenario in our previous experiments, c.f.~Table~\ref{tab:accuracy_comparison}. We vary $J \in \{5, 10, 20, 30, 40, 50\}$.
\begin{table}[htbp]
\centering
\begin{minipage}{0.8\textwidth}
    \captionsetup{width=\linewidth}
    \caption{\footnotesize Average anomaly detection accuracy as a function of $J$. Results improve with increasing $J$. Results support our default of $J=20$ as the performance increase only marginally afterwards.
    }
    \label{tab:accuracy_vs_samples}
    \centering
    \begin{tabular}{@{}lccccc@{}}
    \toprule
    \textbf{$J$} & \textbf{STAD} & \textbf{Oracle Trend} & \textbf{Oracle Season} & \textbf{Oracle} \\
    \midrule
    5   & 0.701 & 0.920 & 0.715 & 0.926 \\
    10  & 0.832 & 0.920 & 0.839 & 0.926 \\
    20  & 0.887 & 0.920 & 0.886 & 0.926 \\
    30  & 0.896 & 0.920 & 0.908 & 0.926 \\
    40  & 0.909 & 0.920 & 0.915 & 0.926 \\
    50  & 0.904 & 0.920 & 0.910 & 0.926 \\
    \bottomrule
    \end{tabular}
\end{minipage}
\end{table}

Table~\ref{tab:accuracy_vs_samples} shows that unsurprisingly the accuracy of STAD improves with increasing $J$. Since the accuracy improves only marginally after $20$, we recommend $J = 20$ as a default choice. Choosing a larger value may improve results but at the expense of larger computational demands.

\subsection{Effect of the penalty}\label{Simu:penalty}
% aim
We now investigate the sensitivity of STAD to the collective anomaly penalty $\lambda_{\operatorname{coll}}$. We contrast its performance in three scenarios. Firstly, an easy scenario with a short anomaly of length 278 and moderate jump size 1.21 with many repeats of the seasonality ($P=250$). Secondly, a challenging scenario with a long anomaly of length 1378 and magnitude $1.19$ and a low-frequency seasonality ($P=1000$). Thirdly, a setting with a high signal to noise ratio, using the same short anomaly and seasonality as in the first case but with a significantly larger jump size of $4.2$. 

\begin{figure}
    \centering
    % First row of subplots
    \begin{subfigure}[b]{0.32\textwidth}
        \centering
        \includegraphics[width=\textwidth]{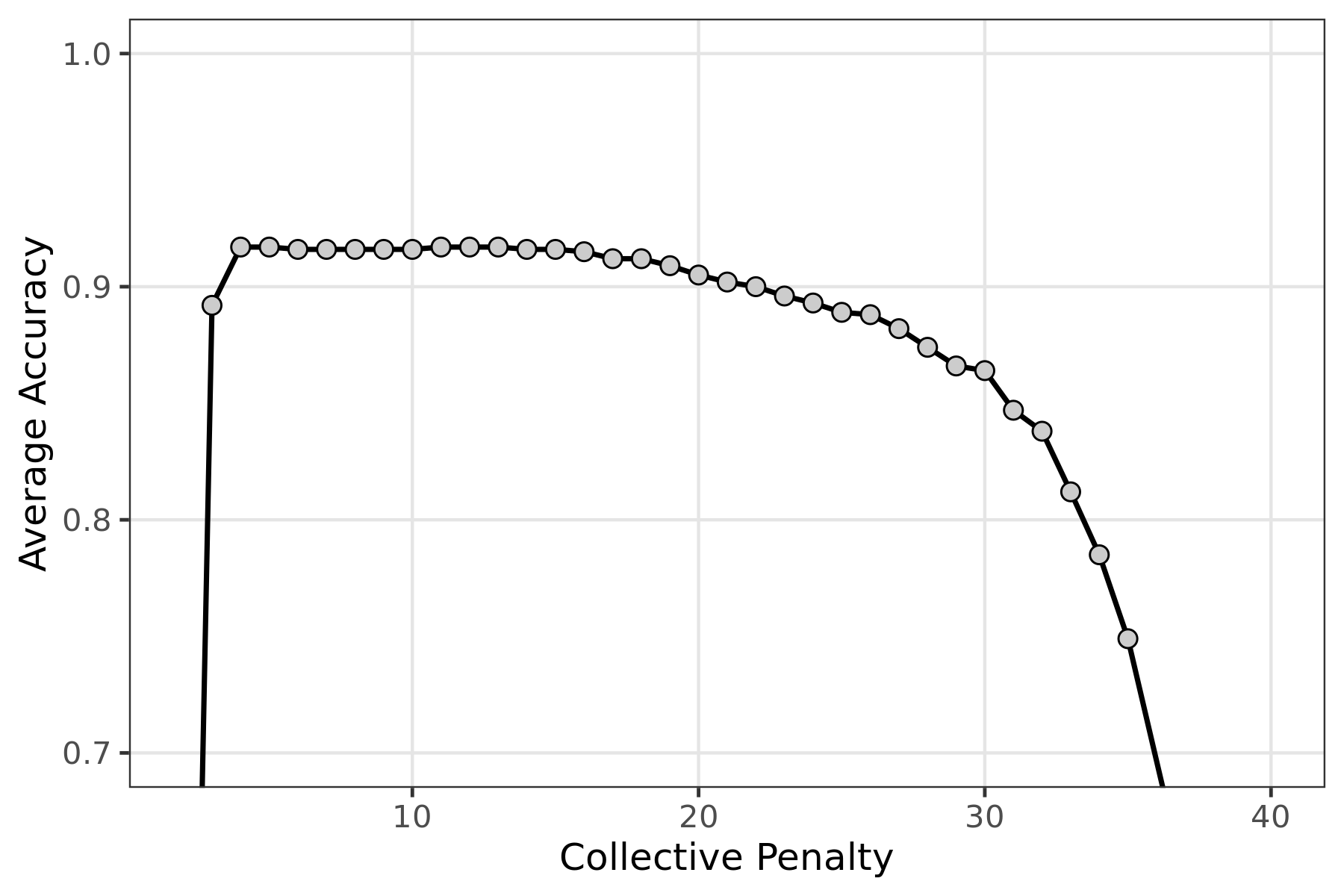}
        \caption{Scenario 1: Short Anomaly, Moderate Jump}
        \label{fig:pencoll_L278}
    \end{subfigure}
    \hfill % Adds horizontal space between the two subfigures
    \begin{subfigure}[b]{0.32\textwidth}
        \centering
        \includegraphics[width=\textwidth]{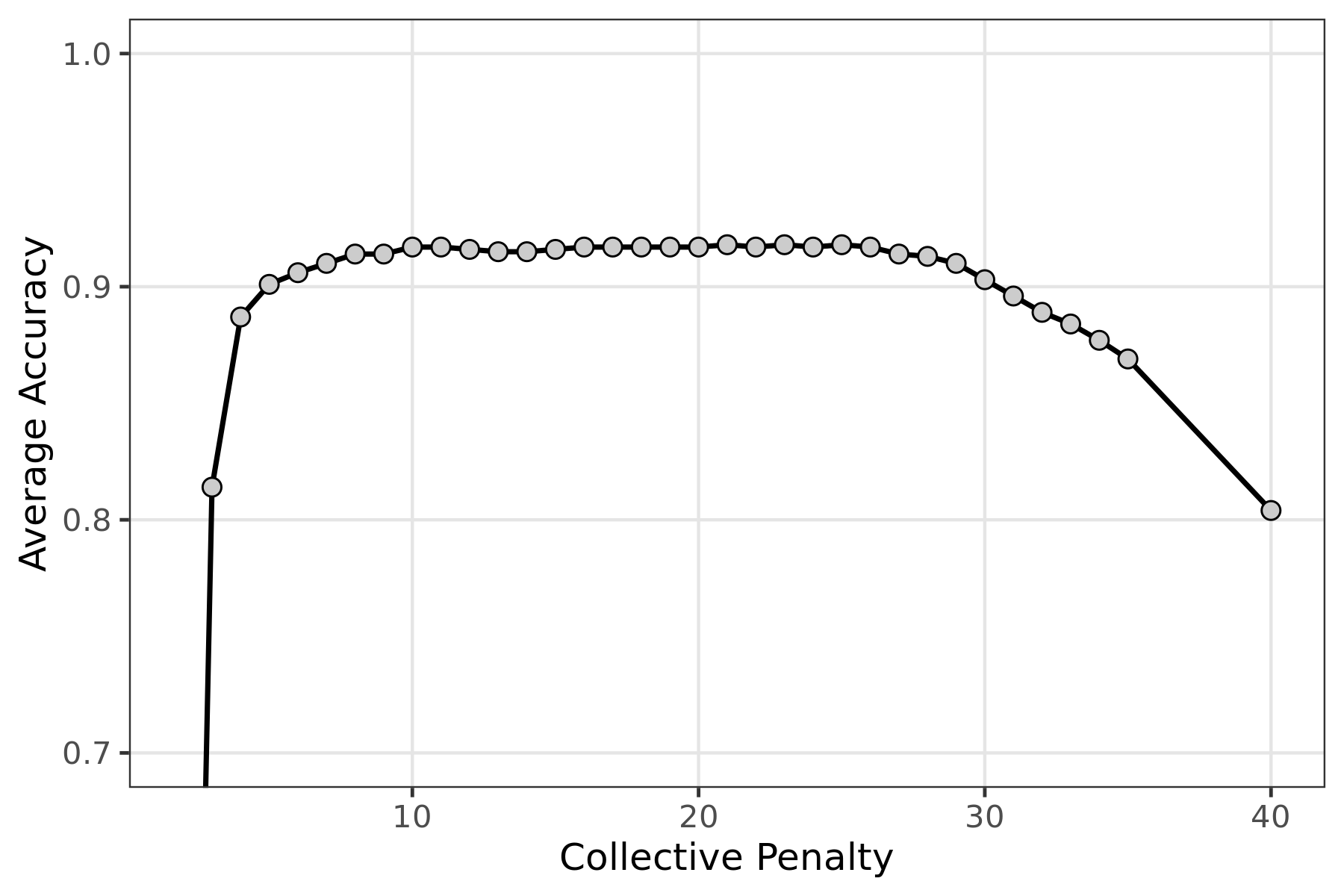}
        \caption{Scenario 2: Long Anomaly, Long Seasonality}
        \label{fig:pencoll_L1378}
    \end{subfigure}
    \begin{subfigure}[b]{0.32\textwidth}
        \centering
        \includegraphics[width=\textwidth]{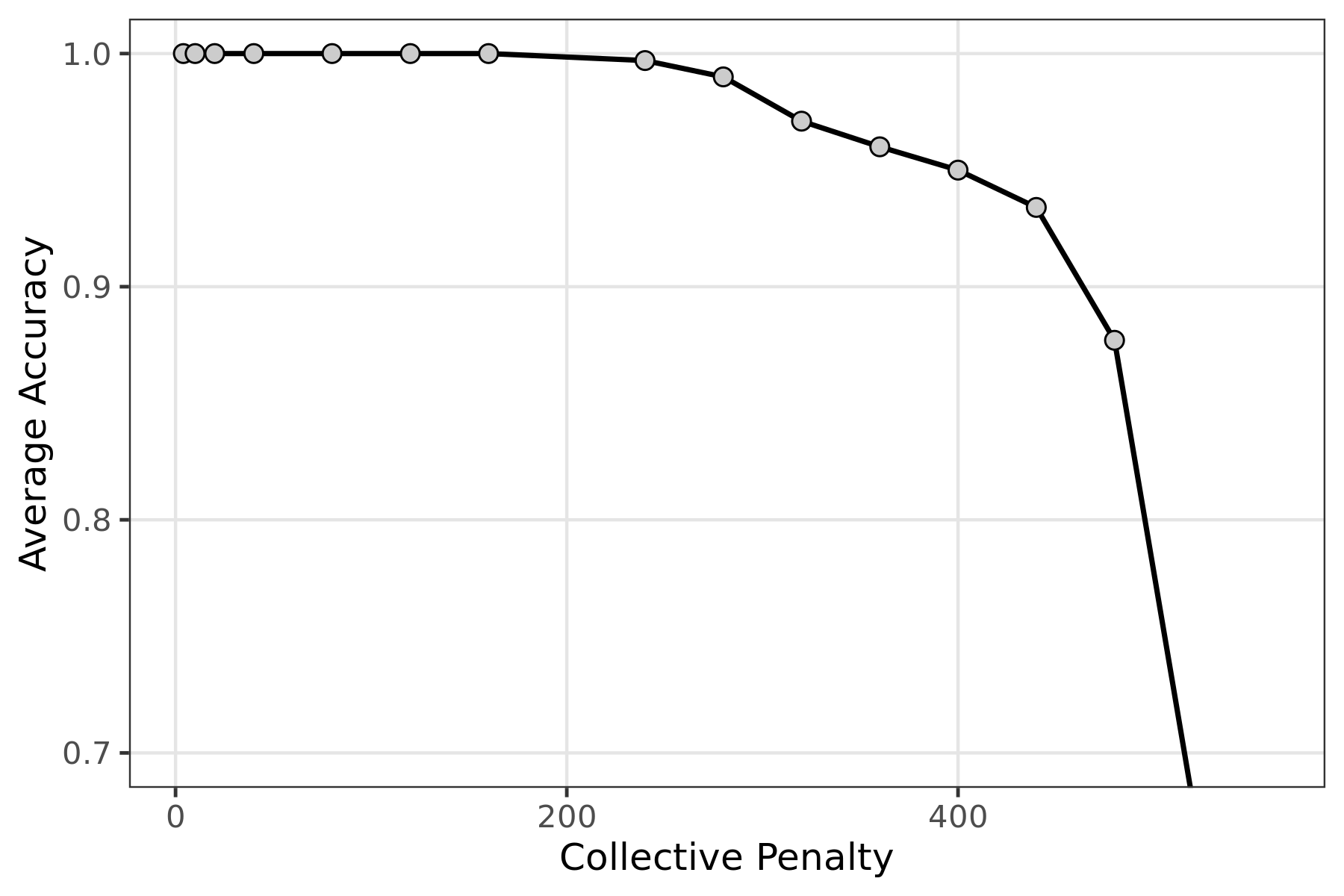} 
        \caption{Scenario 3: Short Anomaly, Large Jump}
        \label{fig:pencoll_L278_largejump}
    \end{subfigure}
    
    \caption{
        \footnotesize STAD's average detection accuracy as a function of the collective anomaly penalty $\lambda_{\operatorname{coll}}$. It shows that a wide range of penalty values offer a good performance.}
    \label{fig:pencoll_effect_all} % Use a new label for the combined figure
\end{figure}
% increase - stable - decrease
Figure~\ref{fig:pencoll_effect_all} shows how the penalty parameter $\lambda_{\operatorname{coll}}$ affects detection accuracy. As expected, in all scenarios, accuracy first rises sharply, then enters a wide, stable plateau of near-optimal performance, and finally declines as the penalty becomes too large. Based on these results, by default we recommend $\lambda_{\operatorname{coll}} = 4$ for collective anomalies and $\lambda_{\operatorname{point}} = 3$ for point anomalies. Note that this choice is slightly larger than the usual $\operatorname{BIC}$ choice that would be $3$ and $2$, respectively. This is to account for the additional errors resulting from the trend and seasonality estimation. However, if the anomaly is of high magnitude then the penalty can be chosen much larger. We will make use of this finding in our application in Section~\ref{section:applications}.

\subsection{Multiple collective anomalies}
\label{simulation:multiple}

Finally, we consider a setting with multiple anomalies. Our setting is similar to the one in \citep{fisch2022linear}, but we restrict ourself to changes-in-mean only. The number of anomalies is drawn randomly from a Poisson distribution with intensity $2.5$. The length $L$ of each anomaly is drawn from a Poisson distribution with mean $\mu_L \in \{30, 100, 400, 800\}$ and the magnitude of each anomaly is sampled from a uniform distribution taking values between $a$ and $b$. We consider $(a,b) = (1,3)$ (weak signals) and $(a,b) = (2,5)$ (strong signals). All random variables are independent of each other.

\begin{table}[htbp]
\centering
\begin{minipage}{0.8\textwidth}\captionsetup{width=\linewidth}
    \caption{\footnotesize
    Average detection accuracy in a setting with multiple anomalies. We consider different anomaly lengths ($\mu_L$) and magnitudes $(a,b)$. Results are consistent with the single anomaly scenario in Section~\ref{simulation:lengthSize}. STAD performs very well unless anomalies are long, in this case $\mu_L = 800$.}
    \label{tab:multi_accuracy_comparison}
    \centering
    \begin{tabular}{@{}lccccc@{}} 
    \toprule
    $\mu_L$ & $(a, b)$ & \textbf{STAD} & \textbf{Oracle Trend} & \textbf{Oracle Season} & \textbf{Oracle} \\
    \midrule
    30  & (1, 3) & 0.734 & 0.739 & 0.746 & 0.754 \\
    100 & (1, 3) & 0.845 & 0.845 & 0.847 & 0.852 \\
    400 & (1, 3) & 0.850 & 0.854 & 0.849 & 0.858 \\
    800 & (1, 3) & 0.646 & 0.872 & 0.647 & 0.871 \\
    \midrule
    30  & (2, 5) & 0.904 & 0.905 & 0.904 & 0.905 \\
    100 & (2, 5) & 0.909 & 0.909 & 0.909 & 0.909 \\
    400 & (2, 5) & 0.896 & 0.907 & 0.905 & 0.910 \\
    800 & (2, 5) & 0.655 & 0.894 & 0.676 & 0.911 \\
    \bottomrule
    \end{tabular}
\end{minipage}
\end{table}

Again, Table~\ref{tab:multi_accuracy_comparison} shows that STAD performs very well, almost matching the performance of the oracles unless anomalies are very long, i.e.~$\mu_L = 800$. This matches the findings in Section~\ref{simulation:lengthSize}. Unsurprisingly, the performance of all methods improves when the signal to noise ratio is higher, i.e.~$(a,b) = (2,5)$.

\section{Application}
\label{section:applications}
The utility of STAD on real data is illustrated by applying it to historical electricity export prices from the Octopus Agile export tariff for the East Midlands (EM) distribution network region in the United Kingdom\footnote{\footnotesize Energy data: \url{https://agile.octopushome.net/historical-data}}. The time series represents the electricity export price, $Y_t$, in pence per kilowatt-hour (p/kWh). We exported the period from April 17, 2025 (00:00 UTC) to May 26, 2025 (23:30 UTC), but we excluded weekends to maintain a daily seasonal pattern and hence to have enough repetitions of the seasonality. The recording frequency is 30 minutes, yielding to a total of $N = 1344$ data points. We chose this data because of its real-world dynamics, specifically its daily seasonality and a long-term trend.

We set $P = 48$, $C = 1$, $J = 50$, $Q \in \{1,\dots,10\}$, $B \in \{1,\dots,8\}$, $\lambda_{\operatorname{coll}} = 28$, and $\lambda_{\operatorname{point}} = 24$. Multiple of these values are increased in comparison to our default values as the underlying trend and seasonal patterns are rather complicated and the real data is also unlikely to perfectly fit our model. Furthermore, anomalies of interest seem to have a large magnitude, so the simulations in Section~\ref{Simu:penalty} suggests that a larger penalty is possible. For comparison we apply CAPA to the dataset with $\lambda_{\operatorname{coll}} = 200$, and $\lambda_{\operatorname{point}} = 150$.

% Figure
\begin{figure}[htbp]
    \centering % Center the figure environment
    \begin{subfigure}[b]{0.48\textwidth}
        \centering
        \includegraphics[width=\linewidth]{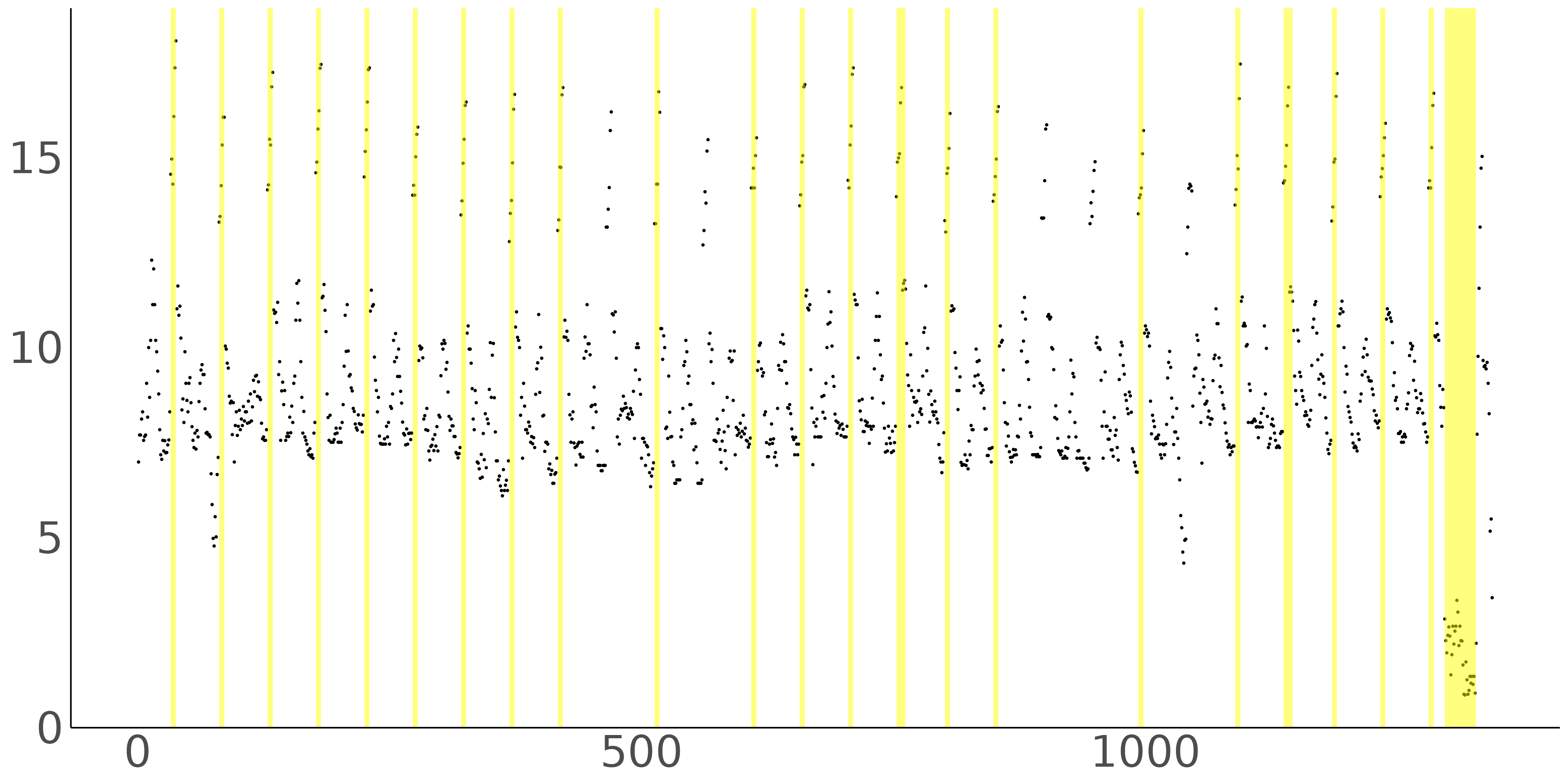} 
        \caption{CAPA}
        \label{fig:capa_app_results}
    \end{subfigure}
    \hfill
    \begin{subfigure}[b]{0.48\textwidth} % Adjust width as needed
        \centering
        \includegraphics[width=\linewidth]{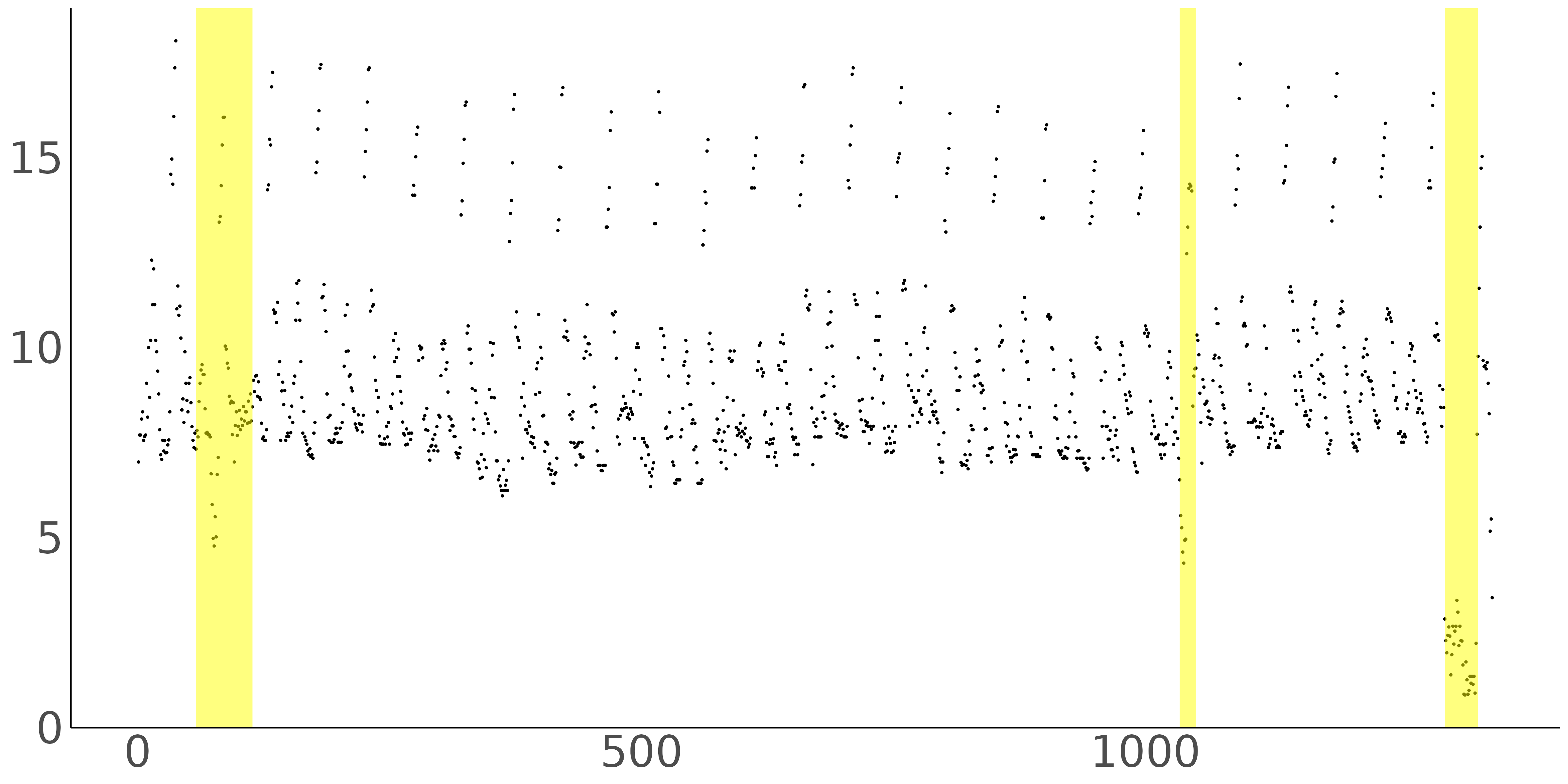}
        \caption{STAD}
        \label{fig:STAD_app_results}
    \end{subfigure}
    \caption{\footnotesize Comparison of CAPA and STAD applied to the Octopus Agile East Midlands export price data (April 17, 2025 -- May 26, 2025). Despite the extremely large penalty CAPA detects a large number of anomalies, mostly because of the seasonal pattern. It is unlikely that these findings are meaningful which illustrates the demand for a method like STAD that allows for a varying baseline. In comparison STAD detects three anomalies that appear reasonable.}
    \label{fig:anomaly_results_energy_app}
\end{figure}

Figure~\ref{fig:anomaly_results_energy_app} demonstrates that CAPA, even with an extremely large penalty, detects anomalies because of the underlying seasonality. In comparison STAD identifies three significant collective anomalies which appear to the naked eye as atypical behavior. To validate our findings, we compared the price data to renewable energy generation from the UK National Grid\footnote{\footnotesize UK Generation Data: National Grid ESO Data Portal, \url{https://data.nationalgrideso.com/}}.

\begin{figure}[htbp]
    \centering   \includegraphics[width=0.7\textwidth]{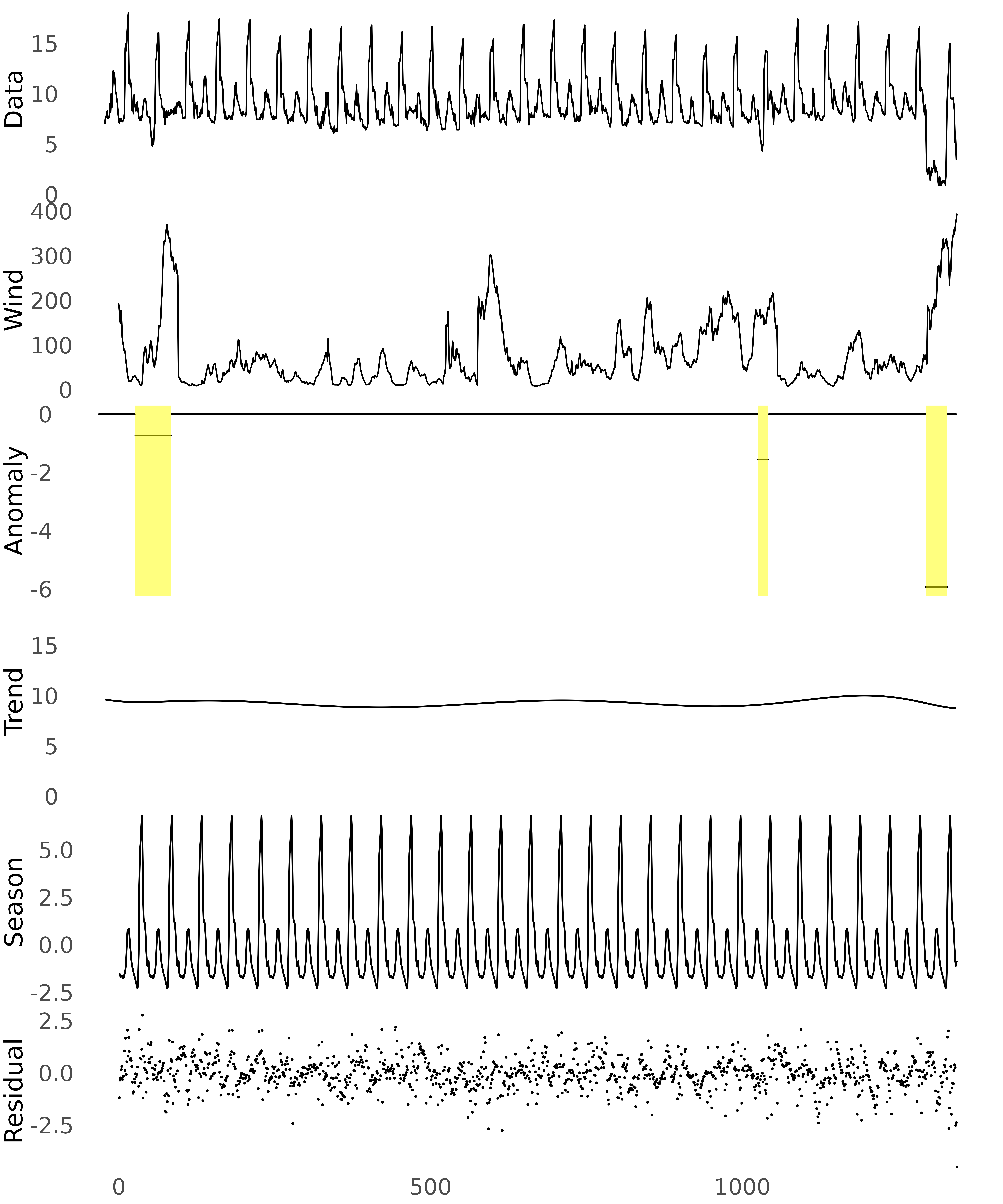}
    \caption{
    \footnotesize Decomposition of the Octopus Agile electricity price data by STAD as illustrated in Figure~\ref{fig:STAD_decompose}. It also displays the energy generation using wind. Two of the tree anomalies matches regions of high energy generation.
}
\label{fig:STAD_electricity_decompose}
\end{figure}

The obtained anomaly, trend, and seasonality in Figure~\ref{fig:STAD_electricity_decompose} appears to be reasonable. We also see that at least two anomalies matches times with extremely large energy generation using wind.

\section{Discussion}

This article introduces STAD. The approach overcomes key practical limitations of existing statistical methods for detecting collective anomalies, which either assume a constant baseline or fail to handle collective anomalies effectively. Both the theoretical results and the numerical experiments undertaken demonstrate that our approach successfully addresses this structural challenge and performs well across a wide range of settings. At the same time, the methodology remains restricted to polynomial trends and stationary seasonal components. Developing methods capable of handling smoothly varying functions and accommodating amplitude and phase modulation in the seasonal component are interesting avenues for future research.

\newpage

\appendix

\section{Pseudocodes}\label{sec:pseudocodes}

\begin{algorithm}[htbp]
\caption{The \textbf{S}easonal \textbf{T}rend \textbf{A}nomaly \textbf{D}etection algorithm (\textbf{STAD}), see Section~\ref{section:method}.}
\label{alg:stad}
\begin{algorithmic}[1]
\Require Time series $\mathbf{Y} = (y_1,\ldots,y_n) \in \mathbb{R}^n$, difference $D = C P$ and other tuning parameters required by \Call{estimateTrend}{}, \Call{estimateSeason}{}, and \Call{estimateAnomaly}{} as described in Algorithms~\ref{alg:estimate_trend},~\ref{alg:estimate_season},~and~\ref{alg:estimate_anomaly}.
\Ensure Estimated anomalies $\hat{\mathcal{A}}$, anomaly component $\hat{A} \in \mathbb{R}^n$, trend component $\hat{T} \in \mathbb{R}^n$, and seasonal component $\hat{S} \in \mathbb{R}^n$.
\Function{STAD}{$Y, D, \mathbf{Q}, \mathbf{B}, J, P$}
\State $\hat{\sigma}_0 \gets \operatorname{IQR}\big((y_2 - y_1,\ldots, y_n - y_{n - 1})\big) / \operatorname{IQR}_{\Phi} / \sqrt{2}$ \eqref{eq:sigma0_iqr}
\State $Y_{\operatorname{diff}} \gets ( y_{D + 1} - y_{1},\ldots,y_{n} - y_{n - D})$

\State $\hat{\beta}_D \gets  \Call{estimateTrend}{Y_{\operatorname{diff}}, \mathbf{Q} - 1, \mathbf{B}, J}$

\State $\hat{S}_0 \gets \Call{estimateSeason}{Y - X (0, \hat{\beta}_D)^\top, P, \{1,\ldots,n\}}$

\State $\hat{T}, \hat{\mathcal{A}}_0 \gets \Call{estimateTrend}{Y - \hat{S}_0, \mathbf{Q}, \mathbf{B}, J}$

\State $\hat{S} \gets \Call{estimateSeason}{Y - \hat{T}, P, \{1,\ldots,n\} \setminus \hat{\mathcal{A}}_0}$

\State $\hat{\mathcal{A}}, \hat{A} \gets \Call{estimateAnomaly}{Y - \hat{T} - \hat{S}}$

\State \Return $\hat{\mathcal{A}}, \hat{A}, \hat{T}, \hat{S}$
\EndFunction
\end{algorithmic}
\end{algorithm}

\begin{algorithm}[htbp]
\caption{The robust trend estimator in Section~\ref{section:trend_estimation}.}
\label{alg:estimate_trend}
\begin{algorithmic}[1]
    \Require Time series $\mathbf{Y} \in \mathbb{R}^n$, tuning parameter sets $\mathbf{Q}$ and $\mathbf{B}$, number of repetitions $J$, and parameters for the cost $C$, see \Call{estimateAnomaly}{} in Algorithm~\ref{alg:estimate_anomaly}.
    \Ensure Estimated trend component $\hat{T} \in \mathbb{R}^n$ and set of anomalous points $\tilde{\mathcal{A}}$.
    \Function{estimateTrend}{$\mathbf{Y}, \mathbf{Q}, \mathbf{B}, J$}
        \For {$Q' \in \mathbf{Q}$, $B' \in \mathbf{B}$, $j \in \{1,\ldots,J\}$}
            \State $l \sim \operatorname{Bern}\big(\frac{1}{2}\big)$
            \State $m_k \sim \operatorname{Uniform}\big(\{1,\ldots,B\}\big)$
            \State $U \gets \bigcup_{k = 0}^{Q'} \tilde{u}_{(2k + l) B' + m_k}$ \eqref{eq:final_sample}
            \State $\hat{T}_{B', Q', j} \gets X \hat{\beta}(U)$ with $\hat{\beta}$ as in \eqref{eq:m-estimator}
        \EndFor
        \State $\hat{T} \gets \argmin_{\hat{T}_{B', Q', j}} \min_{(\tilde{\mathcal{A}}, \tilde{\mu})} C(Y - T, \tilde{\mathcal{A}}, \tilde{\mu}) + Q' \hat{\sigma}_0 \log(n)$ \eqref{eq:trend_estimator}           
        \State $\tilde{\mathcal{A}} \gets$ anomalous points in the minimum above.              
        \State \Return $\hat{T}, \tilde{\mathcal{A}}$
    \EndFunction
\end{algorithmic}
\end{algorithm}

\begin{algorithm}[htbp]
    \caption{The robust seasonality estimator in Section~\ref{section:season}.}
    \label{alg:estimate_season}
    \begin{algorithmic}[1]
        \Require Time series $\mathbf{Y} \in \mathbb{R}^n$, seasonality period $P$, index set $\mathcal{I}$
        \Ensure Estimated seasonal component $\hat{S} \in \mathbb{R}^n$.
        \Function{estimateSeason}{$\mathbf{Y}, P, \mathcal{I}$}
            \For{$p \gets 0 \text{ to } P-1$}
                \State $V_p \gets \big\{ i \in \mathcal{I} : (i \bmod P) = p\big\}$ \eqref{eq:setSeasonalEstimation}
                \State $\hat{\vartheta}_p \gets \hat{\beta}(V_p)$ \eqref{eq:season_m_estimator}; with $\hat{\beta}$ as in \eqref{eq:m-estimator}
            \EndFor

            \If{seasonality is smooth}
                \State $\mathbf{\hat{\vartheta}} \gets \text{SmoothSpline}(\mathbf{\hat{\vartheta}})$
                \Comment{see Section~\ref{sec:smoothing}}
            \EndIf
            
            \State $\tilde{s}_i \coloneq \hat{\vartheta}_{i \bmod P}$ and $\hat{s}_i \coloneq \tilde{s}_i - n^{-1} \sum_{i = 1}^{n} \tilde{s}_{i}$ \eqref{eq:estimateSeasonality}
            \State \Return $\hat{\mathbf{S}} = (\hat{s}_1,\ldots,\hat{s_n})$ 
        \EndFunction
    \end{algorithmic}
\end{algorithm}

\begin{algorithm}[htbp]
    \caption{The anomaly estimator in Section~\ref{section:anomaly_detection}.} % Renamed caption
    \label{alg:estimate_anomaly} % Renamed label
    \begin{algorithmic}[1]
            \Require Time series $\mathbf{Y} \in \mathbb{R}^n$ and standard deviation and penalties for cost function $C$.
            \Ensure Estimated anomalies $\hat{\mathcal{A}}$ and anomaly component $\hat{A} \in \mathbb{R}^n$.
        \Function{estimateAnomaly}{$\mathbf{Y}$}    
            \State $(\hat{\mathcal{A}}, \hat{\mu}) \gets \argmin_{(\tilde{\mathcal{A}}, \tilde{\mu})} C(\mathbf{Y}, \tilde{\mathcal{A}}, \tilde{\mu})$ \eqref{eq:capa_estimator}
            \State $\hat{A} \gets (\hat{a}_1,\ldots,\hat{a}_n)$ as in \eqref{eq:estimateAnomalyComponent}
            \State \Return $\hat{\mathcal{A}}$ and $\hat{A}$
        \EndFunction
    \end{algorithmic}
\end{algorithm}

\section{Illustration of the tuning parameter \texorpdfstring{$B$}{B}}\label{simu:B}

In Scenario 1 we consider a simple quadratic trend and a long anomaly. If we choose $B = 1$, then regardless of how large $J$ is or whether $Q$ is correctly specified, the resulting estimate will be poorly as inevitably the second sub-segment overlaps with the anomaly. If we choose $B = 3$ instead, then some of the samples miss the anomaly and provide a good estimate.

\begin{figure}[htbp]
    \centering
            \begin{subfigure}[b]{0.7\textwidth}
        \includegraphics[width=\textwidth]{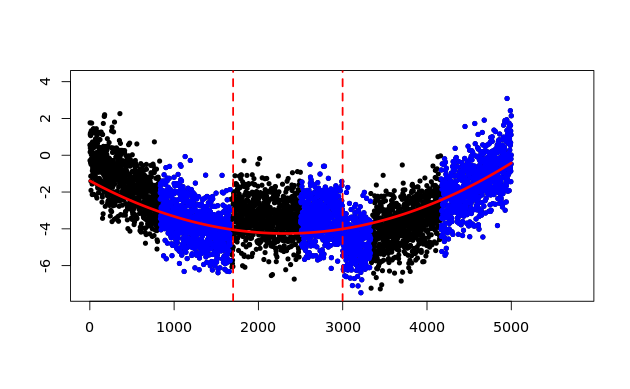}
        \caption{$B = 1$}
        \label{fig:smallB_sample1}
    \end{subfigure}
    \begin{subfigure}[b]{0.7\textwidth}
        \includegraphics[width=\textwidth]{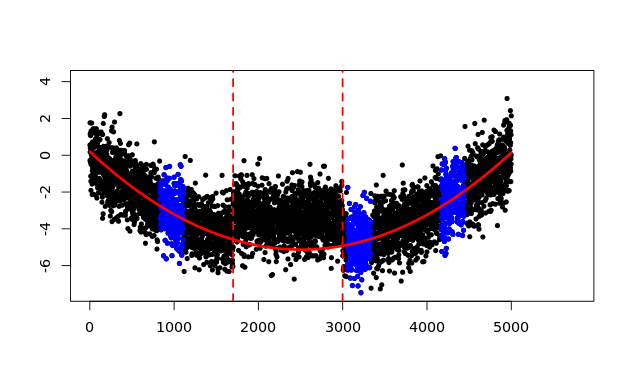}
        \caption{$B = 3$)}
        \label{fig:goodB_sample}
    \end{subfigure}
    \caption{A simple quadratic trend and a long anomaly. A too small $B$ means that the second sub-segment always overlaps with the anomaly, resulting in poor estimates of the trend.}
    \label{fig:goodB}
\end{figure}

In Scenario 2 we consider a more complicated cubic trend and a short anomaly. $B = 3$ shows once again good results. If we choose $B$ very large, e.g.~$B = 15$, then even when the sample $U$ is non-anomalous it does not lead to a good estimate, since the available number of observations in $U$ is too small and the estimate is extremely noisy.

\begin{figure}[htbp]
    \centering
    \begin{subfigure}[b]{0.7\textwidth}
        \includegraphics[width=\textwidth]{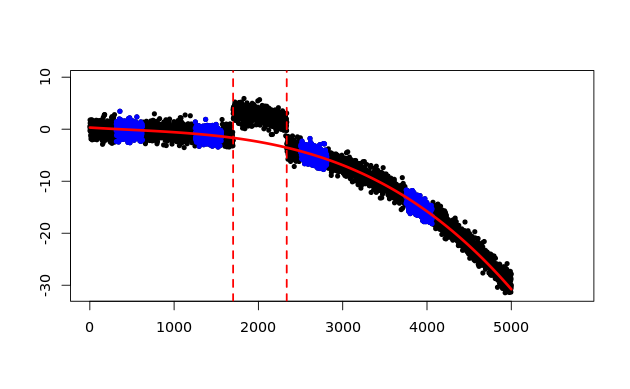}
        \caption{$B = 3$}
        \label{fig:goodB_sample_2}
    \end{subfigure}
        \begin{subfigure}[b]{0.7\textwidth}
        \includegraphics[width=\textwidth]{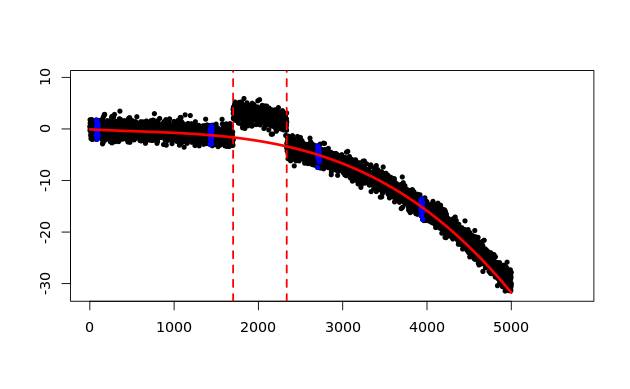}
        \caption{$B = 15$}
        \label{fig:smallB_sample1_2}
    \end{subfigure}
    \caption{A more complicated cubic trend and a short anomaly. A too large $B$ leads to a very small sample $U$ and hence to a bad, noisy estimate of the trend.}
    \label{fig:goodB_2}
\end{figure}

\section{The Iteratively Reweighted Least Squares}\label{IRLS}
For completeness, we detail the Iteratively Reweighted Least Squares (IRLS) algorithm that we use to compute the $M$-estimator.

\begin{enumerate}
    \item Start at $t = 0$ with an initial estimate $\hat{\beta}^{(0)}$. We use ordinary least squares for the tend estimator and the median for the seasonal estimate.

    \item Compute the remainder 
    \[
        e_i^{(t)} = y_i - x_i^\top \hat{\beta}^{(t)},
    \]
    and the corresponding weights
    \[
        w_i^{(t)} = \frac{\psi(e_i^{(t)})}{e_i^{(t)}},
    \]
    where $\psi$ denotes the influence function, i.e.~the derivative of the loss function $\rho$.

    \item Update the estimate $\hat{\beta}$ using weighted least squares regression
    \[
        \hat{\beta}^{(t+1)} = \left(X^\top W^{(t)} X\right)^{-1} X^\top W^{(t)} y,
    \]
    where $W^{(t)} = \text{diag}(w_1^{(t)}, w_2^{(t)}, \dots, w_n^{(t)})$.

    \item Increase $t$ by $1$ and repeat steps 2 and 3 until the estimated coefficients converge, i.e.
    \[
        \sum_j \left| \hat{\beta}^{(t+1)}_j - \hat{\beta}^{(t)}_j \right| < 10^{-5},
    \]
    or a maximum of $500$ iterations is reached.
\end{enumerate}

\section{Proofs}\label{sec:proofs}

We will show in Section~\ref{section:proofVariance} that $\CvarianceEstimateLower \sigma_0 \leq \hat{\sigma}_0 \leq \CvarianceEstimate \sigma_0$ for constants $0 < \CvarianceEstimateLower < 1$ and $\CvarianceEstimate > 1$ with probability converging to one. In Section~\ref{sec:proofDifference} we will study the effects of taking differences of lag $D$. In Sections~\ref{sec:theoryTrend},~\ref{sec:proofSeasonality},~and~\ref{sec:proofAnomaly} we will show that we estimate the trend, the seasonality, and the anomalies well given achievable assumptions. Finally, Theorem~\ref{theorem:main} will be proven in Section~\ref{sec:proofMainTheorem}.

\subsection{Additional notations}
Throughout the proof, we make use of the following notations. For any observation vector $Z = (z_1,\ldots,z_n)$ and index set $U \coloneq \{i_1,\ldots,i_m\} \subset \{1,\ldots,n\}$ we use the definitions $Z_{U} \coloneq (z_{i_1},\ldots,z_{i_m})$ and $\overline{Z}_U \coloneq \frac{1}{m} \sum_{i \in U} z_{i}$. For the regression matrix $X$ we denote by $X_U$ the matrix restricted to the rows in $U$, i.e.~$X_U = (x_{i, j})_{i \in U, j = 1,\ldots,Q + 1}$. For indices $1 \leq i \leq j \leq n$ we also use $i:j$ for the index set $\{i,\ldots,j\}$, e.g. we write  $Z_{i:j}$ instead of $Z_{\{i,\ldots,j\}}$. Furthermore, we denote by $\hat{\theta}_{i:j}$ the best polynomial fit of degree $Q$ to $Y_{i:j}$, i.e.
\begin{align*}
\hat{\beta}_{i:j} \coloneq & \argmin_{\beta \in \R^{Q + 1}} C_{i,j}(Y - X_{i:j} \beta),\\
\hat{\theta}_{i:j} \coloneq & X_{i:j} \hat{\beta}_{i:j} = X_{i:j} (X_{i:j}^\top X_{i:j})^{-1} X_{i:j}^\top Y_{i:j}.
\end{align*}

\subsection{Bound on the variance estimate}\label{section:proofVariance}
The following proposition bounds the variance estimate.
\begin{Proposition}\label{proposition:varianceEstimate}
We assume the model from Section~\ref{sec:model}. Let $\hat{\sigma}_0^2$ be the variance estimator as defined in \eqref{eq:sigma0_iqr}. Suppose Assumption~\ref{assumption:trendSeasonalBound} holds. Then, there exists constants $0 < \CvarianceEstimateLower \leq 1$ and $\CvarianceEstimate \geq 1$ such that
\begin{equation}\label{eq:goodVarianceEstimate}
\Pj\big(\CvarianceEstimateLower \sigma_0 \leq \hat{\sigma}_0 \leq \CvarianceEstimate \sigma_0\big) \to 1, \text{ as } n \to \infty.
\end{equation}
\end{Proposition}
\begin{proof}
Let $\Delta y$ be the first order difference of the observations as defined before \eqref{eq:sigma0_iqr}, let $\operatorname{IQR}$ be the interquartile range, let $\operatorname{IQR}_{\Phi}$ be the interquartile range of a standard Gaussian distribution, let $q_{0.75}$ and $q_{0.25}$ be the 75\% and 25\% quantile, respectively. Then, for a constant $\tilde{C}_1 >0$ to be chosen later, the definition of the interquartile range and a union bound yield
\begin{align*}
\Pj \left( \operatorname{IQR}\big( \Delta y \big) > 2 \tilde{C}_1 \sigma_0 \right)
\leq \Pj \left( q_{0.75}\big( \Delta y \big) > \tilde{C}_1 \sigma_0 \right) + \Pj \left( q_{0.25}\big( \Delta y \big) < - \tilde{C}_1 \sigma_0 \right).
\end{align*}
In the following we focus on bounding the first probability, but the bounding the second probability follows similar steps. It follows from the definition of a quantile for the first inequality, from Assumption~\ref{assumption:trendSeasonalBound} for the second inequality, and from the fact that there are at most $2K$ different $i$'s for which $a_i - a_{i - 1} \neq 0$, since this can only occur at the beginning or end of an anomaly, for the third inequality that
\begin{align*}
 \Pj \left( q_{0.75}\big( \Delta y \big) > \tilde{C}_1 \sigma_0 \right)
\leq & \Pj\left( \sum_{i = 2}^{n} \EINS_{y_i - y_{i - 1}  > \tilde{C}_1 \sigma_0} \geq \frac{1}{4} (n - 1) \right)\\
\leq & \Pj\left( \sum_{i = 2}^{n} \EINS_{\varepsilon_i - \varepsilon_{i - 1} + a_i - a_{i - 1} > \tilde{C}_1 \sigma_0 - 2 \CtrendSeasonalBound \sigma_0} \geq \frac{1}{4} (n - 1) \right)\\
\leq & \Pj\left( \sum_{i = 2}^{n} \EINS_{\varepsilon_i - \varepsilon_{i - 1} > \tilde{C}_1 \sigma_0 - 2 \CtrendSeasonalBound \sigma_0} \geq \frac{1}{4} (n - 1) - 2 K \right).
\end{align*}
We choose $\tilde{C}_1$ such that $\frac{\tilde{C}_1 - 2 \CtrendSeasonalBound}{\sqrt{2}}$ is the $80\%$ quantile of a standard Gaussian distribution. Then, the left hand side in the probability above is binomial distributed with size $n - 1$ and success probability 
\begin{equation*}
\Pj\left(\frac{\varepsilon_2 - \varepsilon_{1}}{\sqrt{2} \sigma_0} > \frac{\tilde{C}_1 - 2 \CtrendSeasonalBound}{\sqrt{2}} \right) \leq \frac{1}{5}.
\end{equation*}
We use the following bound from \citep[(3.5)]{feller1958introduction},
\begin{equation*}
\Pj(S_n \geq r) \leq \frac{r (1- p)}{(r - np)^2},
\end{equation*}
if $r \geq np$, where $S_n$ is Binomial distributed with size $n$ and success probability $p$. It follows that
\begin{align*}
\Pj \left( q_{0.75}\big( \Delta y \big) > \tilde{C}_1 \sigma_0 \right)
\leq \frac{n}{\big(\frac{1}{4} (n - 1) - 2 K - \frac{1}{5}(n - 1)\big)^2}
\to 0,
\end{align*}
as $n \to \infty$. This completes the proof with $\CvarianceEstimate = 2 \tilde{C}_1 \sqrt{2} \operatorname{IQR}_{\Phi}$, as the second probability can be bounded following similar arguments. Also the proof of $\CvarianceEstimateLower \sigma_0 \leq \hat{\sigma}_0$ follows similar arguments.
\end{proof}

\subsection{Differenced sequence}\label{sec:proofDifference}

We detail in Lemma~\ref{lemma:differencedSequence} the properties of the differenced sequence $Y_D$ and show that it satisfies the following assumptions, which are required for a good estimation of the trend.

\begin{Assumption}\label{assumption:segmentLengthTrend}
There exists a $\tilde{\delta} > 0$ such that
\begin{equation*}
e_k - b_k \geq \frac{\sigma_0^2}{\mu_k}\log(n)^{1 + \delta + \tilde{\delta}},\ \forall\ k = 1,\ldots,K,
\end{equation*}
and
\begin{equation*}
b_{k + 1} - e_k \geq \frac{\sigma_0^2}{\mu_k}\log(n)^{1 + \delta + \tilde{\delta}},\ \forall\ k = 0,\ldots,K,
\end{equation*}
\end{Assumption}

\begin{Assumption}\label{assumption:anomalyLengthTrend}
The total length of collective anomalies is bounded. More precisely, there exists a $\gamma < 1$ such that
\begin{equation*}
\sum_{k = 1}^{K} e_k - b_k < \gamma \frac{1}{2 (Q + 1)}.
\end{equation*}
\end{Assumption}

\begin{Lemma}\label{lemma:differencedSequence}
We assume the model from Section~\ref{sec:model}. Suppose Assumptions~\ref{assumption:segmentLength},~\ref{assumption:anomalyLengthCombined},~and~\ref{assumption:gridBandQ} hold. Then,
\begin{equation*}
Y_D = (y_{1,D},\ldots,y_{D, n - D}) = A_D + T_D + \varepsilon_D,
\end{equation*}
where
\begin{equation*}
T_D = (t_{D,1},\ldots,t_{D, n - D}) = X_D \beta_D,
\end{equation*}
with $\beta_D = (\beta_1,\ldots,\beta_Q)$ and $(X_D)_{i, q} = \left(\frac{i+D}{n}\right)^q - \left( \frac{i}{n} \right)^q $, is a polynomial of degree $Q - 1$. Furthermore,
$A_D = (a_{D,1},\ldots,a_{D,n - D})$ contains $2K$ anomalies that satisfy Assumptions~\ref{assumption:gridBandQ},~\ref{assumption:segmentLengthTrend},~and~\ref{assumption:anomalyLengthTrend}. Finally, 
\begin{equation*}
\varepsilon_D = (\varepsilon_{D, 1},\ldots,\varepsilon_{D, n - D}) = (\varepsilon_{D + 1} - \varepsilon_1,\ldots,\varepsilon_{n} - \varepsilon_{n - D})
\end{equation*}
is a multivariate Gaussian distributed with expectation $\E\left[ \varepsilon_{D, i} \right] = 0$ and covariance
\begin{equation*}
\Cov\left[ \varepsilon_{D, i}, \varepsilon_{D, j} \right] = \begin{cases}
2\sigma_0^2 & \text{if } i = j,\\
-\sigma_0^2 & \text{if } \vert i - j \vert = D,\\
0 & \text{otherwise.}
\end{cases}
\end{equation*} 
\end{Lemma}
\begin{proof}[Proof of Lemma~\ref{lemma:differencedSequence}]
We obtain
\begin{equation*}
y_{D, i} = y_{i + D} - y_i = (a_{i + D} - a_{i}) + (t_{i + D} - t_{i}) + (s_{i + D} - s_{i}) + (\varepsilon_{i + D} - \varepsilon_{i}).
\end{equation*}

In the following we look at each component separately. We start with the seasonal component. Since $D = c P$, we have that
\begin{equation*}
s_{i + D} - s_i = s_i - s_i = 0.
\end{equation*}
This shows that $Y_D$ has no seasonal component.

Next, we look at the trend component. We obtain
\begin{align*}
t_{D, i}
\coloneq t_{i + D} - t_i
= & \sum_{q = 0}^{Q} \beta_q X_{i + D, q} - \sum_{q = 0}^{Q} \beta_q X_{i, q}\\
= & \sum_{q = 0}^{Q} \beta_q \big(X_{i + D, q} - X_{i, q}\big)
\eqcolon \sum_{q = 0}^{Q} \beta_q \big((X_D)_{i, q}\big).
\end{align*}
It remains to show that $\beta_0$ disappears and that it is a polynomial of degree $Q - 1$. The difference of powers identity yields
\begin{align*}
\sum_{q = 0}^{Q} \beta_q \big(X_{i + D, q} - X_{i, q}\big)
= & \sum_{q = 0}^{Q} \beta_q \left[\left(\frac{i+D}{n}\right)^q - \left( \frac{i}{n} \right)^q  \right]\\
= & \sum_{q = 1}^{Q} \beta_q \left[\left(\frac{i+D}{n}\right)^q - \left( \frac{i}{n} \right)^q  \right]\\
= & \sum_{q = 1}^{Q} \beta_q \left[ \frac{D}{n} \sum_{r = 0}^{q-1} \left( \frac{i}{n} \right)^r \left( \frac{i+D}{n} \right)^{q-1-r} \right].
\end{align*}
We conclude that $T_D$ is a polynomial of degree $Q - 1$, with regression matrix $X_D$, and parameter vector $\beta_D = (\beta_1,\ldots,\beta_Q)$.

Next we consider the anomaly component. We focus for now on the contribution of the $k$-th component. The final anomaly term will be the sum over all components. Assumption~\ref{assumption:segmentLength} yields $D > e_k - b_k$ and hence $e_k - D < b_k$. We distinguish a few cases.
\begin{equation*}
a_{D, i} \coloneq a_{i+D} - a_i = 
\begin{cases}
0 - 0 = 0, & \text{if } i \leq b_k - D, \\
\theta_k - 0 = \theta_k, & \text{if } b_k - D < i \leq e_k - D, \\
0 - 0 = 0, & \text{if } e_k - D < i \leq b_k, \\
0 - \theta_k = -\theta_k, & \text{if } b_k < i \leq e_k, \\
0 - 0 = 0, & \text{if } i > e_k.
\end{cases}
\end{equation*}
We obtain two anomalies of magnitude $\theta_k$ and of length $e_k - b_k$. Hence, Assumption~\ref{assumption:segmentLengthTrend} is satisfied if all anomalies are well separated. It follows from Assumption~\ref{assumption:segmentLength} that
\begin{equation*}
b_k - (e_k - D) = D - (e_k - b_k) \geq e_k - b_k,
\end{equation*}
and
\begin{equation*}
(b_{k + 1} - D) - (e_k + 1) = b_{k + 1} - e_k - D - 1 \geq \min\{e_k - b_k, e_{k + 1} - b_{k + 1}\},
\end{equation*}
This shows that the anomaly component contains $2K$ anomalies that satisfy Assumption~\ref{assumption:segmentLengthTrend}. Assumption~\ref{assumption:anomalyLengthTrend} follows from Assumption~\ref{assumption:anomalyLengthCombined}, since the total length of anomalies is $2\sum_{k = 1}^K e_k - b_k$. Assumption~\ref{assumption:gridBandQ} was directly formulated for the differenced sequence.

Finally, we consider the noise component. We have that $\varepsilon_D$ is multivariate Gaussian distributed as the difference of two multivariate Gaussian distributions. We further have that
\begin{equation*}
\E\left[ \varepsilon_{D, i} \right] = \E\left[ \varepsilon_{i + D} - \varepsilon_i \right] = \E\left[ \varepsilon_{i + D} \right] - \E\left[ \varepsilon_i \right] = 0 - 0 = 0
\end{equation*}
and
\begin{align*}
& \Cov\left[ \varepsilon_{D, i}, \varepsilon_{D, j} \right]\\
= & \Cov\left[ \varepsilon_{i + D}, \varepsilon_{j + D} \right] + \Cov\left[ \varepsilon_{i}, \varepsilon_{j} \right] - \Cov\left[ \varepsilon_{i + D}, \varepsilon_{j} \right] - \Cov\left[ \varepsilon_{i}, \varepsilon_{j + D} \right]\\
= & \begin{cases}
2\sigma_0^2 & \text{ if } i = j,\\
-\sigma_0^2 & \text{ if } \vert i - j \vert = D,\\
0 & \text{otherwise.}
\end{cases}
\end{align*}

\end{proof}

\subsection{Trend estimation}\label{sec:theoryTrend}

In this section we show that we can estimate the trend well. We first state the setting and then the proposition.

\begin{Setting}\label{setting:trend}
We assume that $Y = A + T + W + \varepsilon$, with trend $T = X \beta$ and deterministic remainder term $W = (w_1,\ldots,w_{\mn})$, where either $\mn = n$ or $\mn = n - D$. The remainder $W = (w_1,\ldots,w_{\mn})$ is a periodic vector with period length $P$ for which there exist a constant $\CremainderSeasonal > 0$ such that 
\begin{equation}\label{eq:remainderSeasonalyT}
\sum_{i = 1}^{\mn} w_i^2 \leq \CremainderSeasonal \sigma_0^2 \log(n).
\end{equation}
Let $\CvarianceEstimateLower \sigma_0 \leq \hat{\sigma}_0 \leq \CvarianceEstimate \sigma_0$ be deterministic. Suppose that Assumptions~\ref{assumption:gridBandQ},~\ref{assumption:segmentLengthTrend},~and~\ref{assumption:anomalyLengthTrend}  hold. We further assume that there are no point anomalies, i.e.~$O = \emptyset$. Finally, we assume that $Y, A, T, \varepsilon, X, \beta$ are either
\begin{enumerate}[(a)]
\item \label{setting:Trenda}
given as defined in Section~\ref{sec:model},
\item \label{setting:Trendb}
or replaced by the differenced versions $Y_D, T_D, A_D, \varepsilon_D, X_D, \beta_D$ as defined in Lemma~\ref{lemma:differencedSequence}.
\end{enumerate}
We have that the trend is a polynomial of degree either $\rQ = Q$ or $\rQ = Q - 1$.
\end{Setting}

\begin{Proposition}\label{proposition:goodFinalTrendEstimate}
Given Setting~\ref{setting:trend}\eqref{setting:Trenda} there exists a constant $\CfinalTrendEstimate > 0$ such that
\begin{equation*}
\Pj\big( (\hat{T} - T)^\top (\hat{T} - T) \leq \CfinalTrendEstimate \sigma_0^2 \log(n) \big) \to 1,
\end{equation*}
as $n \to \infty$.
\end{Proposition}

It will take multiple lemmas to prove this proposition. In Section~\ref{sec:proofNoAnomalies} we show that with large probability at least one of samples $U_{B, Q, j}$ is non-anomalous. Combined with the results in Section~\ref{sec:proofGoodTrend} this ensures that at least one of our trend estimates is good. In Section~\ref{sec:proofBoundNoise} we bound noise terms and then in Section~\ref{sec:proofBoundCostDifference} we show that a good trend estimate leads to lower CAPA cost then a bad trend estimate. This allows us to show that the selected trend estimate has lower error. As a side product we also show that $\hat{\mathcal{A}} \in \mathcal{B}$.

\subsubsection{Anomalous free subsamples}\label{sec:proofNoAnomalies}

To estimate the trend well, we require that at least one of the samples $U_{B, Q, j}$ is non-anomalous, c.f.~Definition~\ref{def:anomalous}. Given Setting~\ref{setting:trend} we show in Lemma~\ref{lemma:intersectionAnomaly} that at most $\frac{V}{2 (Q + 1)}$ of the sub-segments $\tilde{u}_1,\ldots,\tilde{u}_V$ are anomalous. Building on this, we show in Lemma~\ref{lemma:probGoodSample} that there exists a non-anomalous sample $U \in \mathcal{U}$ \eqref{eq:final_sample} with probability at least $1 - \left(\frac{1}{2}\right)^J$. We remark that the sub-segments $\tilde{u}_1,\ldots,\tilde{u}_V$ and the sampling procedure in Section~\ref{section:sampling} in general explicitly depends on $B$ and $Q$. For brevity we have omitted it in the definitions of the segments and sub-segments.

\begin{Lemma}\label{lemma:intersectionAnomaly}
Given Setting~\ref{setting:trend}, if $n$ is large enough, then there exists $B \in \mathbf{B}$ and $Q \in \mathbf{Q}$ such that at most $\frac{V}{2 (Q + 1)}$ of the sub-segments $\tilde{u}_1,\ldots,\tilde{u}_V$ \eqref{eq:segments_sample} are anomalous.
\end{Lemma}
\begin{proof}
Assumption~\ref{assumption:gridBandQ} ensures that the correct degree $Q \in \mathbf{Q}$ and that there exists a $B \in \mathbf{B}$ such that 
\begin{equation}\label{eq:conditionOnB}
V = 2(Q+1)B \geq \frac{2 K}{\frac{1}{2 (Q + 1)} - \frac{\sum_{k = 1}^{K} e_k - b_k}{\gamma n }},
\end{equation}
for a constant $0 < \gamma < 1$. Each sub-segment $\tilde{u}_k$ contains at least $\lfloor n / V \rfloor$ many points. An anomalous sub-segment is either fully part of a collective anomaly or at the boundary of at least one collective anomaly. At most
\begin{equation*}
\frac{\sum_{k = 1}^K e_k - b_k}{\lfloor n / V \rfloor} \eqcolon \frac{L_a}{\lfloor n / V \rfloor}
\end{equation*}
many sub-segments are fully part of a collective anomaly. Furthermore, at most $2 K$ many sub-segments are at the boundary of a collective anomaly. So, in total we have at most 
\begin{equation*}
\frac{L_a}{\lfloor n / V \rfloor} + 2K
\end{equation*}
anomalous sub-segments. We will bound this number in the following. Let $n$ be large enough such that $n \geq \frac{V}{1 - \gamma}$, and hence $n - V \geq \gamma n$. Then,
\begin{equation}\label{eq:bound_ndp}
\frac{1}{\lfloor \frac{n}{V} \rfloor} \leq \frac{1}{\frac{n}{V} - 1} = \frac{1}{n - V} V \leq \frac{V}{\gamma n}.
\end{equation}
Finally, \eqref{eq:conditionOnB}~and~\eqref{eq:bound_ndp} yield that the total number of anomalous sub-segments is upper bounded by
\begin{align*}
\frac{L_a}{\lfloor n / V \rfloor} + 2K
\leq  \frac{L_a}{\gamma n} V + \left( \frac{1}{2 (Q + 1)} - \frac{L_a}{\gamma n } \right) V 
=  \frac{V}{2 (Q + 1)}.
\end{align*}
\end{proof}

\begin{Lemma}\label{lemma:probGoodSample}
Given Setting~\ref{setting:trend}, if $n$ is large enough, then there exists $B \in \mathbf{B}$ and $Q \in \mathbf{Q}$ such that
\begin{equation*}
\Pj\Big( U_{B, Q, j} \cap \mathcal{A} = \emptyset \Big) \geq \frac{1}{2} \ \forall\ j = 1,\ldots,J.
\end{equation*}
Consequently,
\begin{equation*}
\Pj\Big(\exists\ U \in \mathcal{U}\, :\, U \cap \mathcal{A} = \emptyset \Big) \geq 1 - \left(\frac{1}{2}\right)^J.
\end{equation*}
\end{Lemma}
\begin{proof}
Assumption~\ref{assumption:gridBandQ} ensures that the correct degree $Q \in \mathbf{Q}$ and that there exists a $B \in \mathbf{B}$ such that \eqref{eq:conditionOnB} is satisfied. Let $N$ denote the total number of anomalous sub-segments. Let $N_k$ be the number of anomalous sub-segments within segment $u_k$. Then, Lemma~\ref{lemma:intersectionAnomaly} yields
\begin{align*}
\sum_{k = 1}^{2(Q + 1)} N_k = N \leq \frac{V}{2(Q + 1)}.
\end{align*}
By definition, $\frac{V}{2(Q + 1)} = B$ is also the number of sub-segments in each segment. Thus, the segment $u_k$ contains $B \geq N$ sub-segments of which $N_k$ are anomalous. We have that $l \sim \operatorname{Bern}\big(\frac{1}{2}\big)$ is a Bernoulli-distributed random variable and $m_k \sim \text{Uniform}\big(\{1, \ldots, B\}\big)$. All random variables are independent of each other. We will also use that $\prod_{i = 1}^{m} (1 - a_i) \geq 1 - \sum_{i = 1}^{m} a_i$ for any $a_1,\ldots,a_m \in [0,1]$. We obtain
\begin{align*}
& \Pj\Big( U_{B, Q, j} \cap \mathcal{A} = \emptyset \Big)\\
= & \Pj\Big( \bigcup_{k = 0}^{Q} \tilde{u}_{(2k + l) B + m_k} \cap \mathcal{A} = \emptyset \Big)\\
= & \Pj\Big( \tilde{u}_{(2k + l) B + m_k} \cap \mathcal{A} = \emptyset \ \forall\ k = 0,\ldots,Q \Big)\\
= & \Pj(l = 0)\prod_{k = 0}^{Q} \Pj\Big( \tilde{u}_{(2k) B + m_k} \cap \mathcal{A} = \emptyset\Big) + \Pj(l = 1)\prod_{k = 0}^{Q} \Pj\Big( \tilde{u}_{(2k + 1) B + m_k} \cap \mathcal{A} = \emptyset\Big)\\
= & \frac{1}{2}\prod_{k = 0}^{Q} \bigg(\sum_{i = 1}^{B} \Pj(m_k = i) \EINS_{\tilde{u}_{(2k) B + i} \cap \mathcal{A} = \emptyset}\bigg) + \frac{1}{2}\prod_{k = 0}^{Q}\bigg(\sum_{i = 1}^{B} \Pj(m_k = i) \EINS_{\tilde{u}_{(2k+1) B + i} \cap \mathcal{A} = \emptyset}\bigg)\\
= & \frac{1}{2}\prod_{k = 1}^{Q+1} \bigg(\big(B - N_{2k-1}\big) \frac{1}{B} \bigg) + \frac{1}{2}\prod_{k = 1}^{Q + 1} \bigg(\big(B - N_{2k}\big) \frac{1}{B} \bigg)\\
\geq &\frac{1}{2} \prod_{k = 1}^{Q + 1} \bigg(1 - \frac{N_{2k-1}}{N}\bigg) + \frac{1}{2} \prod_{k = 1}^{Q + 1} \bigg(1 - \frac{N_{2k}}{N}\bigg)\\
\geq & \frac{1}{2}\bigg( 1 -  \sum_{k = 1}^{Q + 1} \frac{N_{2k-1}}{N}\bigg) + \frac{1}{2}\bigg( 1 -  \sum_{k = 1}^{Q + 1} \frac{N_{2k}}{N}\bigg)\\
=& 1 - \frac{1}{2}\sum_{k = 1}^{2 (Q + 1)} \frac{N_{k}}{N} = \frac{1}{2}.
\end{align*}
This proves the first statement. For the second statement note that the samples $U \in \mathcal{U}$ are independent of each other. Therefore,
\begin{align*}
& \Pj\Big(\exists\ U \in \mathcal{U}\, :\, U \cap \mathcal{A} = \emptyset \Big)\\
\geq & 1 - \prod_{j = 1}^J \Pj\left( U_{B, Q, j} \cap \mathcal{A} \neq \emptyset\right)\\
\geq & 1 - \prod_{j = 1}^J \frac{1}{2} = 1 - \left(\frac{1}{2}\right)^J.
\end{align*}
This completes the proof.
\end{proof}

\subsubsection{Trend estimation using a single sample}\label{sec:proofGoodTrend}

In this subsection we show that we can estimate the trend well if we have a sample $U = U_{B,Q,j} \subset \{1,\ldots,n\}$ from the sampling procedure in Section~\ref{section:sampling} for which $U \cap \mathcal{A} = \emptyset$ holds. We estimate the trend by 
\begin{equation}\label{eq:trendLS}
\hat{T}_U = X \hat{\beta}_U = X (X_{U}^\top X_{U})^{-1} X_{U}^\top Y_{U}.
\end{equation}
We then show that the error $\hat{T}_U - T$ is small, i.e.~there exists a constant $\CerrorTrend > 0$ such that
\begin{equation}\label{eq:outlineGoodTrend}
(\hat{T}_U - T)^\top (\hat{T}_U - T) \leq \CerrorTrend \sigma_0^2 \log(n)
\end{equation}
holds with probability going to one. To this end, we lower bound in Lemma~\ref{lemma:determinant} the determinant of $X_U^\top X_U$. We use this to bound in Lemma~\ref{lemma:trace} the trace of a matrix related to the covariance matrix, before we finally show \eqref{eq:outlineGoodTrend} in Lemma~\ref{lemma:trendEstimation}.

\begin{Lemma}\label{lemma:determinant}
Let $U \coloneq U_{B,Q,j} \subset \{1,\ldots,n\}$ be a sample from the sampling procedure in Section~\ref{section:sampling}. Suppose that $X$ is given by Setting~\ref{setting:trend}. Then there exists a constant $\Cdeterminant > 0$ such that
\begin{equation*}
\operatorname{det}\left( X_U^\top X_U \right) \geq
\begin{cases}
\Cdeterminant n^{Q + 1} & \text{ if X is given by Setting~\ref{setting:trend}\eqref{setting:Trenda},}\\
\Cdeterminant \left(\frac{D^2}{n}\right)^Q & \text{ if X is given by Setting~\ref{setting:trend}\eqref{setting:Trendb}.}
\end{cases}
\end{equation*}
\end{Lemma}
We require multiple definitions and technical lemmas to show Lemma~\ref{lemma:determinant}. Note that adding a multiple of one row to another row does not change the determinant. We apply Gauss reduction to $X_U^\top X_U$ to obtain an upper right diagonal matrix. The matrix after $s \in \{1,\ldots,\rQ\}$ steps is defined recursively by
\begin{equation*}
M_{i,j}^{(s + 1)} \coloneq 
\begin{cases}
\frac{M_{i,j}^{(s)} M_{s,s}^{(s)} - M_{i,s}^{(s)}M_{s,j}^{(s)}}{M_{s, s}^{(s)}}, & \text{if } i,j \in \{s + 1,\ldots, \rQ + 1\},\\
0 & \text{if } i \in \{s + 1,\ldots, \rQ + 1\} \text{ and } j \leq s,\\
M_{i,j}^{(s)}, & \text{if } i \leq s,
\end{cases}
\end{equation*}
and
\begin{equation*}
M_{i,j}^{(1)} \coloneq \left[ X_U^\top X_U \right]_{i,j}, \quad i,j = 1,\ldots, \rQ + 1.
\end{equation*}
The next lemma simplifies this iterative process.
\begin{Lemma}\label{lemma:simplifyM}
Under the assumptions of Lemma~\ref{lemma:determinant} we have that
\begin{equation}\label{eq:relationMandm}
M_{i,j}^{(s)} = \frac{\sum_{u \in U^{2^{s - 1}}} m_{i}^{(s)}(u)\, m_{j}^{(s)}(u)}{\prod_{t = 1}^{s - 1} \left(2 M_{t, t}^{(t)}\right)^{2^{s - t - 1}}},
\end{equation}
where 
\begin{equation*}
m_{i}^{(1)}(u) \coloneq 
\begin{cases}
\left( \frac{u}{n} \right)^{i - 1} & \text{ if X is given by Setting~\ref{setting:trend}\eqref{setting:Trenda},}\\
\left( \frac{u + D}{n} \right)^{i} - \left( \frac{u}{n} \right)^{i} & \text{ if X is given by Setting~\ref{setting:trend}\eqref{setting:Trendb},}
\end{cases}
\end{equation*}
for any $u \in U$, and
\begin{equation}\label{eq:recursionm}
m_{i}^{(s + 1)}\big((u_0, u_1)\big) \coloneq  m_{i}^{(s)}(u_0)\, m_{s}^{(s)}(u_1) - m_{s}^{(s)}(u_0)\, m_{i}^{(s)}(u_1)
\end{equation}
for any $(u_0, u_1) \in U^{2^{s - 1}} \times U^{2^{s - 1}}$.
\end{Lemma}
\begin{proof}
We use induction. We exemplify the induction for Setting~\ref{setting:trend}\eqref{setting:Trenda}, but Setting~\ref{setting:trend}\eqref{setting:Trendb} follows the same steps. The induction beginning for $s = 1$ is true by definition,
\begin{equation*}
M_{i,j}^{(1)} = \sum_{u \in U} \left( \frac{u}{n} \right)^{i + j - 2} = \sum_{u \in U} \left( \frac{u}{n} \right)^{i - 1} \left( \frac{u}{n} \right)^{j - 1} = \sum_{u \in U} m_{i}^{(1)}(u) m_{j}^{(1)}(u).
\end{equation*}
Let the statement be true for $s$. Then,
\begin{align*}
& \left(\prod_{t = 1}^{s} \left(2 M_{t, t}^{(t)}\right)^{2^{s - t}}\right) M_{i,j}^{(s + 1)}\\
= & \left(\prod_{t = 1}^{s} \left(2 M_{t, t}^{(t)}\right)^{2^{s - t}}\right) \left( M_{i,j}^{(s)} M_{s,s}^{(s)} - M_{i,s}^{(s)}M_{s,j}^{(s)} \right) \left( M_{s,s}^{(s)} \right)^{-1}\\
= & 2\Bigg( \sum_{u \in U^{2^{s - 1}}} m_{i}^{(s)}(u)m_{j}^{(s)}(u) \Bigg)\Bigg( \sum_{u \in U^{2^{s - 1}}} m_{s}^{(s)}(u)m_{s}^{(s)}(u) \Bigg)\\
& - 2\Bigg( \sum_{u \in U^{2^{s - 1}}} m_{i}^{(s)}(u)m_{s}^{(s)}(u) \Bigg) \Bigg( \sum_{u \in U^{2^{s - 1}}} m_{s}^{(s)}(u)m_{j}^{(s)}(u) \Bigg)\\
= & 2\sum_{(u_0, u_1) \in U^{2^{s - 1}} \times U^{2^{s - 1}}} m_{i}^{(s)}(u_0)m_{j}^{(s)}(u_0) m_{s}^{(s)}(u_1)m_{s}^{(s)}(u_1)\\
& - 2\sum_{(u_0, u_1) \in U^{2^{s - 1}} \times U^{2^{s - 1}}} m_{i}^{(s)}(u_0)m_{s}^{(s)}(u_0) m_{s}^{(s)}(u_1)m_{j}^{(s)}(u_1)\\
= & \sum_{(u_0, u_1) \in U^{2^{s - 1}} \times U^{2^{s - 1}}} \left( m_{i}^{(s)}(u_0)m_{s}^{(s)}(u_1) - m_{s}^{(s)}(u_0)m_{i}^{(s)}(u_1) \right) \cdot \\
& \hspace*{90pt} \left( m_{j}^{(s)}(u_0)m_{s}^{(s)}(u_1) - m_{s}^{(s)}(u_0)m_{j}^{(s)}(u_1) \right)\\
= & \sum_{(u_0, u_1) \in U^{2^{s - 1}} \times U^{2^{s - 1}}} m_{i}^{(s + 1)}\big((u_0, u_1)\big)m_{j}^{(s + 1)}\big((u_0, u_1)\big).
\end{align*}
For the second last equality note that summands with $u_0 = u_1$ are $0$. This proves the statement.
\end{proof}
By construction of $U$ \eqref{eq:final_sample} there exists $w_1 < \cdots < w_{\rQ + 1} \in \{1,\ldots,\mn\}$ such that $U = \bigcup_{j = 1}^{\rQ + 1} \{w_j + 1,\ldots,w_j + \tilde{L}_j\}$ with $L - 2 \leq \tilde{L}_j \leq L$ for all $j$ and $L$, where there exists a constant $\ClengthL > 0$ such that $L = \ClengthL n$. Let $W = \{w_1,\ldots,w_{\rQ + 1}\}$. The following lemma bounds the difference between using $U$ and $W$.
\begin{Lemma}\label{lemma:differencemterms}
Under the assumptions of Lemma~\ref{lemma:determinant}, for any $i,s \in \mathbb{N}$, $w \in W^{2^{s - 1}}$, $l \in \{1,\ldots,L\}^{2^{s - 1}}$  there exists a constant $C_{i,s} > 0$ only depending on $s$ and $i$ such that
\begin{equation*}
\big \vert m_i^{(s)}(w + l) - m_i^{(s)}(w) \big\vert \leq 
\begin{cases}
C_{i,s} \frac{L}{n} & \text{ if X is given by Setting~\ref{setting:trend}\eqref{setting:Trenda},}\\
C_{i,s} \left(\frac{D}{n}\right)^{2^{s - 1}}\frac{L}{n} & \text{ if X is given by Setting~\ref{setting:trend}\eqref{setting:Trendb}.}
\end{cases}
\end{equation*}
\end{Lemma}
\begin{proof}
Let $X$ be given by Setting~\ref{setting:trend}\eqref{setting:Trenda}. We use induction to show the statement. For any $w \in W$ and $l \in \{1,\ldots,L\}$, the difference of powers identity yield
\begin{align*}
m_i^{(1)}(w + l) - m_i^{(1)}(w)
= & \left(\frac{w + l}{n}\right)^{i - 1} - \left(\frac{w}{n}\right)^{i - 1}\\
= & \frac{l}{n} \sum_{j = 0}^{i - 2} \left(\frac{w + l}{n}\right)^{i - 2 - j} \left(\frac{w}{n}\right)^{j}
\in \left(0, (i - 1) \frac{L}{n}\right].
\end{align*}
Hence, the statements holds for $s= 1$ with $C_{i, 1} = i - 1$.

Suppose the statement is true for $s$. Note that it follows from \eqref{eq:recursionm} and $m_i^{(1)}(u) = \left(\frac{u}{n}\right)^{i - 1} \leq 1$ that here exists a constant $\tilde{C}_1 > 0$ such that $\big \vert m_s^{(s)}(u) \big\vert  \leq \tilde{C}_1$ for all $u \in U^{2^{s - 1}}$. Thus, for any $w, \tilde{w} \in W^{2^{s - 1}}$ and $l,\tilde{l}\in \{1,\ldots,L\}^{2^{s - 1}}$,
\begin{align*}
& m_i^{(s + 1)}\big((w + l, \tilde{w} + \tilde{l})\big) - m_i^{(s + 1)}\big((w, \tilde{w})\big)\\
= & m_i^{(s)}(w + l) m_s^{(s)}(\tilde{w} + \tilde{l}) - m_i^{(s)}(w) m_s^{(s)}(\tilde{w})\\
& - \big( m_s^{(s)}(w + l) m_i^{(s)}(\tilde{w} + \tilde{l}) - m_s^{(s)}(w) m_i^{(s)}(\tilde{w}) \big)\\
\in & \Big(m_i^{(s)}(w) \pm C_{i, s} \frac{L}{n}\Big) \Big(m_s^{(s)}(\tilde{w}) \pm C_{s, s}\frac{L}{n}\Big) - m_i^{(s)}(w) m_s^{(s)}(\tilde{w})\\
& \ - \bigg(  \Big(m_s^{(s)}(w) \pm C_{s, s} \frac{L}{n}\Big) \Big(m_i^{(s)}(\tilde{w}) \pm C_{i, s}\frac{L}{n}\Big) - m_s^{(s)}(w) m_i^{(s)}(\tilde{w}) \bigg)\\
\in & \left[ - 2\big(\tilde{C}_1 C_{i, s} + \tilde{C}_1  C_{s, s} + C_{i, s}C_{s, s}\big) \frac{L}{n}, 2\big(\tilde{C}_1 C_{i, s} + \tilde{C}_1  C_{s, s} + C_{i, s}C_{s, s}\big) \frac{L}{n} \right].
\end{align*}
This completes the induction with $C_{i, s + 1} = 2\big(\tilde{C}_1 C_{i, s} + \tilde{C}_1  C_{s, s} + C_{i, s}C_{s, s}\big)$.

Let now $X$ be given by Setting~\ref{setting:trend}\eqref{setting:Trendb}. The proof follows the same steps with two modifications. The induction beginning follows from applying difference of powers identity twice. We obtain
\begin{align*}
& m_i^{(1)}(w + l) - m_i^{(1)}(w)\\
= & \left(\frac{w + D + l}{n}\right)^{i} - \left(\frac{w + l}{n}\right)^{i} - \left(\left(\frac{w + D}{n}\right)^{i} - \left(\frac{w}{n}\right)^{i}\right)\\
= & \frac{D}{n} \sum_{j = 0}^{i - 1} \left( \left(\frac{w + D + l}{n}\right)^{i - j - 1} \left(\frac{w + l}{n}\right)^{j} - \left(\frac{w + D}{n}\right)^{i - j - 1} \left(\frac{w}{n}\right)^{j} \right)\\
= & \frac{D}{n} \sum_{j = 0}^{i - 2} \left( \left(\frac{w + D + l}{n}\right)^{i - j - 1} \left(\frac{w + l}{n}\right)^{j} -  \left(\frac{w + D}{n}\right)^{i - j - 1} \left(\frac{w + l}{n}\right)^{j}\right)\\
& + \frac{D}{n} \sum_{j = 1}^{i - 1} \left( \left(\frac{w + D}{n}\right)^{i - j - 1} \left(\frac{w + l}{n}\right)^{j} - \left(\frac{w + D}{n}\right)^{i - j - 1} \left(\frac{w}{n}\right)^{j}\right) \\
= & \frac{D}{n} \frac{l}{n} \sum_{j = 0}^{i - 2} \sum_{k = 0}^{i - j - 2} \left(\frac{w + D + l}{n}\right)^{i - j -2 - k} \left(\frac{w + D}{n}\right)^{k} \left(\frac{w + l}{n}\right)^{j}\\
& +  \frac{D}{n} \frac{l}{n} \sum_{j = 1}^{i - 1} \sum_{k = 0}^{j - 1} \left(\frac{w + l}{n}\right)^{j - 1 - k} \left(\frac{w}{n}\right)^{k} \left(\frac{w + D}{n}\right)^{i - j - 1}\\
\in & \left(0, (i - 1) i \frac{D}{n} \frac{L}{n}\right].
\end{align*}
So, $C_{i, s} = (i - 1) i$. The first three lines also show that
\begin{align*}
m_i^{(1)}(w) \leq \frac{D}{n} i.
\end{align*}
Consequently, there exists a constant $\tilde{C}_2 > 0$ such that 
\begin{equation}\label{eq:boundmisD}
m_i^{(s)}(w) \leq \tilde{C}_2 \left(\frac{D}{n}\right)^{2^{s - 1}}.
\end{equation}
Hence, the statement follows from the same induction.
%\begin{align*}
%C_{i, s + 1} = \tilde{C}_2 \frac{D}{n} C_{i, s} + \tilde{C}_2 \frac{D}{n} C_{s, s} + C_{i, s} C_{s, s}.
%\end{align*}

%A similar calculation shows that $C_{i, 2} = (i - 1)$. This completes the induction to show \eqref{eq:iterationCis}. We show now by induction with respect to $s$ that
%\begin{equation}\label{eq:boundCis}
%C_{i, s} = i - s + \sum_{j = 0}^{s - 2} 3^j.
%\end{equation}
%The induction start is true since $C_{i, 1} = i - 1$ and $C_{i, 2} = i - 1 = i - 2 + 3^0$. Let \eqref{eq:boundCis} be true for $s$. Then, it follows from \eqref{eq:iterationCis} and the induction assumption that
%\begin{align*}
%C_{i, s + 1}
%= C_{i, s} + 2 C_{s, s}
%= i - s + 3 \sum_{j = 0}^{s - 2} 3^{i}
%= i - (s + 1) + \sum_{j = 0}^{(s + 1) - 2} 3^{i}.
%\end{align*}
%This completes the induction and the proof as this shows that $C_{i, s}$ is a constant only depending on $i$ and $s$.
\end{proof}
We now lower bound $m_{i}^{(s)}$ for a special sequence of $w$'s. We consider $v = (v_1,\ldots,v_{2^{s - 1}}) \in W^{2^{s - 1}}$ of the form
\begin{equation}\label{eq:definitionv}
v_i = w_j \text{ if } i \operatorname{mod} 2^{j - 1} = 0, \text{ but } i \operatorname{mod} 2^{j} \neq 0.
\end{equation}
\begin{Lemma}\label{lemma:calculatemterms}
Under the assumptions of Lemma~\ref{lemma:determinant}, we have that
\begin{equation*}
m_{i}^{(s)}(v) \geq 
\begin{cases}
\left( B\frac{L}{n}\right)^{2^{s - 1}} & \text{ if X is given by Setting~\ref{setting:trend}\eqref{setting:Trenda},}\\
\left( \frac{D}{n} B\frac{L}{n} \right)^{2^{s - 1}} & \text{ if X is given by Setting~\ref{setting:trend}\eqref{setting:Trendb}.}
\end{cases}
\end{equation*}
\end{Lemma}
\begin{proof}
Let $X$ be given by Setting~\ref{setting:trend}\eqref{setting:Trenda}. We show by induction that $m_{i}^{(s)}(v) = P_s(w) S_{i,s}(w)$, where
\begin{equation*}
\begin{split}
P_s(w) \coloneq & \prod_{\substack{a, b \in \{1,\ldots,s - 1\},\\ a < b}} \left(\frac{w_a - w_b}{n}\right)^{2^{s - b - 1}} \prod_{a \in \{1,\ldots,s - 1\}} \left(\frac{w_a - w_s}{n}\right),\\
S_{i,s}(w) \coloneq & \sum_{l_0 = i - 1}^{i - 1}\sum_{l_1 = s - 2}^{l_0 - 1} \left(\frac{w_1}{n}\right)^{l_0 - l_1 - 1} \sum_{l_2 = s - 3}^{l_1 - 1} \left(\frac{w_2}{n}\right)^{l_1 - l_2 - 1} \cdots \sum_{l_{s - 1} = 0}^{l_{s - 2} - 1}  \left(\frac{w_{s - 1}}{n}\right)^{l_{s - 2} - l_{s - 1} - 1} \left(\frac{w_s}{n}\right)^{l_{s - 1}}.
\end{split}
\end{equation*}
The induction beginning is satisfied for $s = 1$ by
\begin{equation*}
m_{i}^{(1)}(w_1) = \left(\frac{w_1}{n}\right)^{i - 1} = \sum_{l_0 = i - 1}^{i - 1} \left(\frac{w_1}{n}\right)^{l_0} = P_1(w) S_{i,1}(w).
\end{equation*}
Suppose the statement is true for $s$. Let $w = (w_1,\ldots,w_s)$, $\tilde{w} = (w_1,\ldots,w_{s - 1}, w_{s + 1})$, and $\breve{w} = (w_1,\ldots,w_{s + 1})$. It follows from \eqref{eq:recursionm}, the induction assumption, $S_{s,s}(w) = 1$, and the fact that $v_{2^{s - 1} + 1:2^{s}}$ is equal to $v_{1:2^{s - 1}}$ except that the last entry is $w_{s + 1}$ instead of $w_s$ that
\begin{align*}
& m_{i}^{(s + 1)}\big((v_{1:2^{s - 1}}, v_{2^{s - 1} + 1:2^{s}})\big)\\
= &  m_{i}^{(s)}(v_{1:2^{s - 1}}) m_{s}^{(s)}(v_{2^{s - 1} + 1:2^{s}}) - m_{s}^{(s)}(v_{1:2^{s - 1}}) m_{i}^{(s)}(v_{2^{s - 1} + 1:2^{s}})\\
= & P_s(w) S_{i,s}(w) P_s(\tilde{w}) - P_s(w) P_s(\tilde{w}) S_{i,s}(\tilde{w})\\
= & P_s(w) P_s(\tilde{w}) \big(S_{i,s}(w)  - S_{i,s}(\tilde{w})\big).
\end{align*}
By construction we have that
\begin{align*}
P_s(w) P_s(\tilde{w}) = P_{s + 1}(\breve{w} ) \frac{n}{w_s - w_{s + 1}}.
\end{align*}
Finally, the difference of power identity yields
\begin{align*}
& S_{i,s}(w) - S_{i,s}(\tilde{w})\\
= & \sum_{l_0 = i - 1}^{i - 1}\sum_{l_1 = s - 2}^{l_0 - 1} \left(\frac{w_1}{n}\right)^{l_0 - l_1 - 1} \cdots \sum_{l_{s - 1} = 0}^{l_{s - 2} - 1}  \left(\frac{w_{s - 1}}{n}\right)^{l_{s - 2} - l_{s - 1} - 1} \left(\left(\frac{w_s}{n}\right)^{l_{s - 1}} - \left(\frac{w_{s + 1}}{n}\right)^{l_{s - 1}}\right) \\
= & \frac{w_s - w_{s + 1}}{n} \sum_{l_0 = i - 1}^{i - 1}\sum_{l_1 = (s + 1) - 2}^{l_0 - 1} \left(\frac{w_1}{n}\right)^{l_0 - l_1 - 1}  \cdots \sum_{l_{s - 1} = 1}^{l_{s - 2} - 1}  \left(\frac{w_{s - 1}}{n}\right)^{l_{s - 2} - l_{s - 1} - 1}  \sum_{l_s = 0}^{l_{s - 1} - 1} \left(\frac{w_s}{n}\right)^{l_{s - 1} - l_s - 1} \left(\frac{w_{s + 1}}{n}\right)^{l_s}\\
= & \frac{w_s - w_{s + 1}}{n} S_{i,s + 1}(\breve{w}).
\end{align*}
This completes the induction. Thus,
\begin{equation*}
m_{s}^{(s)}(v) = P_s(w) \geq \left( B\frac{L}{n}\right)^{2^{s - 1}}.
\end{equation*}
Let $X$ now be given by Setting~\ref{setting:trend}\eqref{setting:Trendb}. We follow the same steps, but with
\begin{align*}
P_s(w) \coloneq & \left(\frac{D}{n}\right)^{2^{s - 1}}\prod_{\substack{a, b \in \{1,\ldots,s - 1\},\\ a < b}} \left(\frac{w_a - w_b}{n}\right)^{2^{s - b - 1}} \prod_{a \in \{1,\ldots,s - 1\}} \left(\frac{w_a - w_s}{n}\right)
\end{align*}
and
\begin{align*}
& S_{i,s}(w)\coloneq \\
 & \sum_{l_0 = s - 1}^{i - 1} \left(\frac{w_1 + D}{n}\right)^{i - l_0 - 1} \sum_{l_1 = s - 2}^{l_0 - 1}  \left(\frac{w_1}{n}\right)^{l_0 - l_1 - 1} \sum_{l_2 = s - 3}^{l_1 - 1} \left(\frac{w_2}{n}\right)^{l_1 - l_2 - 1} \cdots \sum_{l_{s - 1} = 0}^{l_{s - 2} - 1} \left(\frac{w_{s - 1}}{n}\right)^{l_{s - 2} - l_{s - 1} - 1} \left(\frac{w_s}{n}\right)^{l_{s - 1}}\\
 + &  \sum_{l_0 = s - 2}^{i - 2}\sum_{l_1 = 0}^{i - 1 - l_0 - 1} \left(\frac{w_1 + D}{n}\right)^{i - 1 - l_0 - 1 - l_1} \left(\frac{w_2 + D}{n}\right)^{l_1} \sum_{l_2 = s - 3}^{l_0 - 1}  \left(\frac{w_2}{n}\right)^{l_0 - l_2 - 1} \cdots \sum_{l_{s - 1} = 0}^{l_{s - 2} - 1}  \left(\frac{w_{s - 1}}{n}\right)^{l_{s - 2} - l_{s - 1} - 1} \left(\frac{w_s}{n}\right)^{l_{s - 1}}\\
 + & \ \cdots\\
 + &  \sum_{l_0 = 1}^{i - (s - 1)}\sum_{l_1 = s - 3}^{i - 1 - l_0 - 1} \left(\frac{w_1 + D}{n}\right)^{i - 1 - l_0 - 1 - l_1} \sum_{l_2 = s - 4}^{l_1 - 1} \left(\frac{w_2 + D}{n}\right)^{l_1 - l_2 - 1} \cdots  \left(\frac{w_{s - 1} + D}{n}\right)^{l_{s - 2}} \sum_{l_{s - 1} = 0}^{l_{0} - 1} \left(\frac{w_{s - 1}}{n}\right)^{l_{0} - l_{s - 1} - 1} \left(\frac{w_s}{n}\right)^{l_{s - 1}}\\
 + & \sum_{l_0 = 0}^{i - s}\sum_{l_1 = s - 2}^{i - 1 - l_0 - 1}\left(\frac{w_1 + D}{n}\right)^{i - 1 - l_0 - 1 - l_1} \sum_{l_2 = s - 3}^{l_1 - 1} \left(\frac{w_2 + D}{n}\right)^{l_1 - l_2 - 1} \cdots \left(\frac{w_{s - 1} + D}{n}\right)^{l_{s - 2} - l_{s - 1} - 1} \sum_{l_{s - 1} = 0}^{l_{s - 2} - 1} \left(\frac{w_s + D}{n}\right)^{l_{s - 1}} \left(\frac{w_s}{n}\right)^{l_{0}},
\end{align*}
where $s$ rows appear. The induction beginning for $s = 1$ follows from the difference of power identity
\begin{equation*}
m_{i}^{(1)}((w_1)) = \left(\frac{w_1 + D}{n}\right)^{i} - \left(\frac{w_1}{n}\right)^{i}
= \frac{D}{n}\sum_{l_0 = 0}^{i - 1} \left(\frac{w_1}{n}\right)^{i - l_0 - 1} \left(\frac{w_1 + D}{n}\right)^{l_0} = P_1(w) S_{i,1}(w).
\end{equation*}
For the induction step we focus on the last sum as all other terms follow the same pattern as before. Again, by difference of power identity,
\begin{align*}
& \left(\frac{w_s + D}{n}\right)^{l_{s - 1}} \left(\frac{w_s}{n}\right)^{l_{0}} -  \left(\frac{w_{s + 1} + D}{n}\right)^{l_{s - 1}} \left(\frac{w_{s + 1}}{n}\right)^{l_{0}}\\
 = & \left(\frac{w_s + D}{n}\right)^{l_{s - 1}} \left(\frac{w_s}{n}\right)^{l_{0}} - \left(\frac{w_{s + 1} + D}{n}\right)^{l_{s - 1}} \left(\frac{w_{s}}{n}\right)^{l_{0}} + \left(\frac{w_{s + 1} + D}{n}\right)^{l_{s - 1}} \left(\frac{w_{s}}{n}\right)^{l_{0}}  -  \left(\frac{w_{s + 1} + D}{n}\right)^{l_{s - 1}} \left(\frac{w_{s + 1}}{n}\right)^{l_{0}}\\
= & \frac{w_{s + 1} - w_{s}}{n} \left(\sum_{l_s = 0}^{l_{s - 1} - 1} \left(\frac{w_{s}}{n}\right)^{l_{0}} \left(\frac{w_{s} + D}{n}\right)^{l_{s - 1} - l_s - 1} \left(\frac{w_{s + 1} + D}{n}\right)^{l_{s}} + \sum_{l_s = 0}^{l_0 - 1} \left(\frac{w_{s + 1} + D}{n}\right)^{l_{s - 1}} \left(\frac{w_s}{n}\right)^{l_{0} - l_s - 1}  \left(\frac{w_{s + 1}}{n}\right)^{l_{s}}\right).
\end{align*}
This completes the induction. Thus,
\begin{equation*}
m_{s}^{(s)}(v) = P_s(w) s \geq s \left( \frac{D}{n} B\frac{L}{n} \right)^{2^{s - 1}} \geq \left( \frac{D}{n} B\frac{L}{n} \right)^{2^{s - 1}} .
\end{equation*}

\end{proof}
We are now able to proof the main Lemma.
\begin{proof}[Proof of Lemma~\ref{lemma:determinant}]
Let $X$ be given by Setting~\ref{setting:trend}\eqref{setting:Trenda}. We show by induction that there exists constants $\tilde{C}_{Q, s} > 0$ only depending on $Q$ and $s$ such that
\begin{equation}\label{eq:boundM}
M_{s,s}^{(s)} = \tilde{C}_{Q, s} L.
\end{equation}
The induction begin is satisfied, since
\begin{align*}
m_1^{(1)}(u) = \left(\frac{u}{n}\right)^{0} = 1 \text{ and } M_{1, 1}^{(1)} = \sum_{u \in U} m_1^{(1)}(u)^2 \in \big[ (Q + 1) (L - 2) \cdot 1, (Q + 1) L \cdot 1\big).
\end{align*}
Suppose \eqref{eq:boundM} holds for $s - 1$. Then from Lemma~\ref{lemma:simplifyM} and the induction assumption it follows that
\begin{align*}
M_{s,s}^{(s)}
=  \frac{\sum_{u \in U^{2^{s - 1}}} m_{s}^{(s)}(u)^2}{\prod_{t = 1}^{s - 1} \left(2 M_{t, t}^{(t)}\right)^{2^{s - t - 1}}}
=  \frac{\sum_{u \in U^{2^{s - 1}}} m_{s}^{(s)}(u)^2}{L^{2^{s - 1} - 1} \prod_{t = 1}^{s - 1} \left(2 \tilde{C}_{Q, t}\right)^{2^{s - t - 1}}}.
\end{align*}
It remains to show that there exists a constant $\tilde{C}_1 > 0$ such that 
\begin{equation*}
\sum_{u \in U^{2^{s - 1}}} m_{s}^{(s)}(u)^2 = \tilde{C}_1 L^{2^{s-1}}.
\end{equation*}
In the following we first upper, then lower bound this expression.

It follows from \eqref{eq:recursionm} and $m_i^{(1)}(u) = \left(\frac{u}{n}\right)^{i - 1} \leq 1$ that there exists a constant $\tilde{C}_2 > 0$ such that $m_s^{(s)}(u) \leq \tilde{C}_2$. Consequently,
\begin{align*}
\sum_{u \in U^{2^{s - 1}}} m_{s}^{(s)}(u)^2 \leq (Q+1)^{2^{s - 1}} L^{2^{s - 1}} \tilde{C}_2^2 \text{ and } \tilde{C}_{1} \leq (Q+1)^{2^{s - 1}} \tilde{C}_2^2.
\end{align*}
We now obtain a lower bound as well. Let $v$ be as defined in \eqref{eq:definitionv}. Then, it follows from Lemmas~\ref{lemma:differencemterms}~and~\ref{lemma:calculatemterms} that
\begin{align*}
\sum_{u \in U^{2^{s - 1}}} m_{s}^{(s)}(u)^2
\geq & \sum_{l \in \{1,\ldots,L - 2\}^{2^{s - 1}}} \bigg( m_{s}^{(s)}(v) - \left(m_{s}^{(s)}(v + l) - m_{s}^{(s)}(v)\right)  \bigg)^2\\
\geq & (L - 2)^{2^{s - 1}} \left( \left(B \frac{L}{n}\right)^{2^{s - 1}} - C_{s, s} \frac{L}{n} \right)^2.
\end{align*}
We have that the constant $\tilde{C}_1$ exists if $B$ is chosen large enough such that
\begin{equation}\label{eq:secondConditionB}
\frac{(L - 2)^{2^{s - 1}}}{L^{2^{s - 1}}}\left( \left(B \frac{L}{n}\right)^{2^{s - 1}} - C_{s, s} \frac{L}{n} \right)^2 \geq \tilde{C}_1.
\end{equation}
This completes the induction to show \eqref{eq:boundM}. The proof of the statement follows, since $L \geq \ClengthL n$ and hence
\begin{equation*}
\operatorname{det}\left( X_U^\top X_U \right)
= \prod_{s = 1}^{Q + 1} M_{s,s}^{(s)}
= \Cdeterminant n^{Q + 1},
\end{equation*}
with $\Cdeterminant = \ClengthL^{Q + 1} \prod_{s = 1}^{Q + 1} \tilde{C}_{Q, s}$.

Let $X$ now be given by Setting~\ref{setting:trend}\eqref{setting:Trendb}. The proof follows the same steps. We show by induction that there exists constants $\tilde{C}_{Q, s} > 0$ only depending on $Q$ and $s$ such that
\begin{equation}\label{eq:boundMD}
M_{s,s}^{(s)} = \tilde{C}_{Q, s} L \left(\frac{D}{n}\right)^2.
\end{equation}
The induction begin is satisfied with $\tilde{C}_{Q, 1} = Q$, since
\begin{align*}
& m_1^{(1)}(u) = \left(\frac{u + D}{n}\right) -  \left(\frac{u}{n}\right) = \frac{D}{n},\\
&  M_{1, 1}^{(1)} = \sum_{u \in U} m_1^{(1)}(u)^2 \in \left[ Q (L - 2) \left( \frac{D}{n} \right)^2, Q L \left( \frac{D}{n} \right)^2\right].
\end{align*}
Suppose \eqref{eq:boundMD} holds for $s - 1$. Then from Lemma~\ref{lemma:simplifyM} and the induction assumption it follows that
\begin{align*}
M_{s,s}^{(s)}
=  \frac{\sum_{u \in U^{2^{s - 1}}} m_{s}^{(s)}(u)^2}{\prod_{t = 1}^{s - 1} \left(2 M_{t, t}^{(t)}\right)^{2^{s - t - 1}}}
=  \frac{\sum_{u \in U^{2^{s - 1}}} m_{s}^{(s)}(u)^2}{L^{2^{s - 1} - 1} \left(\frac{D}{n}\right)^{2^s - 2} \prod_{t = 1}^{s - 1} \left(2 \tilde{C}_{Q, t}\right)^{2^{s - t - 1}} }.
\end{align*}
It remains to show that there exists a constant $\tilde{C}_3 > 0$ such that 
\begin{equation*}
\sum_{u \in U^{2^{s - 1}}} m_{s}^{(s)}(u)^2 = \tilde{C}_3 L^{2^{s - 1}} \left(\frac{D}{n}\right)^{2^{s}}.
\end{equation*}
In the following we first upper, then lower bound this expression.

We had shown in \eqref{eq:boundmisD} that there exists a constant $\tilde{C}_4 > 0$ such that 
\begin{equation*}
m_i^{(s)}(w) \leq \tilde{C}_4 \left(\frac{D}{n}\right)^{2^{s - 1}}.
\end{equation*}
Consequently,
\begin{align*}
\sum_{u \in U^{2^{s - 1}}} m_{s}^{(s)}(u)^2 \leq Q^{2^{s - 1}} L^{2^{s - 1}} \tilde{C}_4^2 \left(\frac{D}{n}\right)^{2^s}  \text{ and } \tilde{C}_{3} \leq Q^{2^{s - 1}} \tilde{C}_4^2.
\end{align*}
We now obtain a lower bound as well. Let $v$ be as defined in \eqref{eq:definitionv}. Then, it follows from Lemmas~\ref{lemma:differencemterms}~and~\ref{lemma:calculatemterms} that
\begin{align*}
\sum_{u \in U^{2^{s - 1}}} m_{s}^{(s)}(u)^2
\geq & \sum_{l \in \{1,\ldots,L - 2\}^{2^{s - 1}}} \bigg( m_{s}^{(s)}(v) - \left(m_{s}^{(s)}(v + l) - m_{s}^{(s)}(v)\right)  \bigg)^2\\
\geq & (L - 2)^{2^{s - 1}} \left( \left( \frac{D}{n} B\frac{L}{n}\right)^{2^{s - 1}} - C_{s, s} \left(\frac{D}{n}\right)^{2^{s - 1}}\frac{L}{n} \right)^2.
\end{align*}
We have that the constant $\tilde{C}_3$ exists if $B$ is chosen large enough such that
\begin{equation}\label{eq:thirdConditionB}
\frac{(L - 2)^{2^{s - 1}}}{L^{2^{s - 1}}} \left( \left( \frac{D}{n} B\frac{L}{n}\right)^{2^{s - 1}} - C_{s, s} \left(\frac{D}{n}\right)^{2^{s - 1}}\frac{L}{n} \right)^2 \geq \tilde{C}_3 \left(\frac{D}{n}\right)^{2^{s - 1}}.
\end{equation}
This completes the induction to show \eqref{eq:boundMD}. The proof follows, since $L \geq \ClengthL n$ and hence
\begin{equation*}
\operatorname{det}\left( X_U^\top X_U \right)
= \prod_{s = 1}^{Q} M_{s,s}^{(s)}
= \Cdeterminant \left(\frac{D^2}{n}\right)^{Q},
\end{equation*}
with $\Cdeterminant = \ClengthL^{Q} \prod_{s = 1}^{Q} \tilde{C}_{Q, s}$.
\end{proof}

\begin{Lemma}\label{lemma:trace}
Suppose that the assumptions of Lemma~\ref{lemma:determinant} hold. Then there exists a constant $\Ctrace > 0$ such that
\begin{equation*}
\operatorname{tr}\left[ X_{U} (X_{U}^\top X_{U})^{-1} (X^\top X) (X_{U}^\top X_{U})^{-1} X_{U}^\top \right] \leq \Ctrace.
\end{equation*}
\end{Lemma}
\begin{proof}[Proof of Lemma~\ref{lemma:trace}]
For any matrix $A$ and any symmetric, positive definite matrix $B$, with dimensions such that the products exist, we have that $\operatorname{tr}(A B A^\top) = \operatorname{tr}(B A^\top A)$. We obtain
\begin{align*}
\operatorname{tr}\left[ X_{U} (X_{U}^\top X_{U})^{-1} (X^\top X) (X_{U}^\top X_{U})^{-1} X_{U}^\top \right]
= \operatorname{tr}\left[ (X_{U}^\top X_{U})^{-1} (X^\top X) \right].
\end{align*}

Let $X$ now be given by Setting~\ref{setting:trend}\eqref{setting:Trenda}. Then, we have that $0 \leq X_{i,j} \leq 1$ for all $i = 1,\ldots,n;\ j = 1,\ldots,Q + 1$. Consequently, $(X^\top X)_{i,j} \leq n$ for all $i= 1,\ldots,Q + 1;\ j = 1,\ldots,Q + 1$.

Next, we focus on the entries in the inverse matrix $(X_U^\top X_U)^{-1}$. Since $X_U$ is a sub-matrix of $X$, we have once again that $(X_U)_{i,j} \leq 1$ and $(X_U^\top X_U)_{i,j} \leq n$. The inverse $(X_U^\top X_U)^{-1}$ is given by the cofactor matrix of $X_U^\top X_U$ divided by the determinant of $X_U^\top X_U$. Each entry in the cofactor matrix is given by the determinant of a $Q\times Q$ sub-matrix of $X_U^\top X_U$. Thus, each entry is bounded by $\CcofactorXu n^{Q}$, where $\CcofactorXu > 0$ is a constant only depending on $Q$. Together with Lemma~\ref{lemma:determinant}, which says $\operatorname{det}\left( X_U^\top X_U \right) \geq \Cdeterminant n^{Q + 1}$, we obtain
\begin{equation*}
\left[(X_{U}^\top X_{U})^{-1}\right]_{i,j} \leq \frac{\CcofactorXu n^{Q}}{\Cdeterminant n^{Q + 1}} = \frac{\CcofactorXu}{\Cdeterminant} n^{-1} \text{ for all } i,j \in \{1,\ldots,Q+1\}.
\end{equation*}
Consequently,
\begin{equation*}
\left[(X_{U}^\top X_{U})^{-1} (X^\top X)\right]_{i,j} \leq (Q+1) \frac{\CcofactorXu}{\Cdeterminant} n^{-1} n = (Q+1) \frac{\CcofactorXu}{\Cdeterminant},
\end{equation*}
for all $i,j \in \{1,\ldots,Q+1\}$. Thus,
\begin{equation*}
\operatorname{tr}\left[ (X_{U}^\top X_{U})^{-1} (X^\top X) \right] \leq (Q+1)^2 \frac{\CcofactorXu}{\Cdeterminant} \eqcolon \Ctrace.
\end{equation*}
Note that $\Ctrace > 0$ only depends on $Q$. This completes the proof if $X$ is given by Setting~\ref{setting:trend}\eqref{setting:Trenda}.

Let $X$ now be given by Setting~\ref{setting:trend}\eqref{setting:Trendb}. Then, we have that $0 \leq X_{i,j} \leq Q \frac{D}{n}$ for all $i = 1,\ldots,n - D;\ j = 1,\ldots,Q$. Consequently, there exists a constant $\tilde{C}_1> 0$ such that $(X^\top X)_{i,j} \leq \tilde{C}_1 \frac{D^2}{n}$ for all $i= 1,\ldots,Q;\ j = 1,\ldots,Q$.

Next, we focus on the entries in the inverse matrix $(X_U^\top X_U)^{-1}$. Since $X_U$ is a sub-matrix of $X$, we have once again that $(X_U)_{i,j} \leq Q \frac{D}{n}$ and $(X_U^\top X_U)_{i,j} \leq \tilde{C}_1 \frac{D^2}{n}$. The inverse $(X_U^\top X_U)^{-1}$ is given by the cofactor matrix of $X_U^\top X_U$ divided by the determinant of $X_U^\top X_U$. Each entry in the cofactor matrix is given by the determinant of a $(Q - 1)\times (Q - 1)$ sub-matrix of $X_U^\top X_U$. Thus, each entry is bounded by $\CcofactorXu \left(\frac{D^2}{n}\right)^{Q - 1}$, where $\CcofactorXu > 0$ is a constant only depending on $Q$. Together with Lemma~\ref{lemma:determinant}, which says $\operatorname{det}\left( X_U^\top X_U \right) \geq \Cdeterminant \left(\frac{D^2}{n}\right)^{Q}$, we obtain
\begin{equation*}
\left[(X_{U}^\top X_{U})^{-1}\right]_{i,j} \leq \frac{\CcofactorXu \left(\frac{D^2}{n}\right)^{Q - 1}}{\Cdeterminant \left(\frac{D^2}{n}\right)^{Q}} = \frac{\CcofactorXu}{\Cdeterminant} \left(\frac{D^2}{n}\right)^{-1} \text{ for all } i,j \in \{1,\ldots,Q\}.
\end{equation*}
Consequently,
\begin{equation*}
\left[(X_{U}^\top X_{U})^{-1} (X^\top X)\right]_{i,j} \leq Q \frac{\CcofactorXu}{\Cdeterminant} \left(\frac{D^2}{n}\right)^{-1} \left(\frac{D^2}{n}\right) = Q \frac{\CcofactorXu}{\Cdeterminant},
\end{equation*}
for all $i,j \in \{1,\ldots,Q\}$. Thus,
\begin{equation*}
\operatorname{tr}\left[ (X_{U}^\top X_{U})^{-1} (X^\top X) \right] \leq Q^2 \frac{\CcofactorXu}{\Cdeterminant} \eqcolon \Ctrace.
\end{equation*}
Note that $\Ctrace > 0$ only depends on $Q$. This completes the proof.
\end{proof}

We now bound the error $\hat{T}_U - T$ of the trend estimate given by $T_U$ in \eqref{eq:trendLS}.
\begin{Lemma}\label{lemma:trendEstimation}
Let $U \coloneq U_{B,Q,j} \subset \{1,\ldots,\mn\}$ be a sample from the sampling procedure in Section~\ref{section:sampling} and suppose that $U \cap \mathcal{A} = \emptyset$. Suppose that our data follows Setting~\ref{setting:trend}. Then there exists a constant $\CerrorTrend > 0$ such that
\begin{equation}\label{eq:statementTrendEstimation}
\lim_{n \to \infty} \Pj\left( (\hat{T}_U - T)^\top (\hat{T}_U - T) > \CerrorTrend \sigma_0^2 \log(n) \right) \to 0.
\end{equation}
\end{Lemma}
\begin{proof}[Proof of Lemma~\ref{lemma:trendEstimation}]

The least squares regression estimator is given by $\hat{T}_U = X (X_{U}^\top X_{U})^{-1} X_{U}^\top Y_{U}$. 

Since $U \cap \mathcal{A} = \emptyset$, we have that $Y_{U} = X_{U} \beta + W + \varepsilon_{U}$. Hence, $\hat{T}_U - T = X (X_{U}^\top X_{U})^{-1} X_{U}^\top (W + \varepsilon_{U})$. Therefore, we have that
\begin{align*}
(\hat{T}_U - T)^\top (\hat{T}_U - T)
 = & (W + \varepsilon_{U})^\top X_{U} (X_{U}^\top X_{U})^{-1} (X^\top X) (X_{U}^\top X_{U})^{-1} X_{U}^\top (W + \varepsilon_{U}) \\
\leq & 2 W^\top X_{U} (X_{U}^\top X_{U})^{-1} (X^\top X) (X_{U}^\top X_{U})^{-1} X_{U}^\top W\\
& + 2 \varepsilon_{U}^\top X_{U} (X_{U}^\top X_{U})^{-1} (X^\top X) (X_{U}^\top X_{U})^{-1} X_{U}^\top \varepsilon_{U}.
\end{align*}

For the first term, we obtain
\begin{align*}
& W^\top X_{U} (X_{U}^\top X_{U})^{-1} (X^\top X) (X_{U}^\top X_{U})^{-1} X_{U}^\top W\\
\leq & W^\top W \operatorname{tr}\left[ X_{U} (X_{U}^\top X_{U})^{-1} (X^\top X) (X_{U}^\top X_{U})^{-1} X_{U}^\top \right]\\
\leq & \CremainderSeasonal \sigma_0^2 \log(n) \Ctrace,
\end{align*}
where the final inequality follows from \eqref{eq:remainderSeasonalyT} and Lemma~\ref{lemma:trace}.

We now focus on the second term. For any symmetric positive definite matrix $A$ and any random variable $X$ with $\E[X] = 0$, it follows from \citet[(B.1)]{Kendrick02stochastic} that $\E\left[ X^\top A X \right] =  \operatorname{tr}\big(A \Cov[X]\big)$. Therefore,
\begin{align*}
 & \E\left[ \varepsilon_{U}^\top X_{U} (X_{U}^\top X_{U})^{-1} (X^\top X) (X_{U}^\top X_{U})^{-1} X_{U}^\top \varepsilon_{U} \right]\\
= & \operatorname{tr}\left[ X_{U} (X_{U}^\top X_{U})^{-1} (X^\top X) (X_{U}^\top X_{U})^{-1} X_{U}^\top \operatorname{Cov}\left[ \varepsilon_{U}\right] \right]\\
\leq &  \rQ\operatorname{colSum}\left( \vert \operatorname{Cov}\left[ \varepsilon_{U}\right] \vert \right) \operatorname{tr}\left[ X_{U} (X_{U}^\top X_{U})^{-1} (X^\top X) (X_{U}^\top X_{U})^{-1} X_{U}^\top \right]\\
\leq & 4\rQ\sigma_0^2 \Ctrace,
\end{align*}
where $\operatorname{colSum}\left( \vert \operatorname{Cov}\left[ \varepsilon_{U}\right] \vert \right)$ is the column sum of the element-wise absolute values of the covariance matrix. The final inequality follows from Lemmas~\ref{lemma:differencedSequence}~and~\ref{lemma:trace}.

Next, since $\varepsilon_{U}^\top X_{U} (X_{U}^\top X_{U})^{-1} (X^\top X) (X_{U}^\top X_{U})^{-1} X_{U}^\top \varepsilon_{U} \geq 0$, Markov's inequality gives us
\begin{align*}
& \Pj\left(\varepsilon_{U}^\top X_{U} (X_{U}^\top X_{U})^{-1} (X^\top X) (X_{U}^\top X_{U})^{-1} X_{U}^\top \varepsilon_{U} > \sigma_0^2 \log(n)\right)\\
 \leq & \frac{\E\left[\varepsilon_{U}^\top X_{U} (X_{U}^\top X_{U})^{-1} (X^\top X) (X_{U}^\top X_{U})^{-1} X_{U}^\top \varepsilon_{U} \right]}{\sigma_0^2 \log(n)} \to 0.
\end{align*}
It follows that \eqref{eq:statementTrendEstimation} holds with $\CerrorTrend \coloneq 2 \CremainderSeasonal \Ctrace + 2$.
\end{proof}

\subsubsection{Uniform Bounds on the Noise}\label{sec:proofBoundNoise}
In this section we bound relevant noise terms.
\begin{Lemma}\label{lemma:noiseLin}
Suppose that our data follows Setting~\ref{setting:trend}. Then, there exists a constant $\CnoiseLin < \infty$ such that
\begin{equation*}
\Pj\left( \max_{1 \leq i \leq j \leq \mn} \varepsilon_{i:j}^\top H_{i:j} \varepsilon_{i:j}
\leq \CnoiseLin \sigma_0^2 \log(n) \right) \to 1,
\end{equation*}
as $n\to \infty$.
\end{Lemma}
\begin{proof}
We first consider the Setting~\ref{setting:trend}\eqref{setting:Trenda}. Note that $H_{i:j}$ is a projection matrix on a $Q + 1$ dimensional column space. Hence, $H_{i:j}$ is symmetric, idempotent, and has rank $\min\{Q + 1, j - i + 1\}$. Consequently, \citep[Theorem~5.1.1]{mathai92quadratic} implies that
\begin{equation*}
\frac{\varepsilon_{i:j}^\top H_{i:j} \varepsilon_{i:j}}{\sigma_0^2}
\end{equation*}
is chi-squared distributed with $\min\{Q + 1, j - i + 1\}$ degrees of freedom. Thus, \citep[Lemma~1]{laurent2000adaptive} yields
\begin{equation*}
\Pj \left( \frac{\varepsilon_{i:j}^\top H_{i:j} \varepsilon_{i:j}}{\sigma_0^2}  > \min\{Q + 1, j - i + 1\} + 2 \sqrt{\min\{Q + 1, j - i + 1\} t} + 2t \right) \leq \exp\left(-t\right),
\end{equation*}
for any $t > 0$. Let $t = \CnoiseLin \sigma_0^2 \log(n)$ for a constant $\CnoiseLin> 0$. If $n$ is large enough, then
\begin{align*}
\min\{Q + 1, j - i + 1\} + 2 \sqrt{\min\{Q + 1, j - i + 1\} t} + 2t
\leq \left(\sqrt{Q + 1} + \sqrt{2 t}\right)^2
\leq 3t
\end{align*}
Thus, if $n$ is large enough, we obtain
\begin{equation*}
\Pj \left( \varepsilon_{i:j}^\top H_{i:j} \varepsilon_{i:j}  > \CnoiseLin \sigma_0^2 \log(n)\right) \leq \exp\left(- \frac{\CnoiseLin}{3} \log(n) \right) = n^{-\frac{\CnoiseLin}{3}}.
\end{equation*}
If $\CnoiseLin > 6$, combining this bound with a union bound gives us
\begin{align*}
& 1 - \Pj\left( \max_{1 \leq i \leq j \leq n} \varepsilon_{i:j}^\top H_{i:j} \varepsilon_{i:j}
\leq \CnoiseLin \sigma_0^2 \log(n) \right)\\
& \leq \sum_{1 \leq i \leq j \leq n} \Pj \left(\varepsilon_{i:j}^\top H_{i:j} \varepsilon_{i:j}  > \CnoiseLin \sigma_0^2 \log(n) \right)\\
& \leq n^2 n^{-\frac{\CnoiseLin}{3}} \to 0,
\end{align*}
as $n\to \infty$.

We now consider the Setting~\ref{setting:trend}\eqref{setting:Trendb}. Lemma~\ref{lemma:differencedSequence} implies that $\varepsilon_D = \varepsilon_{D + 1:n} - \varepsilon_{1:{n-D}}$. Consequently, $(\varepsilon_D)_{i:j} = \varepsilon_{i+D:j+D} - \varepsilon_{i:j}$. Therefore, 
\begin{align*}
& (\varepsilon_D)_{i:j}^\top H_{i:j} (\varepsilon_D)_{i:j}\\
= & \varepsilon_{i+D:j+D}^\top H_{i:j} \varepsilon_{i+D:j+D} + \varepsilon_{i:j}^\top H_{i:j} \varepsilon_{i:j} - 2 \varepsilon_{i+D:j+D}^\top H_{i:j} \varepsilon_{i:j}\\
\leq & \varepsilon_{i+D:j+D}^\top H_{i:j} \varepsilon_{i+D:j+D} + \varepsilon_{i:j}^\top H_{i:j} \varepsilon_{i:j} + 2 \sqrt{\varepsilon_{i+D:j+D}^\top H_{i:j} \varepsilon_{i+D:j+D}} \sqrt{\varepsilon_{i:j}^\top H_{i:j} \varepsilon_{i:j}}\\
\leq & 2\varepsilon_{i+D:j+D}^\top H_{i:j} \varepsilon_{i+D:j+D} + 2\varepsilon_{i:j}^\top H_{i:j}. \varepsilon_{i:j}.
\end{align*}
Following the same steps as in Setting~\ref{setting:trend}\eqref{setting:Trenda} completes the proof.
\end{proof}

Lemma~\ref{lemma:mixedtermNoisePolynomial} bounds the mixed term between noise and a polynomial fit.
\begin{Lemma}\label{lemma:mixedtermNoisePolynomial}
Let $\tilde{T} = \big(\tilde{t}_1,\ldots,\tilde{t}_{\mn}\big)$ be a polynomial of degree $\rQ$, then for any $1 \leq i \leq j \leq \mn$, we have that
\begin{equation*}
\left\vert\sum_{l = i}^{j} \varepsilon_l \tilde{t}_l \right\vert \leq \sqrt{\varepsilon_{i:j}^\top H_{i:j} \varepsilon_{i:j}} \sqrt{\sum_{l = i}^{j} \tilde{t}_l^2}.
\end{equation*}
\end{Lemma}
\begin{proof}
Since $\tilde{T}$ is a polynomial of degree $\rQ$, there exists a $\tilde{\beta} \in \R^{\rQ + 1}$ such that $\tilde{T} = X \tilde{\beta}$. Let $C^2 \coloneq \tilde{\beta}^\top X_{i:j}^\top X_{i:j} \tilde{\beta}$ and $\breve{\beta} = \frac{1}{C} \tilde{\beta}$. Then,
\begin{align*}
\sum_{l = i}^{j} \varepsilon_l \tilde{t}_l 
= \varepsilon_{i:j}^\top X_{i:j} \tilde{\beta} 
=  C \varepsilon_{i:j}^\top X_{i:j} \breve{\beta} 
\leq C \max_{\beta \in \R^{\rQ + 1}:\ \beta^\top X_{i:j}^\top X_{i:j} \beta = 1} \varepsilon_{i:j}^\top X_{i:j}\beta.
\end{align*}
It remains to solve the optimisation problem on the right. We use Lagrangian multipliers. We have to maximise
\begin{equation*}
\varepsilon_{i:j}^\top X_{i:j}\beta + \lambda \left( \beta^\top X_{i:j}^\top X_{i:j} \beta - 1 \right).
\end{equation*}
We obtain
\begin{align*}
& \beta = \frac{1}{\lambda} \left(X_{i:j}^\top X_{i:j}\right)^{-1} X_{i:j}^\top \varepsilon_{i:j},\\
& \beta^\top X_{i:j}^\top X_{i:j} \beta = 1.
\end{align*}
Hence, $\lambda = \sqrt{\varepsilon_{i:j}^\top H_{i:j} \varepsilon_{i:j}}$. Consequently,
\begin{align*}
\sum_{l = i}^{j} \varepsilon_l \tilde{t}_l 
\leq & C \max_{\beta \in \R^{\rQ + 1}:\ \beta^\top X_{i:j}^\top X_{i:j} \beta = 1} \varepsilon_{i:j}^\top X_{i:j}\beta\\ 
= & C \frac{\varepsilon_{i:j}^\top H_{i:j} \varepsilon_{i:j}}{\sqrt{\varepsilon_{i:j}^\top H_{i:j} \varepsilon_{i:j}}}
 = \sqrt{\varepsilon_{i:j}^\top H_{i:j} \varepsilon_{i:j}} \sqrt{\sum_{l = i}^{j} \tilde{t}_l^2}.
\end{align*}
Similarly, we have that
\begin{align*}
\sum_{l = i}^{j} \varepsilon_l \tilde{t}_l 
\geq  -\max_{\beta \in \R^{\rQ + 1}:\ \beta^\top X_{i:j}^\top X_{i:j} \beta = 1} \varepsilon_{i:j}^\top X_{i:j}\beta
 = -\sqrt{\varepsilon_{i:j}^\top H_{i:j} \varepsilon_{i:j}} \sqrt{\sum_{l = i}^{j} \tilde{t}_l^2}.
\end{align*}

This completes the proof.
\end{proof}

\subsubsection{Bounding the cost difference}\label{sec:proofBoundCostDifference}

To show that we select the right trend estimate using the cost of CAPA, we contrast in the following the costs of a good trend estimate with the one of a bad trend estimate. We require multiple lemmas to bound these costs.

We start by rewriting the anomaly model as a changepoint model. Let $\tau_{2k} \coloneq b_k$ and $\tau_{2k + 1} \coloneq e_k$, for $k = 0,\ldots,K$. Furthermore, we define the function values $\theta_{2k - 1} \coloneq \mu_k$, for $k = 1,\ldots,K$, and $\theta_{2k} \coloneq 0$. We then define the jump sizes $\Delta_{2k - 1} \coloneq \vert \mu_{k} \vert$ and $\Delta_{2k} \coloneq \vert \mu_{k} \vert$ for $k = 1,\ldots,K$. We also use the notation $K_\tau \coloneq 2K$.

In several of the following lemmas we make use of the following inequalities. Let $R = (r_1,\ldots,r_{\mn})$ be a polynomial of degree $\rQ$ for which 
\begin{equation}\label{eq:boundRemainderTrendT}
\sum_{l = 1}^{\mn} r_l^2 \leq \Cr \sigma_0^2 \log(n)
\end{equation}
holds. Suppose that 
\begin{equation}\label{eq:boundNoise}
\max_{1 \leq i \leq j \leq \mn} \varepsilon_{i:j}^\top H_{i:j} \varepsilon_{i:j}
\leq \CnoiseLin \sigma_0^2 \log(n).
\end{equation}

Lemma~\ref{lemma:costReductionEst} bounds the reduction in costs by having an estimated polynomial trend instead of the true parameter.
\begin{Lemma}\label{lemma:costReductionEst}
Suppose that our data follows Setting~\ref{setting:trend}. Suppose that \eqref{eq:boundRemainderTrendT}~and~\eqref{eq:boundNoise} hold. Then, there exists a constant $\CcostEst > 0$ such that for any $k \in \{0,\ldots,\Ktau\}$ and any $i,j$ such that $\tau_{k} < i \leq j \leq \tau_{k + 1}$, 
\begin{equation*}
C_{i, j}(\varepsilon + R + W) - C_{i, j}(Y - \hat{\theta}_{i:j}) \leq \CcostEst \sigma_0^2 \log(n).
\end{equation*} 
\end{Lemma}
\begin{proof}
Note that $y_l = \theta_k + t_l + w_l + \varepsilon_l,\ l = i, \ldots,j$. We further remark that $\theta_k + T$ is a polynomial of degree $\rQ$, so $\hat{\theta}_{i:j} = \theta_k + T_{i:j} + H_{i:j} W_{i:j} + H_{i:j} \varepsilon_{i:j}$. For the last equality before the inequalities note that $H_{i:j}$ is a projection matrix, i.e.~$H_{i:j}^\top=H_{i:j}$, $H_{i:j}^2 = H_{i:j}$. We obtain the first inequality from Lemma~\ref{lemma:mixedtermNoisePolynomial}, the Cauchy-Schwarz inequality, and $x^\top H_{i:j} x \leq x^\top x$ for any vector $x$ since the largest eigenvalue of $H_{i:j}$ is $1$ as $H_{i:j}$ is a projection matrix. The second inequality follows from $\sum_{l = i}^j r_l^2 \leq \sum_{l = 1}^{\mn} r_l^2$ and $\sum_{l = i}^j w_l^2 \leq \sum_{l = 1}^{\mn} w_l^2$. The final inequality is a direct consequence of \eqref{eq:boundRemainderTrendT},~\eqref{eq:boundNoise},~
and~\eqref{eq:remainderSeasonalyT}. In total, we obtain
\begin{align*}
& C_{i, j}(\varepsilon + R + W) - C_{i, j}(Y - \hat{\theta}_{i:j})\\
= & \sum_{l = i}^{j} \Big(\varepsilon_l + r_l + w_l)^2 - \sum_{l = i}^{j} (\varepsilon_l + w_l - \big(H_{i:j} \varepsilon_{i:j}\big)_l - \big(H_{i:j} W_{i:j}\big)_l\Big)^2\\
= & 2 \sum_{l = i}^{j} \varepsilon_l r_l + 2 \sum_{l = i}^{j} w_l r_l + \sum_{l = i}^j r_l^2 + \varepsilon_{i:j}^\top H_{i:j} \varepsilon_{i:j} + 2 \varepsilon_{i:j}^\top H_{i:j} W_{i:j} + W_{i:j}^\top H_{i:j} W_{i:j}\\
\leq & 2 \sqrt{\varepsilon_{i:j}^\top H_{i:j} \varepsilon_{i:j}}\sqrt{\sum_{l = i}^j r_l^2} + \sqrt{\sum_{l = i}^j r_l^2}\sqrt{\sum_{l = i}^{j} w_l^2} + \sum_{l = i}^j r_l^2\\
& + \varepsilon_{i:j}^\top H_{i:j} \varepsilon_{i:j} + 2 \sqrt{\varepsilon_{i:j}^\top H_{i:j} \varepsilon_{i:j}} \sqrt{\sum_{l = i}^{j} w_l^2} + \sum_{l = i}^{j} w_l^2\\
\leq & 2 \sqrt{\varepsilon_{i:j}^\top H_{i:j} \varepsilon_{i:j}}\sqrt{\sum_{l = 1}^{\mn} r_l^2} + \sqrt{\sum_{l = 1}^{\mn} r_l^2}\sqrt{\sum_{l = 1}^{\mn} w_l^2} + \sum_{l = 1}^{\mn} r_l^2\\
& + \varepsilon_{i:j}^\top H_{i:j} \varepsilon_{i:j} + 2 \sqrt{\varepsilon_{i:j}^\top H_{i:j} \varepsilon_{i:j}} \sqrt{\sum_{l = 1}^{\mn} w_l^2} + \sum_{l = 1}^{\mn} w_l^2\\
\leq & \big(2\sqrt{\CnoiseLin \Cr} + \sqrt{\Cr \CremainderSeasonal} + \Cr + \CnoiseLin + 2\sqrt{ \CnoiseLin \CremainderSeasonal} + \CremainderSeasonal \big) \sigma_0^2 \log(n) = \CcostEst \sigma_0^2 \log(n),
\end{align*}
where $\CcostEst \coloneq 2\sqrt{\CnoiseLin \Cr} + \sqrt{\Cr \CremainderSeasonal} + \Cr + \CnoiseLin + 2\sqrt{ \CnoiseLin \CremainderSeasonal} + \CremainderSeasonal$.
\end{proof}

Lemma~\ref{lemma:costExtraCp} bounds the increase in costs from having an extra estimated changepoint.
\begin{Lemma}\label{lemma:costExtraCp}
Suppose that our data follows Setting~\ref{setting:trend}. Suppose that \eqref{eq:boundNoise} hold. If $n$ is large enough, then for any $k \in \{0,\ldots,K\}$ and any $a,b,c$ such that $\tau_{k} < a \leq b \leq c \leq \tau_{k + 1}$, then there exits a constant $\CcostExtraCp > 0$ such that
\begin{equation*}
\begin{split}
& C_{a + 1,b}(Y - \hat{\theta}_{a+1:b}) + C_{b+1,c}(Y - \hat{\theta}_{b+1:c}) + \frac{1}{2}\lambda_{\operatorname{coll}}\, \hat{\sigma}_0^2\log(n)^{1 + \delta}\\
& - C_{a + 1,c}(Y - \hat{\theta}_{a+1:c}) \geq \CcostExtraCp\,  \sigma_0^2 \log(n)^{1 + \delta}.
\end{split}
\end{equation*}
\end{Lemma}
\begin{proof}
Similarly to the proof of Lemma~\ref{lemma:costReductionEst}, we have that
\begin{align*}
& C_{a + 1,b}(Y - \hat{\theta}_{a+1:b}) + C_{b+1,c}(Y - \hat{\theta}_{b+1:c}) + \frac{1}{2}\lambda_{\operatorname{coll}}\, \hat{\sigma}_0^2\log(n)^{1 + \delta} - C_{a + 1,c}(Y - \hat{\theta}_{a+1:c})\\
= & \sum_{l = a + 1}^{b} \big(y_l - (H_{a+1:b} y_{a+1:b})_l\big)^2 + \sum_{l = b + 1}^{c} \big(y_l - (H_{b+1:c} y_{b+1:c})_l\big)^2 + \frac{1}{2}\lambda_{\operatorname{coll}}\, \hat{\sigma}_0^2\log(n)^{1 + \delta}\\
& - \sum_{l = a + 1}^{c} \big(y_l - (H_{a+1:c} y_{a+1:c})_l\big)^2\\
= & \sum_{l = a + 1}^{b} \big(\varepsilon_l + w_l - (H_{a+1:b} \varepsilon_{a+1:b})_l - (H_{a+1:b} W_{a+1:b})_l\big)^2\\
& + \sum_{l = b + 1}^{c} \big(\varepsilon_l + w_l - (H_{b+1:c} \varepsilon_{b+1:c})_l - (H_{b+1:c} W_{b+1:c})_l\big)^2 + \frac{1}{2}\lambda_{\operatorname{coll}}\, \hat{\sigma}_0^2\log(n)^{1 + \delta}\\
& - \sum_{l = a + 1}^{c} \big(\varepsilon_l + w_l - (H_{a+1:c} \varepsilon_{a+1:c})_l - (H_{a+1:c} W_{a+1:c})_l\big)^2\\
\geq & \frac{1}{2} \lambda_{\operatorname{coll}}\, \CvarianceEstimateLower  \, \sigma_0^2\log(n)^{1 + \delta} - 3 \big( \CnoiseLin + 2\sqrt{ \CnoiseLin \CremainderSeasonal} + \CremainderSeasonal \big) \sigma_0^2 \log(n)\\
\geq & \CcostExtraCp \sigma_0^2\log(n)^{1 + \delta},
\end{align*}
where the final inequality holds for large enough $n$. 
\end{proof}

Lemma~\ref{lemma:meanMissCp} lower bounds the quadratic difference between a step function and any polynomial of degree $Q$.
\begin{Lemma}\label{lemma:meanMissCp}
For any $Q \in \mathbb{N}$ there exists a constant $\Cmiss > 0$ only depending on $Q$ such that for any polynomial $f$ of degree $Q$, any step function with changepoint $\tau$ and function values $\theta_0$ and $\theta_1$, and any distance $\kappa_0 \geq \max\{\kappa_1, 4^{Q + 1}\}$, we have that
\begin{equation*}
\sum_{l = \tau - \kappa_0 + 1}^{\tau} \left( f\left(\frac{l}{n}\right) - \theta_{0} \right)^2 + \sum_{l = \tau + 1}^{\tau + \kappa_1} \left( f\left(\frac{l}{n}\right) - \theta_{1}\right)^2 \geq \Cmiss \kappa_1 (\theta_1 - \theta_0)^2.
\end{equation*}
The statement also holds with the roles of $\kappa_0$ and $\kappa_1$ reversed.
\end{Lemma}
\begin{proof}
Without loss of generality we assume that $\theta_0 = 0$ and $\theta_{1} > 0$, as subtracting a constant or multiplying everything inside the squares with $-1$ do not change the result. Furthermore, without loss of generality assume that $\kappa_0$ and $\kappa_1$ are as given in the lemma. Otherwise consider the observations in reverse order which is not changing the costs.

We now distinguish the cases $\kappa_1 < 4^{Q + 1}$ and $\kappa_1 \geq 4^{Q + 1}$. If $\kappa_1 < 4^{Q + 1}$, then we distinguish the cases whether
\begin{equation}\label{eq:conditionOnFatTauPlus1}
f\left(\frac{\tau + 1}{n}\right) - \theta_0 \geq \frac{\theta_{1} - \theta_{0}}{2}
\end{equation}
holds or not. If it does not hold, then
\begin{align}
\sum_{l = \tau + 1}^{\tau + \kappa_1} \left( f\left(\frac{l}{n}\right) - \theta_{1}\right)^2
\geq \left( f\left(\frac{\tau + 1}{n}\right) - \theta_{1}\right)^2
\geq \frac{1}{4} (\theta_{1} - \theta_{0})^2.
\end{align}
So the statement holds with $\Cmiss = 4^{-Q - 1}$.

If $\kappa_1 \geq 4^{Q + 1}$, then we distinguish the cases whether
\begin{equation}\label{eq:conditionOnFatTauPlus0p5}
f\left(\frac{\tau + 0.5}{n}\right) - \theta_0 \geq \frac{\theta_{1} - \theta_{0}}{2}
\end{equation}
holds or not. If it does not hold, then we continue with proof below but we consider the observations after $\tau$ instead of the ones before.

The remaining proof uses the following idea. We consider a large number of small boxes around the step function. If $f$ is fully outside of any of these boxes, then the statement holds with $\Cmiss$ resembling the size of the box. If $f$ is inside all boxes for at least one time point, then by using repeatedly the mean value theorem we will show that the $(Q + 1)$-th derivative of $f$ must be non-zero and hence $f$ cannot be a polynomial of degree $Q$. 

We start by constructing the boxes. For any $\iota = (i_1,\ldots,i_m) \in \{0,1\}^m$, $m \in \N$, let $I_{\iota} \coloneq \left( L_{\iota}, R_{\iota} \right]$ with
\begin{equation*}
L_{\iota} = \tau - \kappa_0 + \sum_{l = 1}^{m} 3 i_l \frac{\kappa_0}{4^l} \text{ and } R_\iota = L_\iota + \frac{\kappa_0}{4^m}.
\end{equation*}
For notational convenience we extend the notation to the empty set and define $I_{\emptyset} \coloneq (\tau - \kappa_0, \tau]$. Hence, $I_{(i_1,\ldots,i_m)}$ is of length $\kappa_0 4^{-m}$ and $I_{(i_1,\ldots,i_{m - 1})}$ is the union of $I_{(i_1,\ldots,i_{m - 1}, 0)}$, $I_{(i_1,\ldots,i_{m - 1}, 1)}$, and the interval in between, which is of length $2\kappa_0 4^{-m}$. 

Let
\begin{equation*}
b \coloneq \sum_{l = 2}^{Q+2} 2^l.
\end{equation*}
If there exists an $\iota \in \{0,1\}^Q$ such that for all $x \in I_\iota$ we have that 
\begin{equation*}
f \not\in \left[ \theta_0 - \frac{\theta_{1} - \theta_{0}}{b}, \theta_0 + \frac{\theta_{1} - \theta_{0}}{b} \right],
\end{equation*}
then
\begin{equation*}
\sum_{l = \lceil L_{\iota} \rceil}^{\lfloor R_{\iota} \rfloor } \left( f\left(\frac{l}{n}\right) - \theta_{0} \right)^2 \geq \big(\lfloor R_{\iota} \rfloor - \lceil L_{\iota} \rceil \big) \frac{(\theta_{1} - \theta_{0})^2}{b^2} = \Cmiss \kappa_0 (\theta_1 - \theta_0)^2,
\end{equation*}
with $\Cmiss > 0$ being a constant only depending on $Q$, since $\lfloor R_{\iota} \rfloor - \lceil L_{\iota} \rceil \geq \frac{\kappa_0}{4^Q} - 1$ and $\kappa_0 \geq 4^{Q + 1}$ ensuring that $I_\iota$ contains at least one observation.

Otherwise, we show by induction that for every $m \in \{Q,\ldots,0\}$ and for every $\iota \in \{0,1\}^m$ there exists an $x_{\iota} \in I_{\iota}$ for which
\begin{equation*}
f^{(Q - m)}\big(x_{\iota}\big) \in \left[ -  \frac{\theta_{1} - \theta_{0}}{b} \prod_{l = 1}^{Q - m} \frac{4^{Q - l + 1}}{\kappa_0} , \frac{\theta_{1} - \theta_{0}}{b} \prod_{l = 1}^{Q - m} \frac{4^{Q - l + 1}}{\kappa_0}  \right].
\end{equation*}
Furthermore, for every $m \in \{Q,\ldots, -1\}$ there exists a $\xi_m \in (x_{\mathbf{1}_{m + 1}}, \tau + 1]$, where $\mathbf{1}_m = (1,\ldots,1) \in \{0,1\}^m$ and $\xi_Q = \tau + 0.5$ or $\xi_Q = \tau + 1$, depending on whether \eqref{eq:conditionOnFatTauPlus1} or \eqref{eq:conditionOnFatTauPlus0p5} holds, for which
\begin{equation*}
 f^{(Q - m)}\big(\xi_m\big) \geq \left(1 - \sum_{l = 1}^{Q - m} \frac{2^l}{b}\right) \frac{\theta_{1} - \theta_{0}}{2^{Q - m + 1}} \prod_{l = 1}^{Q - m} \frac{4^{Q - l + 1}}{\kappa_0} .
\end{equation*}

For the first statement, the induction start for $m = Q$ is given by construction of the boxes and that $f$ is inside all boxes for at least one point. Suppose that the statement is true for an $m$. Then, by the mean value theorem we have that there exists a point 
\begin{equation*}
x_{(i_1,\ldots,i_{m - 1})} \in \left[x_{(i_1,\ldots,i_{m - 1}, 0)}, x_{(i_1,\ldots,i_{m - 1}, 1)}\right] \subseteq \left[L_{(i_1,\ldots,i_{m - 1})}, R_{(i_1,\ldots,i_{m - 1})}\right]
\end{equation*}
for which
\begin{align*}
& f^{(Q - m + 1)}\big(x_{(i_1,\ldots,i_{m - 1})}\big)
= \frac{f^{(Q - m)}\big(x_{(i_1,\ldots,i_{m - 1}, 1)}\big) - f^{(Q - m)}\big(x_{(i_1,\ldots,i_{m - 1}, 0)}\big)}{x_{(i_1,\ldots,i_{m - 1}, 1)} - x_{(i_1,\ldots,i_{m - 1}, 0)}}\\
\leq & \frac{2  \frac{\theta_{1} - \theta_{0}}{b} \prod_{l = 1}^{Q - m} \frac{4^{Q - l + 1}}{\kappa_0} }{2\frac{\kappa_0}{4^m}}
= \frac{\theta_{1} - \theta_{0}}{b}\prod_{l = 1}^{Q - m + 1}\frac{4^{Q - l + 1}}{\kappa_0},
\end{align*}
since by construction the distance between the two points is at least $2\kappa_0 4^{-m}$. The lower bound follows from the same arguments but with a minus in front. Hence, the first statement holds for $m - 1$. 

For the second statement, the induction start is given by $f(\xi_Q) \geq \frac{\theta_{1} - \theta_{0}}{2}$ since \eqref{eq:conditionOnFatTauPlus1} or \eqref{eq:conditionOnFatTauPlus0p5} holds. Suppose that the statement holds for an $m$. Then, by the mean value theorem we have that there exists a point 
\begin{equation*}
\xi_{m - 1} \in \left(x_{\mathbf{1}_{m}}, \xi_{m}\right) \subset \left(x_{\mathbf{1}_{m}}, \tau + 1\right]
\end{equation*}
for which
\begin{align*}
&  f^{(Q - m + 1)}\big(\xi_{m - 1}\big) 
= \frac{f^{(Q - m)}\big(\xi_{m}\big) - f^{(Q - m)}\big(x_{\mathbf{1}_{m}}\big)}{\xi_{m} - x_{\mathbf{1}_{m}}}\\
\geq & \frac{\left(1 - \sum_{l = 1}^{Q - m} \frac{2^l}{b}\right) \frac{\theta_{1} - \theta_{0}}{2^{Q - m}} \prod_{l = 1}^{Q - m} \frac{4^{Q - l + 1}}{\kappa_0} -  \frac{\theta_{1} - \theta_{0}}{b} \prod_{l = 1}^{Q - m} \frac{4^{Q - l + 1}}{\kappa_0}}{2\frac{\kappa_0}{4^m}}\\
\geq & \left(1 - \sum_{l = 1}^{Q - m + 1} \frac{2^l}{b}\right)\frac{\theta_{1} - \theta_{0}}{2^{Q - m + 1}}\prod_{l = 1}^{Q - m + 1} \frac{4^{Q - l + 1}}{\kappa_0} ,
\end{align*}
since by construction the distance between the two points is at most $\kappa_0 4^{-m} + 1 \leq 2 \kappa_0 4^{-m}$. This completes the induction step.

All in all, we have shown that there exists a $\xi_{-1} \in (\tau - \kappa_0, \tau + 1]$ for which, by the definition of $b$,
 \begin{equation*}
 f^{(Q +1)}\big(\xi_{-1}\big) \geq \frac{1}{2} \frac{\theta_{1} - \theta_{0}}{2^{Q + 2}} \prod_{l = 1}^{Q + 1} \frac{4^{Q - l + 1}}{\kappa_0} > 0.
\end{equation*}
Thus, the $(Q + 1)$-th derivative is non-zero and hence $f$ cannot be a polynomial of degree $Q$. This is a contradiction. Hence $f$ must be in the boxes and the statement holds as shown before.
\end{proof}

Lemma~\ref{lemma:costMissingCp} bounds the increase in costs from missing a changepoint.
\begin{Lemma}\label{lemma:costMissingCp}
Suppose that our data follows Setting~\ref{setting:trend}. Let $\hat{T} = X \hat{\beta}$ be any polynomial of degree $\rQ$. Suppose that \eqref{eq:boundRemainderTrendT}~and~\eqref{eq:boundNoise}  hold.
If $n$ is large enough, then there exists a constant $\CextraCP > 0$ such that for any $k \in \{1,\ldots,2K\}$, 
\begin{equation*}
\begin{split}
& C_{\tau_k - \kappa_k + 1,\tau_k + \kappa_k}(Y - \hat{T}) - C_{\tau_k - \kappa_k + 1,\tau_k}(\varepsilon + R + W) - C_{\tau_k + 1,\tau_k + \kappa_k}(\varepsilon + R + W)\\
& - \frac{1}{2}\lambda_{\operatorname{coll}}\, \hat{\sigma}_0^2\log(n)^{1 + \delta}
\geq  \CextraCP \sigma_0^2 \log(n)^{1 + \delta}.
\end{split}
\end{equation*}
\end{Lemma}
\begin{proof}
Similarly to the proof of Lemma~\ref{lemma:costReductionEst}, we obtain for the second and third term that
\begin{align*}
& C_{\tau_k - \kappa_k + 1,\tau_k}(\varepsilon + R + W) + C_{\tau_k + 1,\tau_k + \kappa_k}(\varepsilon + R + W)\\
= & \sum_{l = \tau_k - \kappa_k + 1}^{\tau_k} \left( \varepsilon_l + r_l + w_l \right)^2 + \sum_{l = \tau_k + 1}^{\tau_k + \kappa_k} \left( \varepsilon_l + r_l + w_l \right)^2\\
 \leq & \sum_{l = \tau_k - \kappa_k + 1}^{\tau_k + \kappa_k} (\varepsilon_l + w_l)^2 + 2\left( \Cr + 2\sqrt{\CnoiseLin \Cr} + 2\sqrt{\CremainderSeasonal \Cr} \right)\sigma_0^2 \log(n).
\end{align*}
For the first term, we have that
\begin{align*}
& C_{\tau_k - \kappa_k + 1,\tau_k + \kappa_k}(Y - \hat{T})
= \sum_{l = \tau_k - \kappa_k + 1}^{\tau_k + \kappa_k} \left( y_l - \hat{t}_l \right)^2\\
= & \sum_{l = \tau_k - \kappa_k + 1}^{\tau_k} \left(\varepsilon_l + \theta_{k - 1} + w_l + t_l - \hat{t}_l \right)^2 + \sum_{l = \tau_k + 1}^{\tau_k + \kappa_k} \left(\varepsilon_l + \theta_{k} + w_l + t_l - \hat{t}_l\right)^2\\
= & \sum_{l = \tau_k - \kappa_k + 1}^{\tau_k + \kappa_k} (\varepsilon_l + w_l)^2 + 2\sum_{l = \tau_k - \kappa_k + 1}^{\tau_k} (\varepsilon_l + w_l) \left( \theta_{k - 1} + t_l - \hat{t}_l \right) + 2\sum_{l = \tau_k + 1}^{\tau_k + \kappa_k} (\varepsilon_l + w_l) \left( \theta_{k} + t_l - \hat{t}_l \right)\\
&  + \sum_{l = \tau_k - \kappa_k + 1}^{\tau_k} \left( \theta_{k - 1} + t_l - \hat{t}_l \right)^2 + \sum_{l = \tau_k + 1}^{\tau_k + \kappa_k} \left( \theta_{k} + t_l - \hat{t}_l \right)^2.
\end{align*}

Note that $\theta_{k - 1} + T - \hat{T}$ and $\theta_{k} + T - \hat{T}$ are polynomials of degree $\rQ$. Hence, Lemma~\ref{lemma:mixedtermNoisePolynomial}, $\sqrt{a} + \sqrt{b} \leq \sqrt{2a + 2b}$ for any $a,b>0$, and \eqref{eq:boundNoise} yield
\begin{align*}
& 2\sum_{l = \tau_k - \kappa_k + 1}^{\tau_k} \varepsilon_l \left( \theta_{k - 1} + t_l - \hat{t}_l \right) + 2\sum_{l = \tau_k + 1}^{\tau_k + \kappa_k} \varepsilon_l \left( \theta_{k} + t_l - \hat{t}_l \right)\\
\geq & -2 \sqrt{\varepsilon_{\tau_k - \kappa_k + 1:\tau_k} H_{\tau_k - \kappa_k + 1:\tau_k} \varepsilon_{\tau_k - \kappa_k + 1:\tau_k} } \sqrt{\sum_{l = \tau_k - \kappa_k + 1}^{\tau_k} \left( \theta_{k - 1} + t_l - \hat{t}_l \right)^2}\\
& -2 \sqrt{\varepsilon_{ \tau_k + 1:\tau_k + \kappa_k} H_{ \tau_k + 1:\tau_k + \kappa_k} \varepsilon_{ \tau_k + 1:\tau_k + \kappa_k} } \sqrt{\sum_{l = \tau_k + 1}^{\tau_k + \kappa_k} \left( \theta_{k} + t_l - \hat{t}_l \right)^2}\\
\geq & -2 \sqrt{\CnoiseLin \sigma_0^2 \log(n)} \left( \sqrt{\sum_{l = \tau_k - \kappa_k + 1}^{\tau_k} \left( \theta_{k - 1} + t_l - \hat{t}_l \right)^2} + \sqrt{\sum_{l = \tau_k + 1}^{\tau_k + \kappa_k} \left( \theta_{k} + t_l - \hat{t}_l \right)^2} \right)\\
\geq & -2 \sqrt{\CnoiseLin \sigma_0^2 \log(n)}\sqrt{2\sum_{l = \tau_k - \kappa_k + 1}^{\tau_k} \left( \theta_{k - 1} + t_l - \hat{t}_l \right)^2 + 2\sum_{l = \tau_k + 1}^{\tau_k + \kappa_k} \left( \theta_{k} + t_l - \hat{t}_l \right)^2}.
\end{align*}

Similarly, the Cauchy-Schwarz inequality and \eqref{eq:remainderSeasonalyT} yield
\begin{align*}
& 2\sum_{l = \tau_k - \kappa_k + 1}^{\tau_k} w_l \left( \theta_{k - 1} + t_l - \hat{t}_l \right) + 2\sum_{l = \tau_k + 1}^{\tau_k + \kappa_k} w_l \left( \theta_{k} + t_l - \hat{t}_l \right)\\
\geq & -2 \sqrt{\CremainderSeasonal \sigma_0^2 \log(n)}\sqrt{2\sum_{l = \tau_k - \kappa_k + 1}^{\tau_k} \left( \theta_{k - 1} + t_l - \hat{t}_l \right)^2 + 2\sum_{l = \tau_k + 1}^{\tau_k + \kappa_k} \left( \theta_{k} + t_l - \hat{t}_l \right)^2}.
\end{align*}

Moreover, note that $\hat{T} - T$ is a polynomial of degree $\rQ$. We can apply Lemma~\ref{lemma:meanMissCp}, since $\kappa_k \geq 4^{Q + 1}$. Thus, Lemma~\ref{lemma:meanMissCp} yields
\begin{align*}
\sum_{l = \tau_k - \kappa_k + 1}^{\tau_k} \left( \theta_{k - 1} + t_l - \hat{t}_l \right)^2 + \sum_{l = \tau_k + 1}^{\tau_k + \kappa_k} \left( \theta_{k} + t_l - \hat{t}_l \right)^2
\geq  \Cmiss \kappa_k \Delta_k^2.
\end{align*}

Therefore, all in all,
\begin{align*}
& C_{\tau_k - \kappa_k + 1,\tau_k + \kappa_k}(Y - \hat{T}) - C_{\tau_k - \kappa_k + 1,\tau_k}(\varepsilon + R + W) - C_{\tau_k + 1,\tau_k + \kappa_k}(\varepsilon + R + W)\\
& - \frac{1}{2}\lambda_{\operatorname{coll}}\, \hat{\sigma}_0^2\log(n)^{1 + \delta}\\
\geq & \sum_{l = \tau_k - \kappa_k + 1}^{\tau_k} \left( \theta_{k - 1} + t_l - \hat{t}_l \right)^2 + \sum_{l = \tau_k + 1}^{\tau_k + \kappa_k} \left( \theta_{k} + t_l - \hat{t}_l \right)^2\\
& -2 \big(\sqrt{\CnoiseLin} + \sqrt{\CremainderSeasonal} \big) \sqrt{\sigma_0^2 \log(n)}\sqrt{2\sum_{l = \tau_k - \kappa_k + 1}^{\tau_k} \left( \theta_{k - 1} + t_l - \hat{t}_l \right)^2 + 2\sum_{l = \tau_k + 1}^{\tau_k + \kappa_k} \left( \theta_{k} + t_l - \hat{t}_l \right)^2}\\
& - 2\left( \Cr + 2\sqrt{\CnoiseLin \Cr} + 2\sqrt{\CremainderSeasonal \Cr} \right)\sigma_0^2 \log(n)  - \frac{1}{2}\lambda_{\operatorname{coll}}\, \hat{\sigma}_0^2\log(n)^{1 + \delta}\\
\geq & \frac{1}{2}\Cmiss \kappa_k \Delta_k^2 - 2\left( \Cr + 2\sqrt{\CnoiseLin \Cr} + 2\sqrt{\CremainderSeasonal \Cr} \right)\sigma_0^2 \log(n) - \frac{1}{2}\lambda_{\operatorname{coll}}\, \CvarianceEstimate \sigma_0^2\log(n)^{1 + \delta},\\
\geq & \CextraCP \sigma_0^2 \log(n)^{1 + \delta},
\end{align*}
where the last three inequalities follow from the definitions of $\kappa_k$, $\Clambda$ and $n$ being large enough. 
\end{proof}

We can now bound the costs of CAPA between a good and any other trend estimate.
\begin{Lemma}\label{lemma:differenceCosts}
Suppose that our data follows Setting~\ref{setting:trend}. Let $\hat{T} = X \hat{\beta}$ be any polynomial of degree $\rQ$. Let $\tilde{T}$ be a polynomial of degree $\rQ$ such that \eqref{eq:boundRemainderTrendT} holds for $R = (r_1,\ldots,r_{\mn}) \coloneq \tilde{T} - T$. Suppose that \eqref{eq:boundNoise} holds. If $n$ is large enough, then for any $\hat{\mathcal{A}} \not\in \mathcal{B}$, $\hat{\mu}$ we have that
\begin{equation*}
C\big(Y - \hat{T}, \hat{\mathcal{A}}, \hat{\mu}\big) - C\big(Y - \tilde{T}, \mathcal{A}, \mu\big) - \max_{Q \in \mathbf{Q}} Q\, \hat{\sigma}_0 \log(n) > 0.
\end{equation*}
\end{Lemma}
\begin{proof}
There exists a constant $\tilde{C}_1 > 0$ such that
\begin{equation}\label{eq:boundQsigmalogn}   
\max_{Q \in \mathbf{Q}} Q\, \hat{\sigma}_0 \log(n) \leq \tilde{C}_1 \sigma_0 \log(n),
\end{equation}
since $\mathbf{Q}$ is a finite set and $\hat{\sigma}_0 \leq \CvarianceEstimate \sigma_0$. Let
\begin{equation*}
\hat{\mathcal{A}} \coloneq \left\{(\hat{s}_1,\hat{e}_1),\ldots,(\hat{s}_{\hat{K}}, \hat{e}_{\hat{K}})\right\},
\end{equation*}
with $\hat{K}$ being the number of anomalies, $\hat{s}_1,\ldots,\hat{s}_{\hat{K}}$ and $\hat{e}_1,\ldots,\hat{e}_{\hat{K}}$ being the points before the anomalies start and the last point of the anomalies, respectively. We require $0 \leq \hat{s}_1 < \hat{e}_1 \leq \hat{s}_2 < \cdots < \hat{e}_{\hat{K} - 1} \leq \hat{s}_{\hat{K}} < \hat{e}_{\hat{K}} \leq \mn$. Furthermore, we define $\hat{e}_0 = 0$ and $\hat{s}_{\hat{K} + 1} = \mn$.

Furthermore, we rewrite them as a changepoint model: We use the notation $\hat{K}_\tau \coloneq 2\hat{K}$. Furthermore, we define $\hat{\tau}_{2k} \coloneq \hat{e}_k$ and $\hat{\tau}_{2k + 1} \coloneq \hat{s}_k$, for $k = 0,\ldots,\hat{K}$. Moreover, we have function values $\hat{\theta}_{2k - 1} \coloneq \hat{\mu}_k$, for $k = 1,\ldots,\hat{K}$, and $\hat{\theta}_{2k} \coloneq 0$, for $k = 0,\ldots,\hat{K}$. 

Using these definitions, we have that
\begin{align*}
C\big(Y - \hat{T}, \hat{\mathcal{A}}, \hat{\mu}\big)
= & \sum_{k = 0}^{\hat{K}} C_{\hat{e}_k + 1,\hat{s}_{k + 1}}(Y - \hat{T}) + \sum_{k = 1}^{\hat{K}} C_{\hat{s}_k + 1, \hat{e}_k}(Y - \hat{T} - \hat{\mu}_k) + \hat{K} \lambda_{\operatorname{coll}}\, \hat{\sigma}_0 \log(n)^{1 + \delta}\\
= & \sum_{k = 0}^{\hat{K}_\tau} C_{\hat{\tau}_k + 1, \hat{\tau}_{k + 1}}(Y - \hat{T} - \hat{\theta}_k) + \frac{\hat{K}_\tau}{2} \lambda_{\operatorname{coll}}\, \hat{\sigma}_0 \log(n)^{1 + \delta}.
\end{align*}

Furthermore, using the definitions for $\mathcal{A}$, we have that
\begin{align*}
C\big(Y - \tilde{T}, \mathcal{A}, \mu\big)
= & \sum_{k = 0}^{K} C_{e_k + 1,s_{k + 1}}(Y - \tilde{T}) + \sum_{k = 1}^{K} C_{s_k + 1, e_k}(Y - \tilde{T} - \mu_k) + K \lambda_{\operatorname{coll}}\, \hat{\sigma}_0 \log(n)^{1 + \delta}\\
= & \sum_{k = 0}^{K_\tau} C_{\tau_k + 1, \tau_{k + 1}}(Y - \tilde{T} -  \theta_k) + \frac{K_\tau}{2} \lambda_{\operatorname{coll}}\, \hat{\sigma}_0 \log(n)^{1 + \delta}.
\end{align*}

We will split the costs at $\tau_k \pm \kappa_k$. Note that Assumption~\ref{assumption:segmentLengthTrend} yields that $\tau_k - \kappa_k > \tau_{k - 1}$ and $\tau_k + \kappa_k < \tau_{k + 1}$ for large enough $n$.

We further remark that $\hat{T} + \hat{\theta}_k$ is a polynomial of degree $\rQ$ and hence by definition for any $k$, for any $\hat{\tau}_k \leq \tau \leq \hat{\tau}_{k + 1}$ we have that
\begin{align*}
C_{\hat{\tau}_k + 1, \hat{\tau}_{k + 1}}(Y - \hat{T} - \hat{\theta}_k)
 \geq & C_{\hat{\tau}_k + 1, \hat{\tau}_{k + 1}}(Y - \hat{\theta}_{\hat{\tau}_k + 1, \hat{\tau}_{k + 1}})\\
  \geq & C_{\hat{\tau}_k + 1, \tau}(Y - \hat{\theta}_{\hat{\tau}_k + 1:\tau}) + C_{\tau + 1, \hat{\tau}_{k + 1}}(Y - \hat{\theta}_{\tau + 1:\hat{\tau}_{k + 1}}).
\end{align*}

Furthermore, $Y - \tilde{T} = \varepsilon + A + R + W$. Hence, on the set $\{\tau_k + 1, \tau_{k + 1}\}$, we have that $Y - \tilde{T} - \theta_k = \varepsilon + R + W$. Hence, for any $k$, for any $\tau_k \leq \tau \leq \tau_{k + 1}$ we have that
\begin{equation*}
\begin{split}
& C_{\tau_k + 1, \tau_{k + 1}}(Y - \tilde{T} - \theta_k) = C_{\tau_k + 1, \tau_{k + 1}}(\varepsilon + W)\\
= & C_{\tau_k + 1, \tau}(\varepsilon + R + W) + C_{\tau + 1, \tau_{k + 1}}(\varepsilon + R + W).
\end{split}
\end{equation*}
For $k \in \{0,\ldots,K_\tau\}$, let $\tilde{K}_k$ be the number of estimated change-points between $\tau_k + \kappa_k$ and $\tau_{k + 1} - \kappa_{k + 1}$. That is, there exists an $l$ such that
\begin{equation*}
\hat{\tau}_{l} < \tau_k + \kappa_k \leq \hat{\tau}_{l + 1} < \cdots < \hat{\tau}_{l + \tilde{K}_k} \leq \tau_{k + 1} - \kappa_{k + 1} < \hat{\tau}_{l + \tilde{K}_k + 1}.
\end{equation*}

We write $\tilde{\tau}_{k, i} \coloneq \hat{\tau}_{l + i}$ for $i = 1,\ldots,\tilde{K}_k$. Moreover, we write $\tilde{\tau}_{k, 0} \coloneq \tau_k + \kappa_k$ and $\tilde{\tau}_{k, \tilde{K}_k + 1} \coloneq \tau_{k + 1} - \kappa_{k + 1}$. Note that $k = 0$ and $k = K_\tau$ are special boundary cases, for which we have $\kappa_0 = \kappa_{2K+1} = 0$.

Similarly, for $k \in \{1,\ldots,K_\tau\}$, let $\hat{K}_k$ be the number of estimated change-points between $\tau_k - \kappa_k$ and $\tau_k + \kappa_k$. That is, there exists an $l$ such that
\begin{equation*}
\hat{\tau}_{l} \leq \tau_k - \kappa_k < \hat{\tau}_{l + 1} < \cdots < \hat{\tau}_{l + \hat{K}_k} < \tau_k + \kappa_k \leq \hat{\tau}_{l + \hat{K}_k + 1}.
\end{equation*}

We write $\hat{\tau}_{k, i} \coloneq \hat{\tau}_{l + i}$ for $i = 1,\ldots,\hat{K}_k$. Moreover, we write $\hat{\tau}_{k, 0} \coloneq \tau_k - \kappa_k $ and $\hat{\tau}_{k, \hat{K}_k + 1} \coloneq \tau_k - \kappa_k$.

Using these definitions, we have that
\begin{align*}
& C\big(Y - \hat{T}, \hat{\mathcal{A}}, \hat{\mu}\big) - C\big(Y - \tilde{T}, \mathcal{A}, \mu\big)
\geq   \sum_{k = 0}^{K_\tau} T_{1,k} + \sum_{k = 1}^{K_\tau} T_{2,K},
\end{align*}
where 
\begin{equation}\label{eq:defT1k}
\begin{split}
T_{1,k} \coloneq & \sum_{l = 0}^{\tilde{K}_k} C_{\tilde{\tau}_{k, l} + 1,\tilde{\tau}_{k, l + 1}}(Y - \hat{\theta}_{\tilde{\tau}_{k, l} + 1:\tilde{\tau}_{k, l + 1}}) - C_{\tilde{\tau}_{k, 0} + 1:\tilde{\tau}_{k, \tilde{K}_k + 1}}(\varepsilon + R + W)\\
& + \frac{\tilde{K}_k}{2} \lambda_{\operatorname{coll}}\, \hat{\sigma}_0 \log(n)^{1+\delta},
\end{split}
\end{equation}
and
\begin{equation}\label{eq:defT2k}
\begin{split}
T_{2,k} \coloneq & \sum_{l = 0}^{\hat{K}_k} C_{\hat{\tau}_{k, l} + 1,\hat{\tau}_{k, l + 1}}(Y - \hat{\theta}_{\hat{\tau}_{k, l} + 1:\hat{\tau}_{k, l + 1}}) - C_{\hat{\tau}_{k, 0} + 1:\tau_{k}}(\varepsilon + R + W)\\
& - C_{\tau_k + 1:\hat{\tau}_{k, \hat{K}_k + 1}}(\varepsilon + R + W) + \frac{\hat{K}_k - 1}{2} \lambda_{\operatorname{coll}}\, \hat{\sigma}_0 \log(n)^{1+\delta}.
\end{split}
\end{equation}

Lemmas~\ref{lemma:T1k}~and~\ref{lemma:T2k} bound $T_{1,k}$ and $T_{2,k}$, respectively.

If there exists a $k \in \{1,\ldots,K_\tau\}$ such that $\hat{\tau}_1,\ldots,\hat{\tau}_{\hat{K}}\not\in \{\tau_k - \kappa_k,\ldots, \tau_k + \kappa_k\}$, then $\hat{K}_k = 0$. Note that if $\hat{K} < K$, then such a $k$ must exist. Consequently, Lemmas~\ref{lemma:T1k}~and~\ref{lemma:T2k} and \eqref{eq:boundQsigmalogn} imply that
\begin{align*}
& C\big(Y - \hat{T}, \hat{\mathcal{A}}, \hat{\mu}\big) - C\big(Y - \tilde{T}, \mathcal{A}, \mu\big)
\geq  \sum_{k = 0}^{K_\tau} T_{1,k} + \sum_{k = 1}^{K_\tau} T_{2,K}\\
\geq &  -(2K + 1) \CcostEst \sigma_0^2 \log(n) + \CextraCP \sigma_0^2 \log(n)^{1 + \delta} - \max\{2K - 1, 0\} 3 \CcostEst \sigma_0^2 \log(n)\\
\geq & \CextraCP \sigma_0^2 \log(n)^{1 + \delta} - \max\{8K - 2,1\}\CcostEst \sigma_0^2 \log(n)
>  \max_{Q \in \mathbf{Q}} Q\, \hat{\sigma}_0 \log(n),
\end{align*}
where the final inequality follows from $n$ being large enough.

Furthermore, if $\hat{K} > K$, then there exists a $k \in \{0,\ldots,K_\tau\}$ such that $\tilde{K}_k \geq 1$ or there exists a $k \in \{1,\ldots,K_\tau\}$ such that $\hat{K}_k > 1$. Consequently, Lemmas~\ref{lemma:T1k}~and~\ref{lemma:T2k} and \eqref{eq:boundQsigmalogn} imply that
\begin{align*}
& C\big(Y - \hat{T}, \hat{\mathcal{A}}, \hat{\mu}\big) - C\big(Y - \tilde{T}, \mathcal{A}, \mu\big)
\geq  \sum_{k = 0}^{K_\tau} T_{1,k} + \sum_{k = 1}^{K_\tau} T_{2,K}\\
\geq &  \frac{1}{2} \CcostExtraCp\, \sigma_0^2 \log(n)^{1 + \delta} -2 K \CcostEst \sigma_0^2 \log(n) - 2 K 3 \CcostEst \sigma_0^2 \log(n)\\
\geq & \frac{1}{2} \CcostExtraCp\, \sigma_0^2 \log(n)^{1 + \delta} - 8 K \CcostEst \sigma_0^2 \log(n) >  \max_{Q \in \mathbf{Q}} Q\, \hat{\sigma}_0 \log(n),
\end{align*}
where the final inequality follows from $n$ being large enough.

All in all, if $\hat{\mathcal{A}}\not\in\hat{\mathcal{B}}$, then $C\big(Y - \hat{T}, \hat{\mathcal{A}}, \hat{\mu}\big) - C\big(Y - \tilde{T}, \mathcal{A}, \mu\big) - \max_{Q \in \mathbf{Q}} Q\, \hat{\sigma}_0 \log(n) > 0$, as $\hat{\mathcal{A}}\not\in\hat{\mathcal{B}}$ implies that one of the above cases must occur.
\end{proof}

\begin{Lemma}\label{lemma:T1k}
Let $T_{1,k}$ be as defined in \eqref{eq:defT1k}. Under the assumptions of Lemma~\ref{lemma:differenceCosts}, for any $k = 0,\ldots,K_\tau$, we have that
\begin{equation*}
T_{1,k} \geq
\begin{cases}
-\CcostEst \sigma_0^2 \log(n) & \text{if } \tilde{K}_k = 0,\\
\frac{1}{2} \CcostExtraCp \, \sigma_0^2 \log(n)^{1 + \delta} & \text{if } \tilde{K}_k \geq 1.
\end{cases}
\end{equation*}
\end{Lemma}
\begin{proof}
If $\tilde{K}_k = 0$, then the statement follows directly from Lemma~\ref{lemma:costReductionEst}. If $\tilde{K}_k \geq 1$, then applying Lemma~\ref{lemma:costExtraCp} iteratively and finally Lemma~\ref{lemma:costReductionEst} yields
\begin{equation*}
T_{1,k} \geq \CcostExtraCp\, \sigma_0^2 \log(n)^{1 + \delta} - \CcostEst \sigma_0^2 \log(n) \geq \frac{1}{2} \CcostExtraCp\, \sigma_0^2 \log(n)^{1 + \delta},
\end{equation*}
where the last inequality follows from $n$ being large enough.
\end{proof}

\begin{Lemma}\label{lemma:T2k}
Let $T_{2,k}$ be as defined in \eqref{eq:defT2k}. Under the assumptions of Lemma~\ref{lemma:differenceCosts}, for any $k = 1,\ldots,K_\tau$, we have that
\begin{equation*}
T_{2,k} \geq
\begin{cases}
\CextraCP \sigma_0^2 \log(n)^{1 + \delta} & \text{if } \hat{K}_k = 0,\\
-3 \CcostEst \sigma_0^2 \log(n) & \text{if } \hat{K}_k = 1,\\
\frac{1}{2} \CcostExtraCp\, \sigma_0^2 \log(n)^{1 + \delta} & \text{if } \hat{K}_k > 1.
\end{cases}
\end{equation*}
\end{Lemma}
\begin{proof}
If $\hat{K}_k = 0$, then Lemma~\ref{lemma:costMissingCp} gives the desired statement.

From now on, let $\hat{K}_k \geq 1$. We split the costs at $\tau_{k}$. Let $\hat{K}_{k,1}$ be the number of estimated change-points between $\hat{\tau}_{k, 0} + 1$ and $\tau_{k}$ and $\hat{K}_{k,2}$ be the number of estimated change-points between $\tau_k + 1$ and $\hat{\tau}_{k, \hat{K}_k + 1}$. Furthermore, let $\hat{\tau}_{k, 1, l} \coloneq \hat{\tau}_{k, l}$, for $l = 0, \ldots, \hat{K}_{k,1}$, and $\hat{\tau}_{k, 2, l} \coloneq \hat{\tau}_{k, l + \hat{K}_{k,1} }$, for $l = 1, \ldots,\hat{K}_{k,2} + 1$. We also define $\hat{\tau}_{k, 1, \hat{K}_{k,1} + 1} \coloneq \tau_k \eqcolon  \hat{\tau}_{k, 2, 0}$. Since, $\hat{K}_k = \hat{K}_{k,1} + \hat{K}_{k,2}$ and $\hat{K}_k \geq 1$, we have that $\hat{K}_{k,1} \geq 1$ or $\hat{K}_{k,2} \geq 1$ or both. Without loss of generality let $\hat{K}_{k,1} \geq 1$. Otherwise, we can switch the role of the two terms in the following. Thus, we obtain,
\begin{equation*}
T_{2,k} \geq T_{2, k, 1} + T_{2, k, 2},
\end{equation*}
with
\begin{equation*}
\begin{split}
T_{2, k, 1} \coloneq & \sum_{l = 0}^{\hat{K}_{k,1}} C_{\hat{\tau}_{k, 1, l} + 1,\hat{\tau}_{k, 1, l + 1}}(Y - \hat{\theta}_{\hat{\tau}_{k, 1, l} + 1:\hat{\tau}_{k, 1, l + 1}}) - C_{\hat{\tau}_{k, 1, 0} + 1:\hat{\tau}_{k, 1, \hat{K}_{k,1} + 1}}(\varepsilon + R + W)\\
& + \frac{\hat{K}_{k,1} - 1}{2} \lambda_{\operatorname{coll}}\, \hat{\sigma}_0 \log(n)^{1+\delta},
\end{split}
\end{equation*}
and 
\begin{equation*}
\begin{split}
T_{2, k, 2} \coloneq & \sum_{l = 0}^{\hat{K}_{k,2}} C_{\hat{\tau}_{k, 2, l} + 1,\hat{\tau}_{k, 2, l + 1}}(Y - \hat{\theta}_{\hat{\tau}_{k, 2, l} + 1:\hat{\tau}_{k, 2, l + 1}}) - C_{\hat{\tau}_{k, 2, 0} + 1:\hat{\tau}_{k, 2, \hat{K}_{k,2} + 1}}(\varepsilon + R + W)\\
& + \frac{\hat{K}_{k,2}}{2} \lambda_{\operatorname{coll}}\, \hat{\sigma}_0 \log(n)^{1+\delta}.
\end{split}
\end{equation*}

If $\hat{K}_k = 1$, then $\hat{K}_{k,1} = 1$ and $\hat{K}_{k,2} = 0$. Thus, Lemma~\ref{lemma:costReductionEst} yields
\begin{align*}
T_{2, k, 1}
 = & \sum_{l = 0}^{1} \left\{C_{\hat{\tau}_{k, 1, l} + 1,\hat{\tau}_{k, 1, l + 1}}(Y - \hat{\theta}_{\hat{\tau}_{k, 1, l} + 1:\hat{\tau}_{k, 1, l + 1}}) - C_{\hat{\tau}_{k, 1, l} + 1,\hat{\tau}_{k, 1, l + 1}}(\varepsilon + R + W) \right\}\\
\geq & -2 \CcostEst \sigma_0^2 \log(n),
\end{align*}
and 
\begin{align*}
T_{2, k, 2} \geq  -\CcostEst \sigma_0^2 \log(n).
\end{align*}
Thus, 
\begin{equation*}
T_{2,k} \geq -3 \CcostEst \sigma_0^2 \log(n).
\end{equation*}

If $\hat{K}_k > 1$, then $\hat{K}_{k,1} > 1$ and $\hat{K}_{k,2} > 0$ (or both). Applying Lemma~\ref{lemma:costExtraCp} iteratively and then the same steps as before gives us 
\begin{equation*}
T_{2,k} \geq \CcostExtraCp\, \sigma_0^2 \log(n)^{1 + \delta} - 3\CcostEst \sigma_0^2 \log(n) \geq \frac{1}{2} \CcostExtraCp\, \sigma_0^2 \log(n)^{1 + \delta},
\end{equation*}
for large enough $n$.
\end{proof}

We are now ready to show that the estimated anomalies are within the set $\mathcal{B}$ with probability converging to one. This result is of interest by its own but also helps us to show that we select a good trend estimate.
\begin{Proposition}\label{proposition:withinSetB}
Suppose that our data follows Setting~\ref{setting:trend}. Let $J \to \infty$, as $n \to \infty$. Let $\hat{\mathcal{A}}$ be as defined in \eqref{eq:trend_estimator}. Then,
\begin{equation*}
\Pj\big(\hat{\mathcal{A}} \in \mathcal{B}\big) \to 1, \text{ as } n \to \infty.
\end{equation*}
\end{Proposition}
\begin{proof}
Suppose that there exists an arbitrary but fixed polynomial $\tilde{T} \in \mathcal{T}$ \eqref{eq:setOfTrendEstimates} of degree $\rQ$ such that $(\tilde{T} - T)^\top (\tilde{T} - T) \leq \tilde{C}_1 \sigma_0^2 \log(n)$. Assume further that
\begin{equation*}
\max_{1 \leq i \leq j \leq \mn} \varepsilon_{i:j}^\top H_{i:j} \varepsilon_{i:j}
\leq \tilde{C}_2 \sigma_0^2 \log(n)
\end{equation*}
for a constant $\tilde{C}_2 > 0$. Let $n$ be large enough. Then, Lemma~\ref{lemma:differenceCosts} implies that $\hat{\mathcal{A}} \in \mathcal{B}$, since $\mathcal{A} \in \mathcal{B}$ and $\hat{\mathcal{A}}$ is the minimiser of the costs in \eqref{eq:trend_estimator}.

Consequently, it follows from a union bound, Lemmas~\ref{lemma:probGoodSample},~\ref{lemma:trendEstimation},~and~\ref{lemma:noiseLin}, the independence of $\mathcal{U}$ and $\varepsilon$, and $J \to \infty$, as $n\to \infty$, that there exists $\tilde{C}_1, \tilde{C}_2 > 0$ such that
\begin{align*}
& \Pj\big(\hat{\mathcal{A}} \in \mathcal{B}\big)\\
\geq & 1 - \left( 1 - \Pj\left( \max_{1 \leq i \leq j \leq \mn} \varepsilon_{i:j}^\top H_{i:j} \varepsilon_{i:j} \leq \tilde{C}_2 \sigma_0^2 \log(n) \right) \right)\\
& - \left( 1 - \Pj\left( \exists\ \tilde{T} \in \mathcal{T}\,:\, (\tilde{T} - T)^\top (\tilde{T} - T) \leq \tilde{C}_1 \sigma_0^2 \log(n) \right) \right)\\
\geq & \Pj\left( \exists\ U \in \mathcal{U}\, :\,  (\hat{T}_U - T)^\top (\hat{T}_U - T) \leq \tilde{C}_1 \sigma_0^2 \log(n) \right)\\
= & \Pj\Big(\exists\ U \in \mathcal{U}\, :\, U \cap \mathcal{A} = \emptyset \Big)\\
& \cdot  \Pj\left( (\hat{T}_U - T)^\top (\hat{T}_U - T) \leq \tilde{C}_1 \sigma_0^2 \log(n) \mid \exists\ U \in \mathcal{U}\, :\, U \cap \mathcal{A} = \emptyset \right)\\
\geq & 1 - \left(\frac{1}{2}\right)^J
\to 1,
\end{align*}
as $n \to \infty$.
\end{proof}

\begin{proof}[Proof of Proposition~\ref{proposition:goodFinalTrendEstimate}]
Let $\tilde{T} = (\tilde{t}_1,\ldots,\tilde{t}_{\mn})$ be as in Proposition~\ref{proposition:withinSetB}, i.e.~$(\tilde{T} - T)^\top (\tilde{T} - T) \leq \tilde{C}_1 \sigma_0^2 \log(n)$. Assume further that
\begin{equation}\label{eq:goodFinalTrendEstimate:noiseLin}
\max_{1 \leq i \leq j \leq \mn} \varepsilon_{i:j}^\top H_{i:j} \varepsilon_{i:j}
\leq \tilde{C}_2 \sigma_0^2 \log(n)
\end{equation}
for a constant $\tilde{C}_2 > 0$. In the proof of Proposition~\ref{proposition:withinSetB} we have shown that both events hold simultaneously with probability converging to one as $n \to \infty$.

We now show, by contradiction, that the statement in Proposition~\ref{proposition:goodFinalTrendEstimate} holds. If not, then there must exists $\hat{T}$, $\hat{\mathcal{A}}$, and $\hat{\mu}$ such that
\begin{equation}\label{eq:trendEstimateErrorRequirement}
(\hat{T} - T)^\top (\hat{T} - T) = \gamma_n \sigma_0^2 \log(n)
\end{equation}
for a sequence $\gamma_n \to \infty$, as $n \to \infty$, and
\begin{equation}\label{eq:trendEstimateRequirement}
C\big(Y - \hat{T}, \hat{\mathcal{A}}, \hat{\mu}\big) - \max_{Q \in \mathbf{Q}} Q\, \hat{\sigma}_0 \log(n) \leq C\big(Y - \tilde{T}, \mathcal{A}, \mu\big).
\end{equation}
Furthermore, Proposition~\ref{proposition:withinSetB} yields that $\hat{\mathcal{A}} \in \mathcal{B}$. Note that $\mathcal{A} \in \mathcal{B}$ by definition.

Then Lemma~\ref{lemma:mixedtermNoisePolynomial} and the Cauchy-Schwarz inequality for the first inequality, as well as \eqref{eq:goodFinalTrendEstimate:noiseLin}, \eqref{eq:remainderSeasonalyT}, and Lemma~\ref{lemma:noiseLin} for the second inequality imply that
\begin{align*}
& C\big(Y - \tilde{T}, \mathcal{A}, \mu\big) - K \lambda_{\operatorname{coll}}\, \hat{\sigma}_0^2 \log(n)^{1 + \delta}\\
= & \sum_{i = 1}^{\mn} (y_i - a_i - \tilde{t}_i)^2
=  \sum_{i = 1}^{\mn} (\varepsilon_i + w_i + t_i - \tilde{t}_i)^2\\
= & \sum_{i = 1}^{\mn} (\varepsilon_i + w_i)^2 + 2\sum_{i = 1}^{\mn} \varepsilon_i (t_i - \tilde{t}_i) + 2\sum_{i = 1}^{\mn} w_i (t_i - \tilde{t}_i) + \sum_{i = 1}^{\mn} (t_i - \tilde{t}_i)^2\\
\leq & \sum_{i = 1}^{\mn} (\varepsilon_i + w_i)^2 + 2 \left(\sqrt{\varepsilon_{1:\mn}^\top H_{1:\mn} \varepsilon_{1:\mn}} + \sqrt{\sum_{i = 1}^{\mn} w_i^2}\right) \sqrt{\sum_{i = 1}^{\mn} (t_i - \tilde{t}_i)^2} + \sum_{i = 1}^{\mn} (t_i - \tilde{t}_i)^2\\
%\leq & \sum_{i = 1}^{\mn} \varepsilon_i^2 + 2 \sqrt{\CnoiseLin \CerrorTrend}  \sigma_0^2 \log(n) + \CerrorTrend \sigma_0^2 \log(n)\\
= & \sum_{i = 1}^{\mn} (\varepsilon_i + w_i)^2 + \tilde{C}_3 \sigma_0^2 \log(n),
\end{align*}
where $\tilde{C}_3 \coloneq 2 \sqrt{\CnoiseLin \CerrorTrend} + 2 \sqrt{\CremainderSeasonal \CerrorTrend} + \CerrorTrend$.

Let $\{\tilde{\tau}_0, \tilde{\tau}_1, \ldots, \tilde{\tau}_{\tilde{K}}, \tilde{\tau}_{\tilde{K} + 1} \} = \{ e_0, b_1,\hat{b}_1, e_1, \hat{e}_1, \ldots, b_K,\hat{b}_K, e_K, \hat{e}_K, b_{K + 1} \}$ with $\tilde{\tau}_0 \leq  \tilde{\tau}_1 \leq \cdots \leq \tilde{\tau}_{\tilde{K}} \leq \tilde{\tau}_{\tilde{K} + 1}$ and $\tilde{K} = 4 K$. Then, the first inequality is required by \eqref{eq:trendEstimateRequirement} and the remaining steps follow the same arguments as before as well as \eqref{eq:boundQsigmalogn} 
\begin{align*}
& \sum_{i = 1}^{\mn} (\varepsilon_i + w_i)^2 + \tilde{C}_3 \sigma_0^2 \log(n)\\
\geq & C\big(Y - \hat{T}, \hat{\mathcal{A}}, \hat{\mu}\big) - K \lambda_{\operatorname{coll}}\, \hat{\sigma}_0^2 \log(n)^{1 + \delta} - \max_{Q \in \mathbf{Q}} Q\, \hat{\sigma}_0 \log(n)\\
= & \sum_{i = 1}^{\mn} (y_i - \hat{a}_i - \hat{t}_i)^2 - \max_{Q \in \mathbf{Q}} Q\, \hat{\sigma}_0 \log(n)\\
= & \sum_{i = 1}^{\mn} (\varepsilon_i + w_i + a_i - \hat{a}_i + t_i - \hat{t}_i)^2 - \max_{Q \in \mathbf{Q}} Q\, \hat{\sigma}_0 \log(n)\\
= & \sum_{i = 1}^{\mn} (\varepsilon_i + w_i)^2 + 2\sum_{i = 1}^{\mn}\varepsilon_i (a_i - \hat{a}_i + t_i - \hat{t}_i)\\
& + 2 \sum_{i = 1}^{\mn} w_i (a_i - \hat{a}_i + t_i - \hat{t}_i) + \sum_{i = 1}^{\mn} (a_i - \hat{a}_i + t_i - \hat{t}_i)^2 - \max_{Q \in \mathbf{Q}} Q\, \hat{\sigma}_0 \log(n)\\
\geq & \sum_{i = 1}^{\mn} (\varepsilon_i + w_i)^2 
 - 2 \sqrt{\sum_{k = 0}^{\tilde{K}} \varepsilon_{\tilde{\tau}_k + 1:\tilde{\tau}_{k + 1}}^\top H_{\tilde{\tau}_k + 1:\tilde{\tau}_{k + 1}} \varepsilon_{\tilde{\tau}_k + 1:\tilde{\tau}_{k + 1}} } \sqrt{\sum_{i = 1}^{\mn} (a_i - \hat{a}_i + t_i - \hat{t}_i)^2}\\
& - 2 \sqrt{\sum_{i = 1}^{\mn} w_i^2} \sqrt{ \sum_{i = 1}^{\mn} (a_i - \hat{a}_i + t_i - \hat{t}_i)^2} + \sum_{i = 1}^{\mn} (a_i - \hat{a}_i + t_i - \hat{t}_i)^2 - \max_{Q \in \mathbf{Q}} Q\,\CvarianceEstimate\, \sigma_0 \log(n).
\end{align*}
This can only hold if there exists a constant $\tilde{C}_4 > 0$ such that
\begin{equation}\label{eq:boundhatThatA}
\sum_{i = 1}^{\mn} (t_i - \hat{t}_i - a_i - \hat{a}_i)^2 \leq \tilde{C}_4 \sigma_0^2 \log(n).
\end{equation}

This, combined with the Cauchy-Schwarz inequality, and \eqref{eq:trendEstimateErrorRequirement} gives us
\begin{align*}
\tilde{C}_4 \sigma_0^2 \log(n)
\geq & \sum_{i = 1}^{\mn} (t_i - \hat{t}_i + a_i - \hat{a}_i)^2\\
= &  \sum_{i = 1}^{\mn} (t_i - \hat{t}_i)^2 +  2 \sum_{i = 1}^{\mn} (t_i - \hat{t}_i) (a_i - \hat{a}_i)^2 +  \sum_{i = 1}^{\mn} (a_i - \hat{a}_i)^2\\
\geq & \sum_{i = 1}^{\mn} (t_i - \hat{t}_i)^2 -  2 \sqrt{\sum_{i = 1}^{\mn} (t_i - \hat{t}_i)^2} \sqrt{\sum_{i = 1}^{\mn} (a_i - \hat{a}_i)^2} +  \sum_{i = 1}^{\mn} (a_i - \hat{a}_i)^2\\
\geq & \gamma_n \sigma_0^2 \log(n) - 2 \sqrt{\gamma_n \sigma_0^2 \log(n)} \sqrt{\sum_{i = 1}^{\mn} (a_i - \hat{a}_i)^2} +  \sum_{i = 1}^{\mn} (a_i - \hat{a}_i)^2.
\end{align*}
This can only hold if there exists a constant $\tilde{C}_5 > 0$ such that
\begin{equation*}
\sum_{i = 1}^{\mn} (a_i - \hat{a}_i)^2 > \tilde{C}_5 \gamma_n \sigma_0^2 \log(n).
\end{equation*}

Since $\tilde{K}$ is finite, there must exist an $l \in \{0,\ldots,\tilde{K}\}$ and a constant $\tilde{C}_6 > 0$ such that
\begin{equation}\label{eq:boundhatAforAk}
\sum_{i = \tilde{\tau}_l + 1}^{\tilde{\tau}_{l + 1}} (a_i - \hat{a}_i)^2 > \tilde{C}_6 \gamma_n \sigma_0^2 \log(n).
\end{equation}

We will show the statement by contradiction. We distinguish cases depending on which interval $[\tilde{\tau}_l + 1, \tilde{\tau}_{l + 1}]$ is.

\begin{enumerate}
\item Suppose $[\tilde{\tau}_l + 1, \tilde{\tau}_{l + 1}]$ is non-anomalous and to be estimated non-anomalous, i.e.~there exists a $k \in \{1,\ldots,K + 1\}$ such that $\tilde{\tau}_l = \max \{e_{k - 1}, \hat{e}_{k - 1}\}$ and $\tilde{\tau}_{l + 1} = \min \{b_{k}, \hat{b}_{k}\}$. Then,  $\hat{a}_i = a_i = 0$ for all $\tilde{\tau}_l < i \leq \tilde{\tau}_{l + 1}$. Thus,
\begin{equation*}
\sum_{i = \tilde{\tau}_l + 1}^{\tilde{\tau}_{l + 1}} (a_i - \hat{a}_i)^2 = 0.   
\end{equation*}
This is a contradiction to \eqref{eq:boundhatAforAk}.
\item Suppose $[\tilde{\tau}_l + 1, \tilde{\tau}_{l + 1}]$ is either anomalous or to be estimated anomalous, but not both, i.e.~there exists a $k \in \{1,\ldots,K\}$ such that either $\tilde{\tau}_l = \min \{b_{k}, \hat{b}_{k}\}$ and $\tilde{\tau}_{l + 1} > \min \{b_{k}, \hat{b}_{k}\}$ or $\tilde{\tau}_l < \max \{e_{k}, \hat{e}_{k}\}$ and $\tilde{\tau}_{l + 1} = \max \{e_{k}, \hat{e}_{k}\}$. We focus on the first case. The second case follows in the same way but with $\kappa_0$ and $\kappa_1$ reversed.

Let $\kappa_1 = \tilde{\tau}_{l + 1} - \tilde{\tau}_l$ and $\kappa_0 = \tilde{\tau}_{l} - \tilde{\tau}_{l - 1}$. Because of the definition of $\kappa_k$ and Assumption~\ref{assumption:segmentLengthTrend} we have that $\kappa_0 > \max\{ \kappa_1, 4^{Q + 1}\}$ for $n$ large enough. We have that $a _ i = \hat{a}_i = 0$ for all  $\tilde{\tau}_{l - 1} < i \leq \tilde{\tau}_{l}$ since this interval must be non-anomalous and estimated to be non-anomalous by the order of the $\tau_l$'s. Secondly, there exists a $\hat{\mu} \in \mathbb{R}$ such that $a_i - \hat{a}_i = \hat{\mu}$ for $\tilde{\tau}_l + 1 \leq i \leq \tilde{\tau}_{l + 1}$. Consequently, since $T - \hat{T}$ is a polynomial of degree $Q$, Lemma~\ref{lemma:meanMissCp}, and \eqref{eq:boundhatAforAk} imply 
\begin{align*}
& \sum_{i = \tilde{\tau}_{l} - \kappa_0 + 1}^{\tilde{\tau}_{l}} (t_i - \hat{t}_i + a_i - \hat{a}_i)^2 + \sum_{i = \tilde{\tau}_l + 1}^{\tilde{\tau}_{l} + \kappa_1} (t_i - \hat{t}_i + a_i - \hat{a}_i)^2\\
= & \sum_{i = \tilde{\tau}_l - \kappa_0 + 1}^{\tilde{\tau}_{l}} (t_i - \hat{t}_i)^2 + \sum_{i = \tilde{\tau}_l + 1}^{\tilde{\tau}_{l} + \kappa_0} (t_i - \hat{t}_i - \hat{\mu})^2\\
\geq & \Cmiss \kappa_1 \hat{\mu}^2
=  \Cmiss \sum_{i = \tilde{\tau}_l + 1}^{\tilde{\tau}_{l + 1}} (a_i - \hat{a}_i)^2
 > \Cmiss \tilde{C}_6 \gamma_n \sigma_0^2 \log(n).
\end{align*}
This is a contradiction to \eqref{eq:boundhatThatA}.
\item Suppose $[\tilde{\tau}_l + 1, \tilde{\tau}_{l + 1}]$ is anomalous and to be estimated anomalous, i.e.~there exists a $k \in \{1,\ldots,K\}$ such that $\tilde{\tau}_l = \max \{b_{k}, \hat{b}_{k}\}$ and $\tilde{\tau}_{l + 1} = \min \{e_{k}, \hat{e}_{k}\}$. 

Let $\kappa_0 = \tilde{\tau}_{l + 1} - \tilde{\tau}_l$ and $\kappa_1 =  \tilde{\tau}_{l + 2} - \tilde{\tau}_{l + 1}$. Because of the definition of $\kappa_k$ and Assumption~\ref{assumption:segmentLengthTrend} we have that $\kappa_0 > \max\{ \kappa_1, 4^{Q + 1}\}$ for $n$ large enough. If $\kappa_1 = 0$, then consider the next interval and we may have to switch the roles of $\kappa_0$ and $\kappa_1$. If there exists a sequence $\tilde{\gamma}_n \to \infty$ such that $\sum_{i = \tilde{\tau}_{l + 1} + 1}^{\tilde{\tau}_{l + 2}} (t_i - \hat{t}_i)^2 > \tilde{\gamma}_n \sigma_0^2 \log(n)$, then it follows from the same arguments that lead to \eqref{eq:boundhatAforAk} that there exists a constant $\tilde{C}_7 > 0$ such that $\sum_{i = \tilde{\tau}_{l + 1} + 1}^{\tilde{\tau}_{l + 2}} (a_i - \hat{a}_i)^2 > \tilde{C}_7 \tilde{\gamma}_n \sigma_0^2 \log(n)$. Consequently, considering $l - 1$ and case 2.~leads to a contradiction. Otherwise, there exists a constant $\tilde{C}_8$ such that
\begin{equation*}
\sum_{i = \tilde{\tau}_{l + 1} + 1}^{\tilde{\tau}_{l + 2}} (t_i - \hat{t}_i)^2 \leq \tilde{C}_8 \sigma_0^2 \log(n).
\end{equation*}
Then, from the same arguments as before, there exists constants $\tilde{C}_9,\tilde{C}_{10} > 0$ such that
\begin{align*}
& \sum_{i = \tilde{\tau}_l + 1}^{\tilde{\tau}_{l + 1}} (t_i - \hat{t}_i + a_i - \hat{a}_i)^2 + \sum_{i = \tilde{\tau}_{l + 1} + 1}^{\tilde{\tau}_{l + 2} + \kappa} (t_i - \hat{t}_i + a_i - \hat{a}_i)^2\\
= & \sum_{i = \tilde{\tau}_l + 1}^{\tilde{\tau}_{l + 1}} (t_i - \hat{t}_i + a_i - \hat{a}_i)^2 + \sum_{i = \tilde{\tau}_{l + 1} + 1}^{\tilde{\tau}_{l + 2}} (t_i - \hat{t}_i)^2\\
& - \sum_{i = \tilde{\tau}_{l + 1} + 1}^{\tilde{\tau}_{l + 2}} (t_i - \hat{t}_i)^2 + \sum_{i = \tilde{\tau}_{l + 1} + 1}^{\tilde{\tau}_{l + 2} + \kappa} (t_i - \hat{t}_i + a_i - \hat{a}_i)^2\\
\geq & \tilde{C}_9 \gamma_n \sigma_0^2 \log(n)  - \tilde{C}_8 \sigma_0^2 \log(n) \geq \tilde{C}_{10} \gamma_n \sigma_0^2 \log(n).
\end{align*}
This is a contradiction to \eqref{eq:boundhatThatA}.
\end{enumerate}
In summary all cases lead to a contradiction. This completes the proof.
\end{proof}

\subsection{Seasonality}\label{sec:proofSeasonality}

We show that we can estimate the seasonality well if the following assumption and setting is assumed. Recall that we use $\tilde{\vartheta}_p \coloneq \operatorname{median}(Y_{V_p})$.

\begin{Assumption}\label{assumption:anomalyLengthSeason}
The total length of collective anomalies is bounded, i.e.
\begin{equation*}
\sum_{k = 1}^{K} e_k - b_k < \frac{1}{8}.
\end{equation*}
\end{Assumption}

\begin{Setting}\label{setting:seasonality}
We assume that
\begin{equation*}
Y = A + S + B + R + \varepsilon,
\end{equation*}
where $Y, A, S, \varepsilon$ are as defined in Section~\ref{sec:model}, $B = (b,\ldots,b) \in \mathbb{R}^n$ is a constant baseline with $b = 0$ or $b = \beta_0$, and $R = (r_1,\ldots,r_n)$ being a deterministic vector such that there exists a constant $\CremainderTrend > 0$ such that
\begin{equation*}
\max_{l = 1,\ldots,n} \vert r_i \vert \leq \CremainderTrend \sqrt{\frac{\log(n)}{n}} \sigma_0.
\end{equation*}
We further assume that
\begin{equation}\label{eq:settingSeasonalityAnomaly}
\hat{\mathcal{A}} = \emptyset \text{ or } \hat{\mathcal{A}} \in \mathcal{B},
\end{equation}
where $\hat{\mathcal{A}}$ is deterministic. Finally, suppose that Assumptions~\ref{assumption:segmentLengthTrend}~and~\ref{assumption:anomalyLengthSeason} holds.
\end{Setting}

Under Setting~\ref{setting:seasonality} we show the following proposition.
\begin{Proposition}\label{proposition:estSeasonality}
Suppose that $Y$ follows Setting~\ref{setting:seasonality}. Let $\tilde{\vartheta}_p \coloneq \operatorname{median}(Y_{V_p})$. Then, there exists a constant $\CfinalSeasonalEstimate > 0$ such that
\begin{equation*}
\Pj\left(\sum_{l = 1}^{n} (\hat{s}_l - s_l)^2 \leq \CfinalSeasonalEstimate \log(n) \sigma_0 \right) \to 1,
\end{equation*}
as $n \to \infty$.
\end{Proposition}
We require a few definitions and lemmas before we can prove Proposition~\ref{proposition:estSeasonality}. Let
\begin{equation*}
\tilde{V}_p \coloneq \big\{ i \in \{1,\ldots,n\} : (i \bmod P) = p\big\} \text{ and } V_p \coloneq \tilde{V}_p\setminus \hat{\mathcal{A}},
\end{equation*}
where the latter matches \eqref{eq:setSeasonalEstimation}. By construction, $\tilde{V}_0 = \{P,2P, \ldots, \tilde{n}_{0} P \}$, with $\tilde{n}_{0} \coloneq \lfloor \frac{n}{P} \rfloor$, and $\tilde{V}_p = \{ p, p + P,\ldots, p + (\tilde{n}_{p} - 1) P \}$, with $\tilde{n}_{p} = \tilde{n}_{0}$ or $\tilde{n}_{p} = \tilde{n}_{0} + 1$. 

\begin{Lemma}\label{lemma:medianInterval}
Under the setting of Proposition~\ref{proposition:estSeasonality}, with probability converging to one, $\tilde{\vartheta}_p$, i.e.~the median of the observations in $V_p$, is in the interval
\begin{equation*}
\big[\gamma_{L,p}, \gamma_{U,p}\big] \coloneq \left[ \vartheta_p + b + \tilde{q} - \gamma - \CremainderTrend \sqrt{\frac{\log(n)}{n}} \sigma_0,\ \vartheta_p + b +\tilde{q} + \gamma + \CremainderTrend \sqrt{\frac{\log(n)}{n}} \sigma_0\right],
\end{equation*}
with $\gamma \coloneq \sqrt{\frac{\log(n)}{n}} \sigma_0$ and $\tilde{q}$ being the median of the observations
\begin{equation}\label{eq:definitionTildeq}
\{ a_j + \varepsilon_j\,:\, j \in V_0\}.
\end{equation}
\end{Lemma}
We will prove this lemma after we have shown a few more technical lemmas. The following lemma bounds the cardinality of $V_0$ and how many observations are anomalous.
\begin{Lemma}\label{lemma:sizeOfV0}
Under the setting of Proposition~\ref{proposition:estSeasonality}, if $n$ is large enough, then there exists a constant $\CsizeOfVzero > 0$ such that the set $V_0$ contains $\frac{5}{6}\frac{n}{P} - \CsizeOfVzero \leq N_0 \leq \frac{n}{P}$ many observations. Moreover, at most $\frac{1}{6}\frac{n}{P} + \CsizeOfVzero$ many observations are anomalous and at least $\frac{5}{6}\frac{n}{P} - \CsizeOfVzero$ many observations are non-anomalous.
\end{Lemma}
\begin{proof}
Because of \eqref{eq:settingSeasonalityAnomaly}, $V_0$ contains at least $\tilde{n}_{0}$ many observations minus the points that are part of an estimated anomaly. A season containing an estimated collective anomalies is either fully part of an estimated collective anomaly or at the boundary of at least one estimated collective anomaly. The definition of $\mathcal{B}$ yield that at most 
\begin{equation*}
\frac{\sum_{k = 1}^K e_k - b_k + \sum_{k = 1}^{K} 2 \Big\lceil \Clambda\frac{\sigma_0^2}{\Delta_k^2}\log(n)^{1 + \delta} \Big\rceil}{P}
\end{equation*}
seasons are fully part of an estimated collective anomaly. Furthermore, at most $2 K$ many seasons are at the boundary of an estimated anomaly. So, in total $V_p$ contains at least
\begin{equation}\label{eq:boundCardinalityVp}
\tilde{n}_{0} - \frac{\sum_{k = 1}^K e_k - b_k + \sum_{k = 1}^{K} 2 \Big\lceil \Clambda\frac{\sigma_0^2}{\Delta_k^2}\log(n)^{1 + \delta} \Big\rceil}{P} - 2 K
\end{equation}
many observations. Let $\CsizeOfVzero = 2K + 1$. Let $n$ be large enough such that
\begin{equation*}
\sum_{k = 1}^{K} 2 \big\lceil \Clambda\frac{\sigma_0^2}{\Delta_k^2}\log(n)^{1 + \delta} \big\rceil \leq \frac{1}{24} n.   
\end{equation*}
This holds for large $n$ because of Assumption~\ref{assumption:segmentLengthTrend}. Then, it follows from Assumption~\ref{assumption:anomalyLengthSeason} that
\begin{align*}
& \tilde{n}_{0} - \frac{\sum_{k = 1}^K e_k - b_k + \sum_{k = 1}^{K} 2 \big\lceil \Clambda\frac{\sigma_0^2}{\Delta_k^2}\log(n)^{1 + \delta} \big\rceil}{P} - 2 K\\
\geq & \frac{n}{P} - 1 - \frac{1}{8} \frac{n}{P} - \frac{1}{24} \frac{n}{P} - 2 K\\
\geq & \frac{5}{6}\frac{n}{P} - \CsizeOfVzero.
\end{align*}

This shows the first statement. For the second statement note that a season containing anomalous points is either fully part of a collective anomaly or at the boundary of at least one collective anomaly. At most $\frac{\sum_{k = 1}^K e_k - b_k}{P}$ seasons are fully part of a collective anomaly and at most $2 K$ many seasons are at the boundary of an anomaly. So, in total at most
\begin{equation}\label{eq:boundCardinalityAnomalous}
\frac{\sum_{k = 1}^K e_k - b_k}{P} + 2 K
\end{equation}
many observations in $V_0$ are anomalous. It follows from Assumption~\ref{assumption:anomalyLengthSeason} that
\begin{align*}
\frac{\sum_{k = 1}^K e_k - b_k}{P} + 2 K
\leq \frac{1}{8} \frac{n}{P} + 2 K
\leq \frac{1}{6}\frac{n}{P} + \CsizeOfVzero.
\end{align*}
The number of non-anomalous points follows conversely.
\end{proof}

We have two technical lemmas on how to bound the density and cumulative distribution function of a standard Gaussian random variable.
\begin{Lemma}\label{lemma:boundphi}
Let $\phi$ be the density of a standard Gaussian random variable. If $x \leq 1$, then
\begin{equation*}
\phi(x) \geq \frac{1}{\sqrt{2\pi}}\left(1 - \frac{x^2}{2} \right).
\end{equation*}
\end{Lemma}
\begin{proof}
We obtain the Taylor series
\begin{align*}
\phi(x) = \frac{1}{\sqrt{2\pi}}\sum_{k = 0}^{\infty} (-1)^k \frac{1}{k!} 2^{-k} x^{2k}.
\end{align*}
The statement follows from the fact that the terms have alternating signs and decay in absolute values, since $x \leq 1$.
\end{proof}

\begin{Lemma}\label{lemma:boundPhi}
Let $\Phi$ be the cumulative distribution function of a standard Gaussian random variable. If $x \leq 1$, then
\begin{equation*}
\Phi(x) \geq\frac{1}{2} + \frac{5}{6 \sqrt{2 \pi}} x.
\end{equation*}
\end{Lemma}
\begin{proof}
Taylor's theorem gives us
\begin{align*}
\Phi(x) = \frac{1}{2} + \frac{x}{\sqrt{2\pi}} - \frac{\tilde{x}^3}{6\sqrt{2\pi}},
\end{align*}
for an $\tilde{x} \in [0, x]$. The statement follows from $\frac{\tilde{x}^3}{6\sqrt{2\pi}} \leq \frac{x^3}{6\sqrt{2\pi}} \leq  \frac{1}{6 \sqrt{2 \pi}} x$, since $\tilde{x} \leq x \leq 1$.
\end{proof}

The following lemma bounds $\tilde{q}$.
\begin{Lemma}\label{lemma:boundTildeq}
Under the setting of Proposition~\ref{proposition:estSeasonality}, if $n$ is large enough, then there exists a constant $0 < \CtildeQ < 0.752$ and a constant $0 < \CprobBoundTildeq < 1$ such that
\begin{equation*}
\Pj(\tilde{q} > \CtildeQ\, \sigma_0) \leq \CprobBoundTildeq \frac{1}{n}.
\end{equation*}
\end{Lemma}
\begin{proof}
Let $n_0$ be the size of $V_0$. As the median, $\tilde{q} > \CtildeQ\, \sigma_0$ requires 
\begin{equation*}
\sum_{j \in V_0} \EINS_{\{a_{j} + \varepsilon_{j} \geq \CtildeQ\, \sigma_0 \}} > \frac{1}{2} n_0.
\end{equation*}
We obtain from Lemma~\ref{lemma:sizeOfV0} and $n_0 = \lceil \frac{n}{P} \rceil \leq \frac{n}{P} - 1$ that
\begin{align*}
& \Pj\left( \sum_{j \in V_0} \EINS_{\{a_{j} + \varepsilon_{j} \geq \CtildeQ\, \sigma_0 \}} > \frac{1}{2} n_0 \right)\\
\leq & \Pj\left( \sum_{j = 1}^{\frac{5}{6}\frac{n}{P} - \CsizeOfVzero} \EINS_{\{\varepsilon_{j} \geq \CtildeQ\, \sigma_0 \}} > \frac{1}{2} \frac{n}{P} - 1 - \frac{1}{6} \frac{n}{P} - \CsizeOfVzero \right).
\end{align*}

The random variable on the left hand side in the probability is Binomial distributed with size $\frac{5}{6}\frac{n}{P} - \CsizeOfVzero$ and success probability $\Pj(\varepsilon_1 \geq \CtildeQ\, \sigma_0)$. We use the following bound from \citep[(3.5)]{feller1958introduction},
\begin{equation*}
\Pj(S_n \geq r) \leq \frac{r (1- p)}{(r - np)^2},
\end{equation*}
if $r \geq np$, where $S_n$ is Binomial distributed with size $n$ and success probability $p$. It follows from Lemma~\ref{lemma:boundPhi} that
\begin{align*}
\Pj(\varepsilon_1 \geq \CtildeQ\, \sigma_0)
= 1 - \Phi\left( \CtildeQ \right)
\leq \frac{1}{2} - \frac{5}{6}\frac{\CtildeQ}{\sqrt{2\pi}}
\leq \frac{1}{4}.
\end{align*}
Thus,
\begin{align*}
\Pj\left( \tilde{q} > \CtildeQ\, \sigma_0 \right)
\leq \frac{\frac{1}{2} \frac{n}{P}}{\left(\frac{1}{2} \frac{n}{P} - 1 - \frac{1}{6} \frac{n}{P} - \CsizeOfVzero  - \frac{1}{4} \frac{n}{P} \right)^2}
\leq \CprobBoundTildeq \frac{1}{n}.
\end{align*}
\end{proof}

The following lemma bounds the probability that the noise is between $\tilde{q}$ and $\tilde{q} + \gamma$.
\begin{Lemma}\label{lemma:boundSuccessProbability}
Under the setting of Proposition~\ref{proposition:estSeasonality}, if $n$ is large enough, then there exists a constant $\CsuccessProbability > 0$ such that
\begin{equation*}
\Pj(\tilde{q} \leq \varepsilon_1 \leq \tilde{q} + \gamma) \geq \CsuccessProbability \frac{\gamma}{\sigma_0}.
\end{equation*}
\end{Lemma}
\begin{proof}
Let $q \coloneq 0.752\, \sigma_0$. Then, if $n$ is large enough, there exists a constant $0 < \tilde{C}_1 < 1$ such that
\begin{equation}\label{eq:boundqplusgamma}
\frac{q + \gamma}{\sigma_0} \leq \tilde{C}_1.
\end{equation}
In the first inequality we lower bound the integral by minimal value times the length of the interval of integration. We use Lemma~\ref{lemma:boundphi} for the second and \eqref{eq:boundqplusgamma} for the third inequality. We obtain
\begin{align*}
& \Pj(q \leq \varepsilon_1 \leq q + \gamma)\\
= & \int_{\frac{q}{\sigma_0}}^{\frac{q + \gamma}{\sigma_0}} \phi(x) dx\\
\geq & \phi\left( \frac{q + \gamma}{\sigma_0} \right) \frac{\gamma}{\sigma_0} \\
\geq & \frac{1}{\sqrt{2\pi}} \left[ 1 - \frac{1}{2} \left(\frac{q + \gamma}{\sigma_0}\right)^2 \right] \frac{\gamma}{\sigma_0} \\
\geq & \left[ \frac{1}{\sqrt{2\pi}} \left(1 - \frac{1}{2} \tilde{C}_1^2\right)  \right] \frac{\gamma}{\sigma_0}.
\end{align*}
It follows from Lemma~\ref{lemma:boundTildeq} and a union bound that 
\begin{align*}
\Pj(\tilde{q} \leq \varepsilon_1 \leq \tilde{q} + \gamma)
\geq & \Pj(q \leq \varepsilon_1 \leq q + \gamma) - \Pj( \tilde{q} > q )\\
\geq & \left[ \frac{1}{\sqrt{2\pi}} \left(1 - \frac{1}{2} \tilde{C}_1^2\right)  \right] \frac{\gamma}{\sigma_0} - \CprobBoundTildeq \frac{\gamma}{\sigma_0}
\geq \CsuccessProbability \frac{\gamma}{\sigma},
\end{align*}
with $\CsuccessProbability \coloneq \frac{1}{\sqrt{2\pi}} \left(1 - \frac{1}{2} \tilde{C}_1^2\right) - \CprobBoundTildeq > 0$.
\end{proof}

With these lemmas we are finally able to prove Lemma~\ref{lemma:medianInterval}.
\begin{proof}[Proof of Lemma~\ref{lemma:medianInterval}]
For any $p \in \{0,\ldots,P - 1\}$, let $n_p$ be the number of observations in $V_p$. As $\tilde{\vartheta}_p$ is the median of the observations in $V_p$,
\begin{align*}
& \Pj\left( \gamma_{L,p} \leq \tilde{\vartheta}_p \leq  \gamma_{U,p}\right)\\
\geq & 1 - \Pj\left( \sum_{j \in V_p} \EINS_{Y_j > \gamma_{U,p}} \geq \frac{1}{2} n_p \right) - \Pj\left( \sum_{j \in V_p} \EINS_{Y_j < \gamma_{L,p}} \geq \frac{1}{2} n_p \right).
\end{align*}

We focus on the first probability, but bounding the second probability follows similar arguments. The following relies on the fact that $V_0$ and $V_p$ are similar for most observations. For $p \neq 0$, we have that
\begin{align*}
& \sum_{j \in V_p} \EINS_{\{Y_j > \gamma_{U,p}\}}\\
= & \sum_{j = 0}^{\tilde{n}_p - 1} \EINS_{\{Y_{p + j P} > \gamma_{U,p},\ p + j P \in V_p\}}\\
\leq & \sum_{j = 0}^{\tilde{n}_p - 1} \EINS_{\{a_{p + j P} + \varepsilon_{p + j P} > \tilde{q} + \gamma,\ p + j P \in V_p\}}\\
\leq & \sum_{j = 0}^{\tilde{n}_0 - 1} \EINS_{\{a_{(j + 1) P} + \varepsilon_{p + j P} > \tilde{q} + \gamma,\ (j + 1) P \in V_0\}} + \sum_{j = 0}^{\tilde{n}_p - 1} \EINS_{a_{p + j P} \neq a_{(j + 1) P}} + \sum_{j = 0}^{\tilde{n}_0 - 1} \EINS_{p + j P \in V_p,\ (j + 1) P \not\in V_0}\\
\leq & \sum_{j = 0}^{\tilde{n}_0 - 1} \EINS_{\{a_{(j + 1) P} + \varepsilon_{p + j P} > \tilde{q} + \gamma,\ (j + 1) P \in V_0\}} + 4K + 1.
\end{align*}
The last inequality holds, since there are at most $4K + 1$ many $j's$ for which either $a_{p + j P} \neq a_{(j + 1) P}$ or $(j + 1) P \in V_0$, but $p + j P \in V_p$, since this can only happen at the beginning or end of a collective anomaly or an estimated collective anomaly or for the last observation in $V_p$. The last line has the same distribution as
\begin{equation*}
\sum_{j \in V_0} \EINS_{\{a_{j} + \varepsilon_{j} > \tilde{q} + \gamma\}} + 4K + 1.
\end{equation*}
For similar reasons we have that $n_p \geq n_0 - 2K$. Hence, there exists a constant $\tilde{C}_1 > 0$ such that
\begin{equation*}
\Pj\left( \sum_{j \in V_p} \EINS_{Y_j > \gamma_{U,p}} \geq \frac{1}{2} n_p \right)
\leq \Pj\left( \sum_{j \in V_0} \EINS_{\{a_{j} + \varepsilon_{j} > \tilde{q} + \gamma\}} \geq \frac{1}{2} n_0 - \tilde{C}_1 \right).
\end{equation*} 
If $p = 0$, we can obtain the same inequality with $\tilde{C}_1 = 0$. By the definition of $\tilde{q}$ \eqref{eq:definitionTildeq},
\begin{equation*}
\sum_{j \in V_0} \EINS_{\{a_{j} + \varepsilon_{j} \geq \tilde{q} \}} = \frac{1}{2} n_0.
\end{equation*}
Thus,
\begin{align*}
\Pj\left( \sum_{j \in V_0} \EINS_{\{a_{j} + \varepsilon_{j} > \tilde{q} + \gamma\}} \geq \frac{1}{2} n_0 - \tilde{C}_1 \right)
= & \Pj\left( \sum_{j \in V_0} \EINS_{\{\tilde{q} \leq a_{j} + \varepsilon_{j} \leq \tilde{q} + \gamma\}} \leq \tilde{C}_1 \right)\\
\leq & \Pj\left( \sum_{j = 1}^{\frac{5}{6}\frac{n}{P} - \CsizeOfVzero} \EINS_{\{\tilde{q} \leq \varepsilon_{j} \leq \tilde{q} + \gamma\}} \leq \tilde{C}_1 \right),
\end{align*}
where the final inequality follows from Lemma~\ref{lemma:sizeOfV0}.

The random variable on the left hand side in the probability is Binomial distributed with size $\frac{5}{6}\frac{n}{P} - \CsizeOfVzero$ and success probability $\Pj(\tilde{q} \leq \varepsilon_1 \leq \tilde{q} + \gamma)$. We use the following bound from \citep[(3.6)]{feller1958introduction},
\begin{equation*}
\Pj(S_n \leq r) \leq \frac{(n - r)p}{(np - r)^2},
\end{equation*}
if $r \leq np$, where $S_n$ is Binomial distributed with size $n$ and success probability $p$. From this and Lemma~\ref{lemma:boundSuccessProbability} we obtain
\begin{align*}
\Pj\left( \sum_{j = 1}^{\frac{5}{6}\frac{n}{P} - \CsizeOfVzero} \EINS_{\{\tilde{q} \leq \varepsilon_{j} \leq \tilde{q} + \gamma\}} \leq \tilde{C}_1 \right)
\leq \frac{\frac{n}{P}}{\big((\frac{5}{6}\frac{n}{P} - \CsizeOfVzero)\CsuccessProbability \frac{\gamma}{\sigma_0} - \tilde{C}_1\big)^2} \to 0,
\end{align*}
as $n\to \infty$. This completes the proof as showing that the second probability converges to $0$  follows similar arguments.
\end{proof}

We now in the position to show Proposition~\ref{proposition:estSeasonality}.
\begin{proof}[Proof of Proposition~\ref{proposition:estSeasonality}]
It follows from Lemma~\ref{lemma:medianInterval}, $\hat{s}_l = \tilde{s}_l - \frac{1}{P} \sum_{p = 0}^{P - 1} \tilde{s}_p$ that
\begin{equation*}
\Pj\left(\sum_{l = 1}^{n} (\hat{s}_l - s_l)^2 \leq \CfinalSeasonalEstimate \log(n) \sigma_0\right) \to 1,
\end{equation*}
as $n \to \infty$, with $\CfinalSeasonalEstimate \coloneq 1 + \CremainderTrend$.
\end{proof}

\subsection{Anomaly estimation}\label{sec:proofAnomaly}

We will show that under the following setting we can estimate the anomaly component well. In Proposition~\ref{proposition:goodFinalAnomalyDetection} we show that $\hat{\mathcal{A}} \in \mathcal{B}$ with probability converging to one. In Proposition~\ref{proposition:goodFinalAnomalyEstimate} we will bound the L2 error of the anomaly estimate. For both proofs we will make great use of similar results from Section~\ref{sec:theoryTrend}.

\begin{Setting}\label{setting:anomaly}
We assume that $Y = A + R + W + \varepsilon$, with $A, \varepsilon, X, \beta$ as defined in Section~\ref{sec:model}. We further assume that there are no point any anomalies, i.e.~$O = \emptyset$. The remainder $R = (r_1,\ldots,r_{n})$ is a deterministic polynomial of order $Q$ for which there exists a constant $\CremainderTrend > 0$ such that 
\begin{equation}\label{eq:remainderTrendA}
\sum_{i = 1}^{n} r_i^2 \leq \CremainderTrend \sigma_0^2 \log(n).
\end{equation}
The remainder $W = (w_1,\ldots,w_{n})$ is a deterministic periodic vector with period length $P$ for which there exist a constant $\CremainderSeasonal > 0$ such that 
\begin{equation}\label{eq:remainderSeasonalyA}
\sum_{i = 1}^{n} w_i^2 \leq \CremainderSeasonal \sigma_0^2 \log(n).
\end{equation}
Let $\CvarianceEstimateLower \sigma_0 \leq \hat{\sigma}_0 \leq \CvarianceEstimate \sigma_0$ be deterministic. Suppose that Assumptions~\ref{assumption:gridBandQ},~\ref{assumption:segmentLengthTrend},~and~\ref{assumption:anomalyLengthTrend} hold.
\end{Setting}

\begin{Proposition}\label{proposition:goodFinalAnomalyDetection}
Given Setting~\ref{setting:anomaly} we have that
\begin{equation*}
\Pj\big(\hat{\mathcal{A}} \in \mathcal{B}\big) \to 1, \text{ as } n \to \infty.
\end{equation*}
\end{Proposition}
\begin{proof}
Let $\tilde{T}$ be the estimate that was used to obtain the remainder $R$ and let $\mathcal{T} = \{\tilde{T}\}$. Then, the statement is a direct consequence of Proposition~\ref{proposition:withinSetB}.
\end{proof}

\begin{Proposition}\label{proposition:goodFinalAnomalyEstimate}
Given Setting~\ref{setting:anomaly} there exists a constant $\CfinalAnomalyEstimate > 0$ such that
\begin{equation*}
\Pj\big( (\hat{A} - A)^\top (\hat{A} - A) \leq \CfinalAnomalyEstimate \sigma_0^2 \log(n) \big) \to 1,
\end{equation*}
as $n \to \infty$.
\end{Proposition}
\begin{proof}
It follows from the proof of Propositions~\ref{proposition:withinSetB}~and~\ref{proposition:goodFinalAnomalyDetection} that the events $\hat{\mathcal{A}} \in \mathcal{B}$ and
\begin{equation*}
\max_{1 \leq i \leq j \leq \mn} \varepsilon_{i:j}^\top H_{i:j} \varepsilon_{i:j}
\leq \CnoiseLin \sigma_0^2 \log(n)
\end{equation*}
holds simultaneously with a probability converging to one as $n \to \infty$. In the remaining proof we show that this implies the desired statement which completes the proof.

Since $\mathcal{A} \in \mathcal{B}$, it follows from the definition of $\hat{\mathcal{A}}$ that
\begin{equation}\label{eq:inequalityCostsBoundA}
 C\big(Y, \hat{\mathcal{A}}, \hat{\mu}\big) \leq C\big(Y, \mathcal{A}, \mu\big).
\end{equation}
Also note that $\hat{\mathcal{A}} \in \mathcal{B}$ implies $\hat{K} = K$.

Let $\{\tilde{\tau}_0, \tilde{\tau}_1, \ldots, \tilde{\tau}_{\tilde{K}}, \tilde{\tau}_{\tilde{K} + 1} \} = \{ e_0, b_1,\hat{b}_1, e_1, \hat{e}_1, \ldots, b_K,\hat{b}_K, e_K, \hat{e}_K, b_{K + 1} \}$ with $\tilde{\tau}_0 \leq  \tilde{\tau}_1 \leq \cdots \leq \tilde{\tau}_{\tilde{K}} \leq \tilde{\tau}_{\tilde{K} + 1}$ and $\tilde{K} = 4 K$. Then, we obtain from \eqref{eq:inequalityCostsBoundA} for the first inequality, Lemma~\ref{lemma:mixedtermNoisePolynomial} and the Cauchy-Schwarz inequality for the second inequality, as well as Lemma~\ref{lemma:noiseLin}, \eqref{eq:remainderTrendA} and \eqref{eq:remainderSeasonalyA} that
\begin{align*}
& C\big(Y, \mathcal{A}, \mu\big) - K \lambda_{\operatorname{coll}}\, \hat{\sigma}_0^2 \log(n)^{1 + \delta}
= \sum_{i = 1}^{\mn} (y_i - a_i)^2
= \sum_{i = 1}^{\mn} (\varepsilon_i + r_i + w_i)^2\\
\geq & C\big(Y, \hat{\mathcal{A}}, \hat{\mu}\big) - K \lambda_{\operatorname{coll}}\, \hat{\sigma}_0^2 \log(n)^{1 + \delta}
= \sum_{i = 1}^{\mn} (y_i - \hat{a}_i)^2
= \sum_{i = 1}^{\mn} (a_i - \hat{a}_i + \varepsilon_i + r_i + w_i)^2\\
\geq &  \sum_{i = 1}^{\mn} (\varepsilon_i + r_i + w_i)^2 + \sum_{i = 1}^{\mn} (a_i - \hat{a}_i)^2 + 2\sum_{i = 1}^{\mn} (\varepsilon_i + r_i + w_i)(a_i - \hat{a}_i)\\
\geq &  \sum_{i = 1}^{\mn} (\varepsilon_i + r_i + w_i)^2 + \sum_{i = 1}^{\mn} (a_i - \hat{a}_i)^2\\
& - 2 \left(\sqrt{\sum_{k = 0}^{\tilde{K}} \varepsilon_{\tilde{\tau}_k + 1:\tilde{\tau}_{k + 1}}^\top H_{\tilde{\tau}_k + 1:\tilde{\tau}_{k + 1}} \varepsilon_{\tilde{\tau}_k + 1:\tilde{\tau}_{k + 1}} } + \sqrt{\sum_{i = 1}^{\mn} r_i^2} + \sqrt{\sum_{i = 1}^{\mn} w_i^2} \right) \sqrt{\sum_{i = 1}^{\mn} (a_i - \hat{a}_i)^2}\\
\geq & \sum_{i = 1}^{\mn} (\varepsilon_i + r_i + w_i)^2 + \sum_{i = 1}^{\mn} (a_i - \hat{a}_i)^2\\
& - 2 \left(\sqrt{4K \CnoiseLin \sigma_0^2 \log(n) } + \sqrt{\CremainderTrend \sigma_0^2 \log(n)} + \sqrt{\CremainderSeasonal \sigma_0^2 \log(n)} \right) \sqrt{\sum_{i = 1}^{\mn} (a_i - \hat{a}_i)^2}.
\end{align*}
This and $(a+b+c)^2 \leq 3a^2 + 3b^2 + 3c^2$ for any $a,b,c \in \mathbb{R}$ yield
\begin{align*}
\sum_{i = 1}^{\mn} (a_i - \hat{a}_i)^2 \leq & 4 \left(\sqrt{4K \CnoiseLin \sigma_0^2 \log(n) } + \sqrt{\CremainderTrend \sigma_0^2 \log(n)} + \sqrt{\CremainderSeasonal \sigma_0^2 \log(n)} \right)^2\\
\leq & \CfinalAnomalyEstimate \sigma_0^2 \log(n),
\end{align*}
with $\CfinalAnomalyEstimate \coloneq 48 K \CnoiseLin + 12 \CremainderTrend + 12 \CremainderSeasonal$. This completes the proof.
\end{proof}

\subsection{Proof of Theorem~\ref{theorem:main}}\label{sec:proofMainTheorem}

\begin{proof}[Proof of Theorem~\ref{theorem:main}]
Proposition~\ref{proposition:varianceEstimate} implies that there exists constants $0 < \CvarianceEstimateLower \leq 1$ and $\CvarianceEstimate \geq 1$ such that
\begin{equation}\label{eq:mainGoodVarianceEstimate}
\Pj\big(\CvarianceEstimateLower \sigma_0 \leq \hat{\sigma}_0 \leq \CvarianceEstimate \sigma_0\big) \to 1, \text{ as } n \to \infty.
\end{equation}
Let $\tilde{E}_1$ denote the event in \eqref{eq:mainGoodVarianceEstimate}. It follows from Lemma~\ref{lemma:differencedSequence} that if the event $\tilde{E}_1$ holds, then the differenced series $Y_D$ follows Setting~\ref{setting:trend}\eqref{setting:Trendb}.

Given Setting~\ref{setting:trend}\eqref{setting:Trendb}, Proposition~\ref{proposition:goodFinalTrendEstimate} implies that there exists a constant $\tilde{C}_1 > 0$ such that
\begin{equation}\label{eq:goodInitialTrendEstimate}
\Pj\big( (\hat{T}_0 - T_D)^\top (\hat{T}_0 - T_D) \leq \tilde{C}_1 \sigma_0^2 \log(n) \big) \to 1,
\end{equation}
Let $\tilde{E}_2$ be the event in \eqref{eq:goodInitialTrendEstimate}. Since $\hat{T}_0 - T$ is a polynomial of degree $Q$ the event $\tilde{E}_2$ implies that there exists a constant $\tilde{C}_2 > 0$ such that 
\begin{equation*}
\max_{i = 1,\ldots,n} \vert (\hat{T}_0 - T_D)_i \vert \leq \tilde{C}_2 \sigma_0^2 \log(n).
\end{equation*}
Given that they have the same parameter vector, it also implies that there exists a constant $\tilde{C}_3 > 0$ such that 
\begin{equation*}
\max_{i = 1,\ldots,n} \vert (T - \beta_0 - X \hat{\beta})_i \vert \leq \tilde{C}_3 \sigma_0^2 \log(n)
\end{equation*}
Therefore, given $\tilde{E}_2$ the remainder $Y - X \hat{\beta}$ satisfy Setting~\ref{setting:seasonality}.

Given Setting~\ref{setting:seasonality}, Proposition~\ref{proposition:estSeasonality} implies that there exists a constant $\tilde{C}_4 > 0$ such that
\begin{equation}\label{eq:goodInitialSeasonalEstimate}
\Pj\big( (\hat{S}_0 - S)^\top (\hat{S}_0 - S) \leq \tilde{C}_4 \sigma_0^2 \log(n) \big) \to 1,
\end{equation}
Let $\tilde{E}_3$ be the event in \eqref{eq:goodInitialSeasonalEstimate}. Then, given $\tilde{E}_3$ the remainder $Y - \hat{S}_0$ satisfy Setting~\ref{setting:trend}\eqref{setting:Trenda}.

Given Setting~\ref{setting:trend}\eqref{setting:Trenda}, Proposition~\ref{proposition:goodFinalTrendEstimate} implies that there exists a constant $\CfinalTrendEstimate > 0$ such that 
\begin{equation}\label{eq:goodFinalTrendEstimate}
\Pj\big( (\hat{T} - T)^\top (\hat{T} - T) \leq \CfinalTrendEstimate \sigma_0^2 \log(n) \big) \to 1,
\end{equation}
as $n \to \infty$. Moreover, Proposition~\ref{proposition:withinSetB} ensures that there exists an anomaly estimate $\tilde{\mathcal{A}}$ for which
\begin{equation}\label{eq:goodInitialAnomalyDetection}
\Pj\big(\tilde{\mathcal{A}} \in \mathcal{B}\big) \to 1, \text{ as } n \to \infty.
\end{equation}
Let $\tilde{E}_4$ be union of the events in \eqref{eq:goodFinalTrendEstimate} and \eqref{eq:goodInitialAnomalyDetection}. Since $\hat{T} - T$ is a polynomial of degree $Q$ the event $\tilde{E}_4$ implies that there exists a constant $\tilde{C}_5 > 0$ such that 
\begin{equation*}
\max_{i = 1,\ldots,n} \vert (\hat{T} - T)_i \vert \leq \tilde{C}_5 \sigma_0^2 \log(n).
\end{equation*}
Therefore, given $\tilde{E}_4$ the remainder $Y - \hat{T}$ and $\tilde{\mathcal{A}}$ satisfy Setting~\ref{setting:seasonality}.

Given Setting~\ref{setting:seasonality}, Proposition~\ref{proposition:estSeasonality} implies that there exists a constant $\CfinalSeasonalEstimate > 0$ such that 
\begin{equation}\label{eq:goodFinalSeasonalEstimate}
\Pj\left( (\hat{S} - S)^\top (\hat{S} - S) \leq \CfinalSeasonalEstimate \log(n) \sigma_0 \right) \to 1,
\end{equation}
as $n\to\infty$. Let $\tilde{E}_5$ be the event in \eqref{eq:goodFinalSeasonalEstimate}. Then, given $\tilde{E}_4 \cup \tilde{E}_5$ the remainder $Y - \hat{T} - \hat{S}$ satisfy Setting~\ref{setting:anomaly}.

Given Setting~\ref{setting:anomaly}, Proposition~\ref{proposition:goodFinalAnomalyDetection} implies that \eqref{eq:goodFinalAnomalyDetection} holds. Finally, given Setting~\ref{setting:anomaly} Proposition~\ref{proposition:goodFinalAnomalyEstimate} yields that there exists a constant $\CfinalAnomalyEstimate > 0$ such that
\begin{equation*}
\Pj\big( (\hat{A} - A)^\top (\hat{A} - A) \leq \CfinalAnomalyEstimate \sigma_0^2 \log(n) \big) \to 1,
\end{equation*}
as $n \to \infty$.

Thus, both statements in Theorem~\ref{theorem:main} hold given the events $\tilde{E}_1, \tilde{E}_2, \tilde{E}_3, \tilde{E}_4$, and $\tilde{E}_5$ with $\CfinalEstimate = \max\{\CfinalTrendEstimate, \CfinalSeasonalEstimate, \CfinalAnomalyEstimate\}$. The proof follows from a union bound as this is a finite number of events and all of them hold with probability converging to one as $n \to \infty$.
\end{proof}

\section*{Funding}
This research was conducted whilst Y. Zhang was an EPSRC funded PhD student at Lancaster University (EP/V520214/1, project number 2614268). Eckley gratefully acknowledges the financial support of EPSRC grants EP/T025964/1 (Net0i) and EP/Z531327/1 (DASS).

\bibliographystyle{apalike}
\bibliography{ref}

\end{document}